\documentclass[a4paper,oneside]{book}
\usepackage[utf8]{inputenc}
\usepackage[T1]{fontenc}
\usepackage{graphicx}
\usepackage{longtable}
\usepackage{wrapfig}
\usepackage{rotating}
\usepackage[normalem]{ulem}
\usepackage{amsmath}
\usepackage{amssymb}
\usepackage{capt-of}
\usepackage{hyperref}
\usepackage[utf8]{inputenc}
\usepackage[T1]{fontenc}

\usepackage{marginnote}
\usepackage{sidenotes}
\usepackage{caption}
\pdfoutput=1

\usepackage{layout}
\usepackage[left=1cm,right=7cm,top=2.5cm,bottom=1.5cm]{geometry}
\newcommand{\titlelayout}{
  \newgeometry{left=1in,right=1in,top=1in,bottom=1in}}

\IfFileExists{mathpazo.sty}{\RequirePackage[osf,sc]{mathpazo}}{}
\IfFileExists{helvet.sty}{\RequirePackage[scaled=0.90]{helvet}}{}
\IfFileExists{beramono.sty}{\RequirePackage[scaled=0.85]{beramono}}{}
\RequirePackage[T1]{fontenc}
\RequirePackage{textcomp}

\usepackage{xcolor}

\usepackage[english]{babel}
\usepackage{wrapfig}
\let\proof\relax
\usepackage{amsthm}
\usepackage{hyperref}
\usepackage{mathtools}
\usepackage{amssymb}

\usepackage{amsfonts}

\usepackage{physics}
\usepackage{float}
\usepackage{enumitem}
\setlist{nosep}

\usepackage{algcompatible}
\usepackage{subfig}
\usepackage{bm} 
\usepackage{qcircuit}
\usepackage{adjustbox}
\usepackage{multirow}

\usepackage{tikz, graphics} 
\usetikzlibrary{arrows.meta, 
  bending,     
  patterns,     
  shapes,
  positioning,
  decorations.pathreplacing,
  calc,
  matrix,
}
\tikzset{
  vertex/.style={circle, draw, thick, minimum size=3mm, inner sep=1.5pt},
  edge/.style={thick},
  wire/.style={thick},
  cwire/.style={double, thick}, 
  connector/.style={thick},
  dot/.style={circle, fill, inner sep=1.5pt },
  brace/.style={decorate, decoration={brace, amplitude=6pt}, thick },
  layer/.style={row sep=0.8cm,column sep=1.8cm},
  open/.style={circle,draw,inner sep=1.2pt},
  phase/.style={draw,thick,minimum size=5mm},
  label/.style={font=\Large },
  gate/.style={draw, minimum width=8mm, minimum height=6mm},
  measure/.style={draw, circle, minimum size=6mm},
  added edge/.style={thick,gray},
  subdiv/.style={circle,draw,inner sep=1.2pt},
  subdiv gray/.style={circle,draw,gray,inner sep=1.2pt},
  edge/.style={thick},
  redv/.style={vertex,fill=red!70},
  greenv/.style={vertex,fill=green!60!black},
  whitev/.style={vertex}
}

\usepackage [
n,
advantage,
operators,
sets,
adversary,
landau,
probability, 
notions,
logic,
ff,
mm,
primitives,
events,
complexity,
asymptotics,
keys
] {cryptocode}

\usepackage[strings]{underscore}
\usepackage{bm}
\usepackage{booktabs} 

\newcommand{\cmp}[1]{\complclass{#1}}
\newcommand{\abort}{\mathsf{Abort}}
\newcommand{\accept}{\mathsf{Accept}}

\newcommand{\BQP}{\mathsf{BQP}}

\newcommand{\X}{\mathsf{X}}
\newcommand{\Y}{\mathsf{Y}}
\newcommand{\Z}{\mathsf{Z}}
\newcommand{\Id}{\mathsf{I}}
\newcommand{\Ha}{\mathsf{H}}
\newcommand{\CZ}{\mathsf{CZ}}

\newcommand{\CNOT}{\mathsf{CNOT}}

\newcommand{\T}{\mathsf{T}}

\newcommand{\TP}{C\cup T}
\newcommand{\Deco}{\cptp D_{O, \mathfrak{C}}}
\newcommand{\Cdev}[1]{\tilde{\cptp C}_{T, #1}}
\newcommand{\Redo}{\mathsf{Redo}}
\newcommand{\ok}{\mathsf{Ok}}

\newcommand{\cptp}[1]{\mathsf{#1}} 
\newcommand{\pd}[1]{\mathcal{#1}}  
\newcommand{\sch}[1]{\bm{#1}} 

\newcommand{\Pois}{\mathrm{Poisson}}

\newcommand{\grp}[1]{\mathrm{#1}} 
\usepackage{dsfont}
\newcommand{\zwt}{\mathrm{zwt}}

\newcommand{\one}{\mathds{1}}
\newcommand{\I}{\cptp{I}}

\newcommand{\mpar}[1]{}

\usepackage{todonotes}
\todostyle{holl}{color=green!30, size=\footnotesize}
\todostyle{holli}{color=green!30, inline, size=\footnotesize}
\todostyle{dl}{color=blue!30, size=\footnotesize}
\todostyle{dli}{color=blue!30, inline, size=\footnotesize}
\todostyle{lm}{color=red!30, size=\footnotesize}
\todostyle{lmi}{color=red!30, inline, size=\footnotesize}
\todostyle{tk}{color=yellow!30, size=\footnotesize}
\todostyle{tki}{color=yellow!30, inline, size=\footnotesize}
\usepackage{mdframed}
\usepackage{thmtools}
\usepackage{needspace}
\usepackage[capitalise]{cleveref}

\newmdenv[
leftmargin = 0pt,
innerleftmargin = 1em,
innertopmargin = 0pt,
innerbottommargin = 0pt,
innerrightmargin = 0pt,
rightmargin = 0pt,
linewidth = 0.5pt,
linecolor = crimson,
topline = false,
rightline = false,
bottomline = false
]{leftbar}

\declaretheoremstyle[
notebraces={}{},
headpunct=,
headformat=\NAME{} \NUMBER. \NOTE, 
postheadhook=\leavevmode\begin{leftbar},
  prefoothook=\end{leftbar},
]{leftbarred}

\declaretheoremstyle[
notebraces={}{},
headpunct=,
headformat=\NAME{} \NUMBER. \NOTE,
postheadhook=\leavevmode\begin{leftbar}\itshape,
  prefoothook=\end{leftbar},
]{leftbarredit}

\declaretheorem[name=Protocol, refname={Protocol, Protocols}, style=leftbarred]{protocol}
\declaretheorem[name=Resource, refname={Resource, Resources}, style=leftbarred]{resource}
\declaretheorem[name=Algorithm, refname={Algorithm, Algorithms}, style=leftbarred]{algorithm}

\declaretheorem[name=Definition, refname={Definition, Definitions}, style=leftbarred]{definition}

\declaretheorem[name=Theorem, refname={Theorem, Theorems}, style=leftbarredit]{theorem}
\declaretheorem[name=Lemma, refname={Lemma, Lemmas}, style=leftbarredit]{lemma}

\expandafter\let\expandafter\oldproof\csname\string\proof\endcsname
\let\oldendproof\endproof
\renewenvironment{proof}[1][\proofname]{%
  \begin{leftbar}%
    \oldproof[#1]%
  }{\oldendproof%
  \end{leftbar}}

\definecolor{crimson}{cmyk}{.00, .91, .73, .14}

\usepackage{lettrine}
\usepackage{Rothdn}

\usepackage{appendix}

\makeatletter
\def\cleardoublepage{%
  \clearpage%
  \ifodd\c@page%
  \else%
    \hbox{}%
    \thispagestyle{empty}%
    \newpage%
  \fi}%
\makeatother

\usepackage{titlesec}
\titleformat{name=\part}[display]%
{\Huge\rmfamily\itshape\bfseries}
{Part~\thepart}
{0pt}
{}
[]

\titleformat{name=\chapter}%
{\huge\rmfamily\itshape\bfseries}
{\llap{\colorbox{crimson}{\parbox{1.5cm}{\hfill\color{white}\vphantom{123456789}\thechapter}}\hspace*{5pt}}}
{0pt}
{}
[]

\titleformat{\section}%
{\Large\rmfamily\itshape\bfseries}
{\llap{\colorbox{crimson}{\parbox{1.0cm}{\hfill\color{white}\vphantom{123456789}\thesection}}\hspace*{5pt}}}
{0pt}
{}
[]

\titleformat{\subsection}%
{\large\rmfamily\itshape\bfseries}
{\llap{\colorbox{crimson}{\parbox{0.5cm}{\hfill\color{white}\vphantom{123456789}}}\hspace*{5pt}}}
{0pt}
{}
[]

\titleformat{\subsubsection}%
{\normalsize\rmfamily\itshape\bfseries}
{\llap{\colorbox{crimson}{\parbox{0.25cm}{\hfill\color{white}\vphantom{123456789}}}\hspace*{5pt}}}
{0pt}
{}
[]

\titleformat{\paragraph}[block]%
{\normalsize\rmfamily\itshape\bfseries}
{}
{0pt}
{}
[]

\makeatletter
\renewcommand{\HyOrg@maketitle}{
  \begin {titlepage}
    \titlelayout
    \let \footnotesize \small \let \footnoterule \relax
    \let \footnote \thanks \null \vfil \vskip 60\p@
    \begin{center}
      {\huge \bfseries\itshape\rmfamily\@title \par }
      \vskip 3em
      {\large \bfseries\rmfamily \lineskip .75em\begin{tabular}[t]{c}\@author \end {tabular}\par }
      \vskip 1.5em
      {\large \rmfamily \@date\par }
    \end{center}
    \par \@thanks \vfil \null
  \end{titlepage}
  \setcounter{footnote}{0}
  \global\let\thanks\relax
  \global\let\maketitle\relax
  \global\let\@thanks\@empty
  \global\let\@author\@empty
  \global\let\@date\@empty
  \global\let\@title\@empty
  \global\let\title\relax
  \global\let\author\relax
  \global\let\date\relax
  \global\let\and\relax}
\makeatother

\author{Harold Ollivier}
\date{Spring 2026}
\title{Verification of Quantum Computations: \newline Hardware-Efficient Security Proofs}
\hypersetup{
  pdfauthor={Harold Ollivier},
  pdftitle={Verification of Quantum Computations: \newline Hardware-Efficient Security Proofs},
  pdfkeywords={},
  pdfsubject={},
  pdfcreator={Emacs 30.2 (Org mode 9.7.39)}, 
  pdflang={English}}
\begin{document}

\maketitle

\clearpage
\thispagestyle{empty}
\mbox{}
\cleardoublepage
\thispagestyle{empty}
\vspace*{\fill}
  {\hspace*{\fill} \itshape \large{To David and Ray.}}
\vspace*{\fill}
\vspace*{\fill}
\cleardoublepage

\setcounter{tocdepth}{1}
\tableofcontents

\chapter*{Foreword}
\label{sec:org7f16d57}
It is a well-established, if rarely admitted, truth that the two most compelling sections of an Habilitation à Diriger des Recherches (HDR) are---or should be---this Foreword and the subsequent Acknowledgments. 

There are obvious reasons for this: the totality of the work presented herein has already been made publicly available with far greater detail in peer-reviewed venues. Consequently, there is nothing new from a technical perspective in the chapters that follow. New content is thus strictly confined to these two preliminary sections, where I am permitted to speak without the shield of clinical, scientific anonymity.

Indeed, this manuscript is synthesized from the following works: \cite{LMKO21verifying,KKLM22unifying,KKLM23asymmetric,KLMO24verification,GLMO24composably,KLMO25plugging}. While this list is chronological, it does not reflect the order in which the papers appear in this manuscript. Our work in \cite{LMKO21verifying} was initially an effort to ensure that experimentalists would be able to implement verification with noisy devices; however, it was only after completing the project that we realized it contained many ingredients that were significantly under-exploited. 

Generalizing these elements led to \cite{KKLM22unifying}, which contains the bulk of the new conceptual insights for verification and pushes modularity to its logical limit by providing a security proof compiler---one that can derive security solely from error detection schemes. This is expanded in \cref{hdr:sec:toolbox}, which also comprises a useful protocol extracted from \cite{KKLM23asymmetric} for implementing Remote State Preparation without ever deviating from the equatorial plane of the Bloch sphere. This is a crucial feature that serves many purposes, as will become clear as one progresses through the document. 

Section \cref{hdr:sec:reducing} is devoted to using these tools following the original motivation that prevailed in \cite{LMKO21verifying}: making it as light as possible for an experimentalist to implement secure computation. In this sense, \cite{KLMO24verification,GLMO24composably} justify the title of this manuscript by demonstrating how improved proof techniques provide secure protocols with lighter hardware requirements. Finally, I conclude with prospective protocols in \cref{hdr:sec:extending} exploring multi-party computation \cite{KKLM23asymmetric} and the delegation of fault-tolerant computations \cite{KLMO25plugging}. Both areas represented significant technical challenges, as the existence of protocols achieving security in these scenarios had remained open questions for nearly two decades.

Since the novelty of this exercise lies solely in the presentation and succession of ideas, I have deliberately omitted all formal proofs, referring the interested reader to the original publications. This is a calculated attempt to find the least-damaging trade-off available to French researchers when they are cornered into producing an HDR.

To put it bluntly, this has not been a pleasant exercise. At times, the process of satisfying administrative obligations feels less like a validation of mentoring quality and more like a form of institutionalized hazing. It is a curious paradox: the requirement for an HDR is to safeguard the quality of PhD supervision, yet it achieves the opposite by diverting hundreds of hours away from actual students toward the paraphrasing of existing work.

To be precise, this manuscript represents over 360 hours of labor---roughly 22\% of the official French annual working time. Even when smoothed over the two-year period dedicated to its writing, this 11\% tax on my productivity is staggering. For the sake of comparison, this is twice the time I spent working directly with one of my PhD student over the same period. Did this investment improve my mentoring capabilities? Certainly not. To offer another metric: this effort represents half the yearly energy required to manage the Hybrid Quantum Initiative, a 36M€ R\&D project.

The cumulative impact—time stolen from students and energy distracted from managing large-scale structuring projects—represents a massive loss in missed opportunities and, ultimately, public funding. While the international competition demands that we move fast, train students, and lead ambitious projects, we remain anchored by the archaic requirement to hold an HDR simply to create a team or mentor students.

I was fortunate enough to have either the support to bend these rules or the sheer stubbornness to ignore them. This allowed me to establish a research team of over 20 people and hire four PIs---tasks for which, strictly speaking, I was not yet habilitated. Had I followed the injunctions to the letter, over 50 publications since 2023 would have been credited to other institutions.

However, the situation eventually became untenable. My newly hired PIs also face the hurdle of their own HDRs; without my own degree, their ability to mentor students is restricted, harming their careers, their collaborations, and their ability to secure funding. This, and only this, is the reason for this manuscript. Otherwise, I would have been quite happy to continue my administrative guerrilla warfare indefinitely.

Why include such a rant? Because the HDR is a lingering artifact of the 1984 reform. When the Thèse d’État was suppressed to align France with international standards, the HDR was seemingly promoted as a way to undo what had just been done: to reintroduce a hierarchy that ensured the Caciques with a Thèse d'État would not be confused with the new PhDs. It ensured that they could maintain their positions without facing unseemly competition from younger, and often sharper, PIs.

Today, the HDR resembles military service: before you endure it, you know it is maladapted; while you endure it, you confirm that opinion; but once finished, there is a lingering, cynical desire to see the young ones go through it as well---simply because one does not want them to escape the suffering one had to endure.

For over forty years, this situation has irritated generations of researchers. It is perhaps time to move on and focus on what actually matters---mentoring students, advancing research, and doing science---and replacing a static diplomas with real emphasis in the recruitment process on HR qualities required to manage a team.

\chapter*{Acknowledgments}
\label{sec:org99c01cb}
If two figures in my career must stand out, they are those of two remarkable scientists: Pascale Charpin and Elham Kashefi.

Pascale welcomed me into the CODES team at INRIA, displaying a level of trust that, in hindsight, bordered on the heroic. She agreed to supervise a PhD student she did not know, on a subject---the information-theoretical understanding of decoherence---that was entirely foreign to her own expertise in symmetric cryptography. For the subsequent three years, we did not work together at all. 

I can only imagine the secondary stress of hosting a student who spent his time abroad, working with researchers outside her community, and occasionally picking up fights with the entire administration and the president of his own thesis jury. While I sincerely hope I never have to supervise a student quite like myself, I must admit it was a dream PhD. Thank you, Pascale, for the rarest of academic gifts: total autonomy.

Elham Kashefi was instrumental in dragging me back to the world of research. After thirteen years spent founding, funding, and exiting companies, I had reached the realization that the corporate world was not where I belonged. The intellectual challenge had lost its luster. 

The breaking point arrived in 2014, when I found myself forced to write into a legal contract that \(-500\) is indeed smaller than \(-300\). This was necessary to ensure that neither lawyers nor judges could misinterpret a clause stating: \emph{"In case the EBITDA is below -300 kEur, we shall have the right to convert our bonds."} I thought, it was time to do something else. 

I never believed a return into research would be possible.  Yet, here we are. This return was sparked by a chance encounter, orchestrated by Stéphane Buttigieg---my second recruit at the Institut Louis Bachelier---who recognized my name on a slide Elham had presented at a "Women in Science" conference. Elham took the significant risk of onboarding me into her team, granting me the necessary time to slowly catch up with thirteen years of scientific progress that had passed me by. Once again I was granted trust, autonomy and a wonderful research environment.

I am also deeply indebted to Jean-Frédéric Gerbeau and Bruno Sportisse, who invited me to rejoin INRIA and afforded me the opportunity to build the QAT team. I have been exceptionally lucky to hire incredibly talented colleagues; working with them is a constant pleasure that almost makes the administrative overhead of the French system bearable. I must also thank Luka Music and Domink Leichtle, PhD students when during my time at CNRS, who taught me abstract cryptography, and with whom I enjoy working so much. Thanks also to Anne Broadbent, Ulysse Chabaud and Maxime Garnier for having taken time to read the manuscript and provide feedback and locations of typos.

Finally, I want to thank Laure and Sidoine for their continuous support and infinite patience. They know all too well the manic cycles of research: the euphoria of a new proof, the despair of breaking said proof, the frantic excitement of a new new proof, and the despair\ldots.

Cyclicity, it seems, is the defining characteristic of my career.

\part{Introduction \& Challenges}
\label{sec:org79c4fb7}
\chapter{Introduction}
\label{sec:org7925b9e}
\lettrine[refstring, lines=3, lraise=0.15]{M}{odern science} rests on a simple loop: predict a phenomenon, perform the experiment, compare numbers.  In quantum physics that loop breaks as soon as the device under study becomes computationally rich.  Simulating a 50-qubit circuit, or the many-body dynamics of a 200-site spin chain, exhausts classical super-computers; predicting the result of a generic 100 qubit quantum circuit with depth 100 is entirely out of reach.  Yet such high-complexity regimes are exactly where the quantum computer promises its greatest impact.  How, then, can a classically limited scientist verify claims about a machine that surpasses classical limits?

The tension is already visible in today’s quantum advantage experiments.  Random circuit sampling on a few dozen qubits can be cross-checked—at the cost of consuming every GPU cluster within reach.  Add twenty extra qubits and the classical cross-check evaporates.  By contrast, Shor’s factoring algorithm permits crisp verification: if the device returns the prime factors of a 2048-bit RSA key, multiplication on a laptop confirms the answer.  The difference lies not in the physics of the device but in the complexity class of the problem.  Factoring belongs to \cmp{NP}: a classical certificate---the prime factors---can be checked quickly, even if finding it is hard.  However, a generic quantum computation---those in the Bounded Error Quantum Polynomial (\cmp{BQP}) time class but not in \cmp{NP}---offer no such classical certificate.

Aharonov and Vazirani \cite{AV13is} proposed a way out.  They lifted the predict-verify loop into an interactive dialogue between two agents:
\begin{itemize}
\item A \emph{verifier}, a Bounded Error Polynomial time (\cmp{BPP}) machine, who may additionally prepare or measure single-qubits but cannot simulate the entire computation.
\item A \emph{prover}, an untrusted \cmp{BQP} machine, holding the large quantum computer.
\end{itemize}
Through a sequence of tailored challenges the verifier convinces himself that the prover executed the desired algorithm and refrained from deviations harming the validity of the reported result.  An interactive protocol thus replaces the impossible task of predicting the output of a quantum computation by the feasible task of testing the prover’s responses. Two questions follow naturally:
\begin{enumerate}
\item Which quantum computations admit an interactive proof?  We seek protocols that cover the whole of \cmp{BQP}, not merely special cases like Shor’s factoring algorithm.
\item Can these protocols survive the practical hurdles of possibly noisy quantum hardware used by the verifier and/or the prover?
\end{enumerate}

The first question---is verification even possible in principle?---is now essentially settled. The protocols presented in \cite{FK17unconditionally,ABE10interactive} prove statistical soundness. Several works have improved upon these initial results to cover a wider range of settings. For instance, quantum capabilities on the verifier's side can be removed with the help of two non-communicating provers \cite{RUV13classical,GKW15robustness,MF16post,NV17quantum}, or by resorting to computational hardness assumptions \cite{M18classical}; the verifier can be restricted to performing measurements only \cite{FHM18post,HKSE17direct}; while reference \cite{B18how} provides a circuit-model-based verification protocol where quantum communication is proportional to the number of \(\T\) gates .

The second question---are these protocols practical enough to be of any use?---is more than just rhetorical. Verification is essential for the development of quantum computing as users accessing a remote machine need a guarantee of integrity for their computation, service providers want to know that the machine they buy is really capable of solving classically intractable problems, and even quantum hardware providers need to test whether their design improvements are going in the right direction. Yet, protocols cited above fail in two aspects. They lack robustness and/or they impose high overheads on verifiers and provers. In effect, the lack of robustness amounts to mistake honest noise for a malicious behavior forcing aborts, while the overhead imposes to dedicate precious qubits to security at the expense of those that are available for computing. 

Closing this robustness and overhead gap is the challenge the rest of this document sets out to address.
\section*{Organization and Contributions}
\label{sec:orgd98eb26}
In \cref{hdr:sec:essentials} we review techniques and protocols that provide the necessary background for statistically secure prepare-and-send verification of quantum computations using the approach developed in \cite{FK17unconditionally}. These provide the necessary ingredients for the extensions presented in this document. In \cref{hdr:sec:challenges}, we expose the challenges faced when implementing verification protocols in future machines and hint at the theoretical developments needed to solve them.

In \cref{hdr:sec:toolbox}, we develop our toolbox for addressing the challenges of practical verification. We start by introducing a modular and composable framework that extracts the necessary functionalities to construct verification protocols (\cref{hdr:sec:framework}). We follow it with two direct applications. The first one (\cref{hdr:sec:rvbqc}) is the simplest of the newly introduced protocols but it already pushes the boundaries of what can be done in practice. It removes the space overhead for running a verified computation, therefore suppressing the security vs. computation trade-off mentioned earlier. It also allows the prover to be noisy provided the noise stays below a given threshold expressed at the circuit level---and  unfortunately  not at the gate level. Yet, it shows that some robustness is indeed possible. In the second one (\cref{hdr:sec:dummyless}), a seemingly innocuous improvement is performed---reducing the number of states the verifier needs to prepare---from 10 to 8. More than quantitative, the improvement is qualitative as the 8 states are in the \(XY\) plane of the Bloch sphere. This property alone, will in fact be used consistently in all the subsequent protocols that we introduce.

In \cref{hdr:sec:reducing} we focus on reducing the physical requirements on the verifier's side. In \cref{hdr:sec:rotations}, we show how to verify quantum computations using rotations and bit flips only---i.e. without the need to prepare quantum states---under the assumption that a possibly maliciously controlled source is sending qubits and not higher dimensional systems. In \cref{hdr:sec:wcp} we take a complementary approach where the verifier is allowed to access an off-the-shelf attenuated laser.

In \cref{hdr:sec:extending} we propose to extend the realm of verification in two regimes. The first one is the multi-party case, where it was known how to lift a classical secure multiparty computation to the quantum domain in a mutli-round symmetrically powerful setup. We show (\cref{hdr:sec:smpqc}) that an asymmetric setup with a single powerful server is possible, thereby also simplifying the quantum network topology from being complete to being star-shaped. We then address in \cref{hdr:sec:sdftqc} the important question of long-term scalability of these protocols under noise. It was unknown if statistically secure verification protocols could be implemented for error protected quantum computations. The difficulty is that logically encoded quantum computations greatly enlarge the attack surface for the prover to bypass the requirement for a sound interaction with the verifier. More precisely, error correction is creating side-channels that jeopardize previously known protocol designs and proof techniques. Carefully tweaking fault-tolerant compilation procedures led to closing these side channels and allowed to provide the first statistically secure verification protocol for fault-tolerant quantum computations.

In \cref{hdr:sec:perspectives}, we provide perspective for future work.
\chapter{Essentials}
\label{hdr:sec:essentials}
\lettrine[refstring, lines=3, lraise=0.15]{T}{his chapter} establishes the notation, summarizes measurement-based quantum computation (MBQC), and reviews the two foundational secure delegation schemes that underpin the remainder of this work: Universal Blind Quantum Computation (UBQC) and Verified Blind Quantum Computation (VBQC).
\section{Notation}
\label{sec:orgee35afe}
\begin{itemize}
\item \emph{Sets.}
  \begin{itemize}
  \item For a set \(B \subseteq A\), we denote by \(B^{c}\) the complement of \(B\) in \(A\), where \(A\) will often be the vertex set of a graph and \(B\) a subset of vertices, usually input or output locations.
  \item For \(n\in \mathbb N\), the set of all integers from \(0\) to \(n\) included is denoted \([n]\).
  \item For a set \(A\), \(|A|\) denotes the number of elements in \(A\).
  \item We denote by \(\Theta\) the set of angles \(\qty{\frac{k\pi}{4}}_{k \in \qty{0, \ldots, 7}}\).
  \end{itemize}
\item \emph{Probabilities.}
  \begin{itemize}
  \item For a distribution probability \(\pd D\), \(x\sample \pd D\) indicates that \(x\) is sampled from \(\pd D\). Similarly, for a set \(\Lambda\), \(\lambda \sample \Lambda\) indicates that \(\lambda\) is sampled uniformly at random from \(\Lambda\).
  \item For a real function \(\epsilon(\eta)\), we say that \(\epsilon(\eta)\) is \emph{negligible in $\eta$} if, for all polynomials \(p(\eta)\) and \(\eta\) sufficiently large, we have \(\epsilon(\eta) \leq \frac{1}{p(\eta)}\).
  \item For a real function \(\mu(\eta)\), we say that \(\mu(\eta)\) is \emph{overwhelming in $\eta$} if there exists a negligible \(\epsilon(\eta)\) such that \(\mu(\eta) = 1 - \epsilon(\eta)\).
  \end{itemize}
\item \emph{Unitaries and Operators.}
  \begin{itemize}
  \item We note \(\X, \Y, \Z\) the Pauli operators. Then \(\mathcal{P}_1 = \langle \X, \Y, \Z \rangle\) is the single-qubit Pauli group.
  \item The rotation operator around the \(Z\)-axis of the Bloch sphere by an angle \(\theta\) is noted \(\Z({\theta}) = \begin{pmatrix} 1 & 0 \\ 0 & e^{i\theta} \end{pmatrix}\).
  \item The 2-qubit controlled-phase gate is denoted \(\CZ\) and flips the phase of the \(\ket 1\otimes \ket 1\) computational basis vector, leaving all others unchanged.
  \item Given a set of qubits indexed by elements in set \(V\), for all \(i \in V\) and any single-qubit unitary \(\cptp U\), we denote \(\cptp U_i\) the unitary obtained by applying \(\cptp U\) to qubit \(i\) and identity to the rest of the qubits in \(V\). This is easily extended to multi-qubit gates by using multiple indices.
  \item Using the previous notation the \(n\)-qubit Pauli group is $$\mathcal{P}_n = \{\pm 1, \pm i \} \times \left \{ \cptp P_1 \otimes \ldots \otimes \cptp P_n \mid \cptp P_j \in \{\Id,\X,\Y,\Z\} \right \}.\nonumber$$
  \item When there is no ambiguity, we will overload the notation \(\cptp P_i\) to refer to the \(i\)-th factor in \(\cptp P\) whenever \(\cptp P\) can  be written as a tensor product of single-qubit operators, e.g. for \(\cptp P= \cptp A \otimes \cptp B \otimes \cptp C\), \(\cptp P_2 = B\).
  \item The states in the \(XY\) plane of the Bloch sphere are noted \(\ket{+_{\theta}} = \Z(\theta)\ket{+} = \frac{1}{\sqrt{2}}(\ket{0} + e^{i\theta}\ket{1})\).
  \item For a measurement in the basis \(\ket{\pm_\theta}\), we associated the value \(0\) to outcome \(\ket{+_{\theta}}\) and \(1\) to outcome \(\ket{-_{\theta}}\).
  \item For an \(n\)-qubit operator \(\cptp U\) and \(n\)-qubit mixed state \(\rho\), we write \(\cptp U[\rho]\) for \(\cptp U\rho\cptp U^\dagger\).
  \item For two operators \(\cptp U\) and \(\cptp V\) acting on the same number of qubits, we write \(\cptp U \circ \cptp V\) for the composition of the two operators.
  \end{itemize}
\end{itemize}
\emph{Graph States.}
\begin{itemize}
\item Given a graph \(G = (V,E)\) where \(V\) is the set of vertices and \(E\) the set of edges, \(\ket G\) is the normalized eigenstate of \(\X_i\bigotimes_{i\sim j} \Z_j\) for \(i\in V\) and where \(i\sim j\) whenever \((i,j) \in E\).
\item We denote the set of neighbors of a node \(i\) with \(N_G(i) = \{j \in V, \ i \sim j\}\).
\item \(\ket G\) is also defined as \(\prod_{(i,j)\in E} \CZ_{i,j} \left [\bigotimes_{k\in V} \ketbra +_{k}\right ]\) where \(\CZ_{i,j}\) is the controlled phase gate applied according to the edge \((i,j)\) of \(G\).
\item An open graph state with input set \(I \subset V\) is obtained by applying the same \(CZ\) gates as for regular graph states, the only difference is that a  state \(\ket \psi\) over the qubits in \(I\) is provided instead of \(\bigotimes_{k\in I} \ket +_k\) (\cref{hdr:fig:open-graph-state}).
\end{itemize}
\emph{Complexity}
\begin{itemize}
\item \cmp{BQP}, for Bounded-Error Quantum Polynomial-Time, is the class of decision problems solvable by a uniform family of polynomial-size quantum circuits, with at most 1/3 probability of error.
\item \cmp{BPP}, for Bounded-Error Probabilistic Polynomial-Time, is the class of decision problems solvable by a non-deterministic polynomial time machine such that:  if the answer is 'yes' then at least 2/3 of the computation paths accept; if the answer is 'no' then at most 1/3 of the computation paths accept.
\end{itemize}

\begin{marginfigure}
  \resizebox{6cm}{!}{
\begin{tikzpicture}

\node[vertex,label=above right:{1}] (n1) at (0, 1.2) {};
\node[vertex,label=below right:{2}] (n2) at (0,-1.2) {};
\node[vertex,label=below:{3}]       (n3) at (2, 0)   {};
\node[vertex,label=below:{4}]       (n4) at (4, 0)   {};
\node[vertex,label=above right:{5}] (n5) at (6, 1.2){};
\node[vertex,label=below right:{6}] (n6) at (6,-1.2){};

\draw[edge] (n1)--(n3);
\draw[edge] (n2)--(n3);
\draw[edge] (n3)--(n4);
\draw[edge] (n4)--(n5);
\draw[edge] (n4)--(n6);

\draw[brace] (-0.6,-1.2) -- (-0.6,1.2)
    node[midway,left=8pt] {$\rho$};

\end{tikzpicture}}
  \par\bigskip\bigskip
  \resizebox{6cm}{!}{
\begin{tikzpicture}

\foreach \y/\lab in {0/1,-1/2,-2/3,-3/4,-4/5,-5/6}{
    \draw[wire] (0,\y) -- (6,\y);
    \node[left] at (0,\y) {\lab};
}

\draw[brace] (-0.5,-1) -- (-0.5,0)
    node[midway,left=8pt] {$\rho$};

\draw[connector] (0.9,0) -- (0.9,-2);
\node[dot] at (0.9,0) {};
\node[dot] at (0.9,-2) {};

\draw[connector] (1.8,-1) -- (1.8,-2);
\node[dot] at (1.8,-1) {};
\node[dot] at (1.8,-2) {};

\draw[connector] (2.7,-2) -- (2.7,-3);
\node[dot] at (2.7,-2) {};
\node[dot] at (2.7,-3) {};

\draw[connector] (3.7,-3) -- (3.7,-4);
\node[dot] at (3.7,-3) {};
\node[dot] at (3.7,-4) {};

\draw[connector] (4.7,-3) -- (4.7,-5);
\node[dot] at (4.7,-3) {};
\node[dot] at (4.7,-5) {};

\end{tikzpicture}}
  \caption{\raggedright Open graph state (up) and corresponding circuit model preparation (down).}
  \label{hdr:fig:open-graph-state}
\end{marginfigure}
\section{Measurement Based Quantum Computation (MBQC)}
\label{sec:prelim-mbqc}
MBQC is an alternative model for quantum computing. To perform a computation in MBQC, we need to prepare large entangled state, measure selected qubits, and apply feed-forward Pauli corrections fixed by previous outcomes.

This model emerged from the gate teleportation principle (\cref{hdr:fig:gate-teleportation}). It was introduced in \cite{RB01one}. The measurement calculus \cite{DKP07measurement-calculus} formalized the procedure and the correspondence with the circuit model.

\begin{marginfigure}
  \resizebox{6cm}{!}{\begin{tikzpicture}

\coordinate (in1) at (0,1);
\coordinate (in2) at (0,0);

\node[left] at (in1) {$\ket{\psi}$};
\node[left] at (in2) {$\ket{+}$};

\draw[wire] (in1) -- (1,1);
\draw[wire] (in2) -- (1,0);

\node[dot] (d1) at (1,1) {};
\node[dot] (d2) at (1,0) {};
\draw[wire] (d1) -- (d2);

\node[gate] (z) at (2,1) {$Z(\theta)$};
\draw[wire] (d1) -- (z);

\node[gate] (h) at (3,1) {$H$};
\draw[wire] (z) -- (h);

\node[measure] (m) at (4,1) {$\scriptstyle m$};
\draw[wire] (h) -- (m);

\node[gate] (x) at (4,0) {$X$};
\draw[wire] (4,0) -- (x);

\node[right] at (5,0) {$H Z(\theta)\ket{\psi}$};
\draw[wire] (x) -- (5,0);

\draw[wire] (d2) -- (x.west);

\draw[cwire] (m.south) -- (x.north);

\end{tikzpicture}}
  \caption{\raggedright Gate Teleportation}
  \label{hdr:fig:gate-teleportation}
\end{marginfigure}
\subsection{Computations as Patterns}
\label{sec:org0ce0959}
The core of MBQC is to perform a computation by following these three steps:
\begin{enumerate}
\item Construct a graph state \(\ket G\) associated to the graph \(G=(V,E)\).
\item Successively perform single-qubit measurements on a subset of this state;.
\item Update the subsequent measurements after each new single-qubit measurement is performed.
\end{enumerate}

A computation in the MBQC model is defined by a \emph{measurement pattern}:
\begin{definition}[Measurement Pattern]
  A \emph{measurement pattern} is a tuple \((G,I,O,\{\phi(i)\}_{i\in O^{c}}, f)\) where:
  \begin{itemize}
  \item \(G = (V,E)\) is a graph.
  \item \(I \subset V\) and \(O \subset V\) are input and output vertex sets.
  \item \(\{\phi(i)\}_{i\in O^{c}}\) assigns an angle to non-output vertices.
  \item \(f: O^{c}  \mapsto I^{c}\) is an injective \emph{flow} function inducing a partial order \(\prec\) on the vertices  \(V\).
  \end{itemize}
  \label{hdr:def:pattern}
\end{definition}

The computation is executed by measuring each qubit successively using the partial order defined by \(f\). For qubit \(i\), the measurement basis is \(\left \{\ket{+_{\phi'}}, \ket{-_{\phi'}} \right\}\). The outcome of the measurement at vertex \(i\) is denoted \(b(i)\) and, keeping our convention, \(b(i)=0\) is associated to \(\ket{+_{\phi'}}\) and \(b(i)=1\) to \(\ket{-_{\phi'}}\). The measurement angle \(\phi'(i)\) is dependent on \(\phi(i)\) and also on the outcomes \(b(l)\) of previous measurements. More precisely, to each vertex \(i\) are associated the sets \(S_{\X}(i) = f^{-1}(i)\) and \(S_{\Z}(i) = \{j, i \in N_G(f(j)), \}\) which are respectively called the \(\X\) and \(\Z\) dependency sets for vertex \(i\). A measurement outcome of \(1\) in a qubit from \(S_{\X}(i)\) will multiply the angle of \(i\) by \(-1\), while the \(\Z\) dependencies add \(\pi\) to the angle. Hence the measurement angle for qubit \(i\) is
\begin{align}
  \phi'(i) & = (-1)^{s_{\mathsf{X}} (i)}\phi(i) + \pi s_{\mathsf{Z}}(i), \mbox{ where } \label{hdr:eq:angle-update}\\
  s_{\mathsf{X}}(i) & = \bigoplus_{j\in S_{\mathsf{X}}}b(j), \\
  s_{\mathsf{Z}}(i) & = \bigoplus_{j\in S_{\mathsf{Z}}}b(j). 
\end{align}

The existence of a flow function \(f\) guarantees that patterns can be executed so that the output is independent of the intermediate measurement outcomes. The conditions on \(f\) are that it is an injective function from non-output vertices \(O^c\) to non-input vertices \(I^c\). The reason for this choice of domain and co-domain is that outputs are not measured and therefore do not generate corrections, while inputs are measured first and therefore do generate corrections. Further details regarding the definition of the flows and its generalizations can be found in \cite{DK06determinism,BKMP07generalized,S21relating}.
\subsection{Graph Transformations as Patterns}
\label{hdr:sec:bridge-and-brick}
Using MBQC, it is also possible to describe some manipulations on graph states or open graph states---i.e. graph states with input vertices where qubits are provided externally instead of being prepared in \(\ket +\).

\begin{itemize}
\item \emph{Breaking.} The \emph{break operator} removes a non-input vertex of \(G\) and all incoming edges to this vertex. It works by preparing the vertex in \(\ket 0\).  (\cref{hdr:fig:breaking})
\item \emph{Bridging.} The \emph{bridge operation}  acts on a non-input vertex of degree two. By measuring such vertex in the \(Y\) basis and by applying appropriate corrections on its neighbors, one can produce the open graph state corresponding to \(G'\) obtained from \(G\) by removing the measured vertex and its incoming edges, and instead linking its neighbors directly (\cref{hdr:fig:bridging}).
\end{itemize}

\begin{marginfigure}
  \resizebox{6cm}{!}{
\begin{tikzpicture}[
    scale=1,
    every node/.style={circle, inner sep=1.5pt},
    filled/.style={fill=black},
    hollow/.style={draw, fill=white, thick}
]

\node[hollow] (c) at (0,0) {};

\foreach \angle in {90, 210, 330} {
    \node[filled] (a\angle) at (\angle:1.2) {};
    \draw (c) -- (a\angle);

    \node[filled] (b\angle) at (\angle+30:2.0) {};
    \node[filled] (d\angle) at (\angle-30:2.0) {};

    \draw (a\angle) -- (b\angle);
    \draw (a\angle) -- (d\angle);
}


\end{tikzpicture}}
  \par\bigskip\bigskip
  \resizebox{6cm}{!}{
\begin{tikzpicture}[
    scale=1,
    every node/.style={circle, inner sep=1.5pt},
    filled/.style={fill=black},
    hollow/.style={draw, fill=white, thick}
]

\node[hollow] (c) at (0,0) {};

\foreach \angle in {90, 210, 330} {
    \node[filled] (a\angle) at (\angle:1.2) {};

    \node[filled] (b\angle) at (\angle+30:2.0) {};
    \node[filled] (d\angle) at (\angle-30:2.0) {};

    \draw (a\angle) -- (b\angle);
    \draw (a\angle) -- (d\angle);
}

\end{tikzpicture}}
  \caption{\raggedright Breaking: The initial graph (up) is broken by preparing the empty node in $\ket 0$ state (down).}
  \label{hdr:fig:breaking}
\end{marginfigure}

\begin{marginfigure}
  \resizebox{6cm}{!}{\begin{tikzpicture}[
    scale=1,
    every node/.style={circle, inner sep=1.5pt},
    filled/.style={fill=black},
    hollow/.style={draw, fill=white, thick}
]

\coordinate (C)   at (0,0);        
\coordinate (L1)  at (-1.5,0);   
\coordinate (R1)  at (1.5,0);   

\coordinate (LL)  at (-2.5,-1.5);
\coordinate (LR)  at (-0.5,-1.5);

\coordinate (RL)  at (0.5,-1.5);
\coordinate (RR)  at (2.5,-1.5);

\node[hollow] (c)  at (C)  {};
\node[filled] (l1) at (L1) {};
\node[filled] (r1) at (R1) {};

\node[filled] (ll) at (LL) {};
\node[filled] (lr) at (LR) {};
\node[filled] (rl) at (RL) {};
\node[filled] (rr) at (RR) {};

\draw (c) -- (l1);
\draw (c) -- (r1);

\draw (l1) -- (ll);
\draw (l1) -- (lr);

\draw (r1) -- (rl);
\draw (r1) -- (rr);

\end{tikzpicture}}
  \par\bigskip\bigskip
  \resizebox{6cm}{!}{\begin{tikzpicture}[
    scale=1,
    every node/.style={circle, inner sep=1.5pt},
    filled/.style={fill=black}
]

\coordinate (L1)  at (-1.5,0);
\coordinate (R1)  at (1.5,0);

\coordinate (LL)  at (-2.5,-1.5);
\coordinate (LR)  at (-0.5,-1.5);

\coordinate (RL)  at (0.5,-1.5);
\coordinate (RR)  at (2.5,-1.5);

\node[filled] (l1) at (L1) {};
\node[filled] (r1) at (R1) {};

\node[filled] (ll) at (LL) {};
\node[filled] (lr) at (LR) {};
\node[filled] (rl) at (RL) {};
\node[filled] (rr) at (RR) {};

\draw (l1) -- (ll);
\draw (l1) -- (lr);

\draw (r1) -- (rl);
\draw (r1) -- (rr);

\draw (l1) -- (r1);

\end{tikzpicture}}
  \caption{\raggedright Bridging: The initial graph (up) is bridged by measuring the empty node in $\Y$ (down).}
  \label{hdr:fig:bridging}
\end{marginfigure}
\subsection{Computing on the 3-Qubit Line Graph}
\label{hdr:sec:linegraph}
To make this more concrete, we will now describe an example of an MBQC pattern and corrections on the three-vertex linear graph. In that case we have \(V = \{1, 2, 3\}\) and \(E = \{(1, 2), (2, 3)\}\). The first qubit in the line will be the only one in the input set \(I\) and the last qubit the only one in the output set \(O\). We note \(\phi(1)\) and \(\phi(2)\) the measurement angles of the first two (non-output) qubits. We start with a single-qubit in state \(\ket{\psi}\) as input, the qubits associated to the other two vertices are initialized in the \(\ket{+}\) state. We apply one \(\CZ\) gate for each pair of qubits whose associated vertices are linked by an edge in \(E\). If \(\ket{\psi} = \alpha \ket{0} + \beta\ket{1}\), the resulting state is
\begin{align}
  \CZ_{1, 2}\CZ_{2, 3}\ket{\psi}\ket{+}\ket{+} =
  & \frac{\alpha}{2}\left(\ket{000} + \ket{001} + \ket{010} - \ket{011}\right) \nonumber \\
  & \quad + \frac{\beta}{2}\left(\ket{100} + \ket{101} - \ket{110} + \ket{111}\right).
\end{align}

In order to perform the measurement on the qubit in vertex \(1\), we apply the rotation \(\Z(-\phi(1))\) and project either onto state \(\ket{+}\), associated to the measurement outcome \(0\), or \(\ket{-}\) associated to outcome \(1\). The joint state of the unmeasured qubits---vertices \(2\) and \(3\)---is then
\begin{align}
  \ket{\psi_0} &= \frac{1}{\sqrt{2}}(\alpha + \beta e^{-i\theta})\ket{0}\ket{+} + \frac{1}{\sqrt{2}}(\alpha - \beta e^{-i\theta})\ket{1}\ket{-},\\
  \ket{\psi_1} &= \frac{1}{\sqrt{2}}(\alpha - \beta e^{-i\theta})\ket{0}\ket{+} + \frac{1}{\sqrt{2}}(\alpha + \beta e^{-i\theta})\ket{1}\ket{-}.
\end{align}
There are two things that we can notice from this:
\begin{itemize}
\item In the first case, the state is the same as if we had started with the state \(\ket{\psi'} = \Ha\Z(-\phi(1))\ket{\psi}\) and entangled it to a single \(\ket{+}\) state using a single \(\CZ\) operation---a two-qubit linear graph. This fact allows us to perform the translation between the MBQC model and the circuit model.
\item We see also that, in order to recover \(\ket{\psi_0}\) from the state \(\ket{\psi_1}\), we need to apply an \(\X\) operation on the qubit associated to vertex \(2\) and a \(\Z\) operation on the qubit associated to vertex \(3\). After applying these operations, the state will be independent of the outcome of the measurement. This induces an ordering on the vertices since \(1\) must be measured before \(2\) and \(3\) if we want to use this correction strategy.
\end{itemize}

These corrections can indeed be absorbed into the angle of future measurements since \(\bra{\pm}\Z(\phi)\X = \bra{\pm}\Z(-\phi)\) and \(\bra{\pm}\Z(\phi)\Z = \bra{\pm}\Z(\phi + \pi)\). This can also be done for the qubit in vertex \(2\).  For the output qubits, the corrections need to be applied and cannot be absorbed into a subsequent measurement as, by construction, there are no measurement on output qubits.

In this example above, the flow function is  \(f(1) = 2\) and \(f(2) = 3\), which induces the measurement order \(1 \preceq 2 \preceq 3\). More generally, finding the flow relies on the stabilizers of the graph state associated to \(G\) \cite{BKMP07generalized,MPS22shadow}.
\section{Delegation of Quantum Computations}
\label{sec:org9639605}
A measurement-based quantum computation can be executed remotely. The heavy lifting---preparing a large entangled graph state and performing single-qubit mea\-su\-rements---rests with the server.  The client needs only to: 
\begin{enumerate}
\item Send the classical description of the pattern.
\item Provide its input qubits, if any.
\end{enumerate}

This is summarized in the following protocol, where Alice plays the client and Bob the server:
\begin{protocol}[Delegated MBQC Protocol]
  \begin{algorithmic}[0]
    \STATE \textbf{Alice's Inputs:} A measurement pattern $(G, I, O, \{\phi(i)\}_{i\in O^c}, f)$ and a quantum register containing the input qubits $i\in I$.

    \STATE \textbf{Open Graph State Preparation:}
    \begin{enumerate}
    \item Alice sends the graph's description $(G, I, O)$ to Bob.
    \item Alice sends its input qubits for positions $I$ to Bob.
    \item Bob prepares $\ket +$ states for qubits $i\in I^c$.
    \item Bob applies a $\CZ$ gate between qubits $i$ and $j$ if $(i,j)$ is an edge of $G$.
    \end{enumerate}

    \STATE \textbf{Computation:}
    \begin{enumerate}
    \item Alice sends the measurement angles $\{\phi(i)\}_{i\in O^c}$ along with the description of $f$ to Bob.
    \item Bob measures the qubits $i \in O^c$ in the order $\preceq$ induced by $f$ in the basis $\ket{\pm_{\phi'(i)}}$ where
      \begin{align}
        s_{\X}(i) & = \bigoplus_{j \in S_{\X}(i)} b(j), \ s_{\Z}(i) = \bigoplus_{j \in S_{\Z}(i)} b(j),\\
        \phi'(i) & = (-1)^{s_{\X}(i)}\phi(i) + s_{\Z}(i) \pi, \label{eq:updt-mbqc}
      \end{align}
      where $b(j)\in \{0,1\}$ is the measurement outcome for qubit $j$.
    \item Bob applies the correction $\Z^{s_{\Z}(i)}_i \X^{s_{\X}(i)}_i$ for each output qubits $i \in O$, which it sends back to Alice.
    \end{enumerate}
  \end{algorithmic}
  \label{hdr:proto:dqc}
\end{protocol}
\section{Blind Delegation}
\label{sec:prelim-blindness}
Privacy is the first line of defense in delegated quantum computing.  \emph{Blindness} provides privacy by ensuring that the server learns nothing about either the client’s data or the algorithm beyond a small, explicitly declared leakage.\sidenote{\raggedright \footnotesize It can be set to the size of the graph and the measurement order, if the  delegated computation is compiled to rely a universal graph state. This is the least that can leak as the server will always be communicated the graph and the measurement order as it is precisely what it can always deduce from the instructions given by the client.}  Correctness is not enforced at this stage; an honest-but-curious server returns the right answer, but a malicious one might return garbage without detection.

We formalize blindness as an ideal cryptographic resource, then present the Universal Blind Quantum Computation (UBQC) protocol that constructs the resource, and then illustrate it on a three-qubit example. This formalization follows the principles of Abstract Cryptography \cite{MR11abstract-cryptography} where security stems from the inability for an all-powerful distinguisher having access to all interfaces involved in the ideal resource and the protocol to distinguish between them (See \cref{hdr:sec:ac} for more details).
\subsection{Blindness}
\label{sec:org7f6f53b}
In the language of abstract cryptography, blindness of computation is captured by the following resource, where Alice plays the role of a client who delegates her computation to Bob, playing the server:
\begin{resource}[Blind Delegated Quantum Computation (BDQC)]
  \begin{algorithmic}[0]
    \STATE \textbf{Public Information:} Nature of the leakage $l$.

    \STATE \textbf{Permitted leakage $l$:} A set of computations $\mathfrak{C}$, and two subspaces, one for inputs $\Pi_{I,\mathfrak{C}}$ and one for outputs $\Pi_{O, \mathfrak{C}}$, so that for $\mathsf{C} \in \mathfrak{C}$, an input in $\Pi_{{I, \mathfrak{C}}}$ is mapped to an output in $\Pi_{O, \mathfrak{C}}$.

    \STATE \textbf{Inputs:} 
    \begin{itemize}
    \item Alice inputs the classical description of a computation $\cptp C$ from subspace $\Pi_{I,\mathfrak{C}}$ to subspace $\Pi_{O,\mathfrak{C}}$ and a quantum state $\rho_A$ in $\Pi_{I,\mathfrak{C}}$.
    \item Bob chooses whether or not to deviate. This interface is filtered by a control bit $c$ (set to $0$ by default for honest behavior). If $c = 1$, Bob has an additional input CPTP map $\cptp F$ and state $\rho_B$ possibly entangled with its own working register.
    \end{itemize}

    \STATE \textbf{Computation by the Resource:}
    \begin{enumerate}
    \item If $c = 1$, the Resource sends the leakage $l$ to Bob's interface.
    \item If $c = 0$, it outputs $\cptp{C}[\rho_A]$ at Alice's output interface. Otherwise, it waits for the additional input and outputs $\cptp F[\rho_{AB}]$ at Alice's interface.\sidenotemark
    \end{enumerate}
  \end{algorithmic}
  \label{hdr:res:bdqc}
\end{resource}
\sidenotetext[][-10cm]{\raggedright \footnotesize Here, $\rho_{AB}$ is the joint state of registers $A$ and $B$. It is not possible to write it as $\rho_A \otimes \rho_B$ as both registers could be entangled with one another, and with an external register such as the private register of Bob. This is especially relevant for the security proofs as Alice and Bob are played by a single party, the distinguisher.}

Above the leakage is the graph state and order of measurement that would be used to perform the corresponding MBQC. Note that if one uses a universal graph, then the leakage reduces to an upper bound on the size of the computation. The subspaces \(\Pi_{I,\mathfrak{C}}\) and \(\Pi_{O,\mathfrak{C}}\) are only specifying how the input and output are encoded into qubit systems. 

A direct inspection shows that this definition follows the textual description of what a blind delegated quantum computation is supposed to do:
\begin{itemize}
\item Alice delegates the execution of the computation \(\cptp C\) on \(\rho_{A}\).
\item In case \(c=0\), Bob is honest and Alice receives \(\cptp{C}[\rho_{A}]\).
\item In case \(c=1\), Bob is malicious, Alice receives an arbitrary deviated state, but Bob learns at most \(l\).
\end{itemize}
Because the resource's behavior strictly abides by its specification, it is said to be \emph{secure-by-design}.
\subsection{Protocol for Blind Delegated Quantum Computation}
\label{sec:orgb948b55}
Blind Delegated Quantum Computation (\cref{hdr:res:bdqc}) can be constructed from single-qubit preparations, \(\Z(\theta)\) and \(\X\) unitaries as well as quantum communication. The intuition for the construction is to rely on the client to encrypt its inputs and the ancillary qubits used to perform the computation by applying random \(\Z(\theta)\) rotations and \(\X^{a}\) flips independently for each of them (\cref{hdr:fig:pre-post-rotations}). This will allow the client to one-time pad its inputs and the measurement angles. Commuting the \(\Z(\theta)\) rotations through the \(\CZ\) gates simply induces corresponding rotations right before the measurements, while commuting the \(\X^{a}\) flips induces additional \(\Z^a\) on the  neighbors. Both can be absorbed by adapting the measurement angle accordingly, resulting in seemingly random instructions sent to the server. The measurement outcomes can also be hidden behind a fresh random bit \(r\), and recovered by the client by undoing this last one-time pad when computing the updated angles for the subsequent measurements. All in all, the computation can be carried out correctly whenever the server follows the instructions and without leaking neither inputs, outputs nor measurements angles of the pattern. In essence, only the graph used and the order of measurement is known to the server.

\begin{marginfigure}[-18cm]
  \resizebox{6cm}{!}{
\begin{tikzpicture}

\node[vertex,label=above:{1}] (1) at (-2,1) {};
\node[vertex,label=below:{2}] (2) at (-2,-1) {};
\node[vertex,label=below:{3}] (3) at (-1,0) {};
\node[vertex,label=above:{4}] (4) at (0.5,0) {};
\node[vertex,label=above:{5}] (5) at (1.8,1) {};
\node[vertex,label=below:{6}] (6) at (1.2,-0.8) {};
\node[vertex,label=right:{7}] (7) at (2.4,-0.5) {};
\node[vertex,label=below:{8}] (8) at (1.5,-1.8) {};

\draw[edge] (1)--(3)--(4)--(5);
\draw[edge] (2)--(3);
\draw[edge] (4)--(6)--(7);
\draw[edge] (6)--(8);

\end{tikzpicture}}
  \par\bigskip\bigskip
  \resizebox{6cm}{!}{
\begin{tikzpicture}

\foreach \y in {0,...,7}{
    \draw[edge] (0,-\y) -- (8,-\y);
}

\draw[edge] (1,0) -- (1,-2);
\draw[edge] (2,-1) -- (2,-2);
\draw[edge] (3,-2) -- (3,-3);
\draw[edge] (4,-3) -- (4,-4);
\draw[edge] (5,-3) -- (5,-5);
\draw[edge] (6,-5) -- (6,-6);
\draw[edge] (7,-5) -- (7,-7);

\foreach \p in {
  (1,0),(1,-2),(2,-1),(2,-2),(3,-2),(3,-3),(4,-3),(4,-4),(5,-3),(5,-5),(6,-5),(6,-6),(7,-5),(7,-7)
}{
  \node[dot] at \p {};
}

\end{tikzpicture}}
  \par\bigskip\bigskip
  \resizebox{6cm}{!}{
\begin{tikzpicture}
  
  \definecolor{c0}{RGB}{255,255,255}
\definecolor{c1}{RGB}{220,50,47}
\definecolor{c2}{RGB}{230,120,20}
\definecolor{c3}{RGB}{60,180,75}
\definecolor{c4}{RGB}{50,100,200}
\definecolor{c5}{RGB}{150,50,200}
\definecolor{c6}{RGB}{200,80,180}
\definecolor{c7}{RGB}{50,200,200}
\definecolor{c8}{RGB}{240,200,0}

\foreach \y in {0,...,7}{
    \draw[edge] (0,-\y) -- (8,-\y);
}

\foreach \y/\col in {
0/c1,1/c2,2/c3,3/c4,4/c5,5/c6,6/c7,7/c8}{
    \node[phase,fill=\col] at (0.3,-\y) {};
}

\foreach \y/\col in {
0/c1,1/c2,2/c3,3/c4,4/c5,5/c6,6/c7,7/c8}{
    \node[phase,draw=\col,fill=c0] at (7.7,-\y) {};
}

\node[label,above] at (0.3,0.5) {$\Z(\theta)$};
\node[label,above] at (7.7,0.5) {$\Z(\theta)$};

\draw[edge] (1,0) -- (1,-2);
\draw[edge] (2,-1) -- (2,-2);
\draw[edge] (3,-2) -- (3,-3);
\draw[edge] (4,-3) -- (4,-4);
\draw[edge] (5,-3) -- (5,-5);
\draw[edge] (6,-5) -- (6,-6);
\draw[edge] (7,-5) -- (7,-7);

\foreach \p in {
  (1,0),(1,-2),(2,-1),(2,-2),(3,-2),(3,-3),(4,-3),(4,-4),(5,-3),(5,-5),(6,-5),(6,-6),(7,-5),(7,-7)
}{
  \node[dot] at \p {};
}

\end{tikzpicture}}
  \caption{\raggedright Initial graph state (up), circuit preparing the graph state (middle), equivalent circuit with inserted random $\Z(\theta)$ rotations (colored boxes) compensated by the conjugated rotation $\Z(-\theta)$ before the measurement (down).}
  \label{hdr:fig:pre-post-rotations}
\end{marginfigure}

\begin{protocol}[Universal Blind Quantum Computation (UBQC)]
  \begin{algorithmic}[0]
    \STATE \textbf{Alice's Inputs:} A measurement pattern $(G, I, O, \{\phi(i)\}_{i\in O^c}, f)$ and a quantum register containing the input state $\rho_A$ on qubits $i\in I$.

    \STATE \textbf{Protocol:}
    \begin{enumerate}
    \item Alice sends the graph's description $(G, I, O)$ and the measurement order to Bob.
    \item Alice prepares and sends all the qubits in $V$ to Bob:
      \begin{enumerate}
      \item For $i \in I$, it chooses a random bit $a(i)$. For $i \in I^c$, it sets $a(i) = 0$.
      \item For $i \in O$, it chooses a random bit $r(i)$ and sets $\theta(i) = (r(v) + a_N(v))\pi$ where $a_N(i) = \sum_{j \in N_G(i)} a(j)$. For $i \in O^c$, it samples a random $\theta(i) \in \Theta$.
      \item For $i \in I$, it sends $\prod_{i \in I}\Z_i(\theta(i))\X_i^{a(i)}[\rho_A]$. For $i \in I^c$ it sends $\ket{+_{\theta(i)}}$.
      \end{enumerate}
    \item The Bob applies a $\CZ$ gate between qubits $i$ and $j$ if $(i,j)$ is an edge of $G$.
    \item For all $i \in O^c$, in the order specified by the flow $f$, Alice computes the measurement angle $\delta(i)$ and sends it to Bob, receiving in return the corresponding measurement outcome $b(i)$:
      \begin{align}
        s_{\X}(i) & = \bigoplus_{j \in S_{\X}(i)} b(j) \oplus r(j), \ s_{\Z}(i) = \bigoplus_{j \in S_{\Z}(i)} b(j) \oplus r(j), \\
        \delta(i) & = (-1)^{a(i)}\phi'(i) + \theta(i) + (r(i) + a_N(i)) \pi,\label{eq:updt-ubqc}
      \end{align}
      where $\phi'(i)$ is computed using Equation \ref{eq:updt-mbqc} with the new values of $s_{\X}(i)$ and $s_{\Z}(i)$.
    \item The Bob sends back the output qubits $i \in O$.
    \item Alice applies $\Z_i^{s_{\Z}(i) + r(i)}\X_i^{s_{\X}(i) + a(i)}$ to the received qubits $i \in O$.
    \end{enumerate}
  \end{algorithmic}
  \label{hdr:proto:ubqc}
\end{protocol}

As announced above, the uniform randomness of \(\theta(i)\) perfectly hides the value of \(\phi'(i)\) in \(\delta(i)\), while the value \(r(i)\) hides the measurement outcome but also the output of the computation since it is propagated by the flow and results in a Quantum One-Time-Pad of the output. 

In the original protocol from \cite{BFK09universal}, the outputs are prepared by the server in the \(\ket{+}\) state and are encrypted by the computation flow. The additional randomization of the output qubit might seem superfluous since the server can simply measure the qubit to recover the value of the state. However, including this additional randomization makes a smoother exposition of verification protocols as they can be directly built on top of UBQC (\cref{hdr:proto:ubqc}).

Note that if the output of the client's computation is classical, the set \(O\) is empty and the client only receives measurement outcomes. 

The security properties of UBQC can be summarized as follows:
\begin{theorem}[Security of UBQC]
  UBQC (\cref{hdr:proto:ubqc}) perfectly constructs \cref{hdr:res:bdqc} in the abstract cryptography framework.
\end{theorem}
\subsection{Blind Computing on the 3-Qubit Line Graph}
\label{sec:org9beed48}
We use the example described above to demonstrate an execution of this protocol. Alice would like to hide her input state \(\ket{\psi}\), the measurement angles \(\phi(1)\) and \(\phi(2)\) and the output. For the input, Alice uses a variant of the Quantum One-Time-Pad, sampling a random bit \(a(1) \in \bin\) and a random angle \(\theta(1)\) and applying the operation \(\Z_i(\theta(i))\X_i^{a(i)}\) to \(\ket{\psi}\). For the non-input qubits she sets \(a(2) = a(3) = 0\). For the non-output qubit she samples at random \(\theta(2) \in \Theta\) and for the output vertex she samples at random \(\theta(3) \in \{0, \pi\}\). she creates the states \(\ket{+_{\theta(2)}}\) and \(\ket{+_{\theta(3)}}\) and sends these two states and its encrypted input to Bob.

Bob receives these three qubits and performs the entangling operations \(\CZ_{1, 2}\) and \(\CZ_{2, 3}\) as above. For now the security is guaranteed since the input is perfectly encrypted and the other qubits are in random states uncorrelated to the computation. However, Alice still desires to run her computation and must do so through the encryption. To do so, she will instruct Bob to measure the qubits \(1\) and \(2\) with the angle \(\delta(i) = (-1)^{a(i)}\phi'(i) + \theta(i) + (r(i) + a_N(i)) \pi\), where \(a_N(i)\) is the sum of values of \(a(j)\) for \(j\) neighbors of \(i\) and \(r(i)\) is a random bit.

We can see that this performs the same computation as the MBQC example described above. The \(\Z\) rotation encryption commutes with the \(\CZ\) operations and cancels out the encryption of the measurement angle. On the other hand, when the \(\X\) encryption of the input commutes with the \(\CZ\) gates, it creates an additional \(\Z\) on the neighbor which is then taken care of by \(a_N(2) = a(1)\). Commuting the \(\X\) to the end also flips the sign of the measurement angle, which is why \((-1)^{a(i)}\) appears in from of \(\phi'(i)\) in the expression of \(\delta(i)\). In the end the computation is the same as the unencrypted one above so long as Alice corrects the measurement outcomes returned by Bob to account for the additional \(r(i)\) by flipping the outcome \(b(i)\) if \(r(i) = 1\). If \(\theta(3) = \pi\), Alice must also apply \(\Z\) to the output to compensate. This process is summarized in Figure \ref{fig:ubqc-cor}.

\begin{figure*}[h]%
  \subfloat[]{
    \label{fig:ubqc-cor1}
    \makebox[\textwidth]{
      $
      \Qcircuit @C=1.0em @R=.7em {
        \lstick{\ket{\psi}} & \gate{\X^{a(1)}} & \gate{\Z(\theta(1))} & \ctrl{1}  & \qw       & \gate{\Z( - (-1)^{a(1)}\phi'(1) - \theta(1) + r(1) \pi)} & \qw & \measureD{\pm} \\
        \lstick{\ket{+}}    & \qw              & \gate{\Z(\theta(2))} & \ctrl{-1} & \ctrl{1}  & \gate{\Z( - \phi'(2) - \theta(2) + (r(2) + a(1)) \pi)}   & \qw & \measureD{\pm} \\
        \lstick{\ket{+}}    & \qw              & \gate{\Z(\theta(3))} & \qw       & \ctrl{-1} & \qw                                                      & \qw & \qw
      }
      $
    }
  }\\ \vspace{5mm}
  \subfloat[]{
    \label{fig:ubqc-cor2}
    \makebox[\textwidth]{
      $
      \Qcircuit @C=1.0em @R=.7em {
        \lstick{\ket{\psi}} & \gate{\X^{a(1)}} & \ctrl{1}  & \qw       & \gate{\Z( - (-1)^{a(1)}\phi'(1) + r(1) \pi)} & \qw                  & \measureD{\pm} \\
        \lstick{\ket{+}}    & \qw              & \ctrl{-1} & \ctrl{1}  & \gate{\Z( - \phi'(2) + (r(2) + a(1)) \pi)}   & \qw                  & \measureD{\pm} \\
        \lstick{\ket{+}}    & \qw              & \qw       & \ctrl{-1} & \qw                                          & \gate{\Z(\theta(3))} & \qw
      }
      $
    }
  }\\ \vspace{5mm}
  \subfloat[]{
    \label{fig:ubqc-cor3}
    \makebox[\textwidth]{
      $
      \Qcircuit @C=1.0em @R=.7em {
        \lstick{\ket{\psi}} & \ctrl{1}  & \qw       & \gate{\Z( - \phi'(1) + r(1) \pi)} & \qw                  & \measureD{\pm} \\
        \lstick{\ket{+}}    & \ctrl{-1} & \ctrl{1}  & \gate{\Z( - \phi'(2) + r(2) \pi)} & \qw                  & \measureD{\pm} \\
        \lstick{\ket{+}}    & \qw       & \ctrl{-1} & \qw                               & \gate{\Z(\theta(3))} & \qw
      }
      $
    }
  }%
  \caption{\footnotesize Correctness of UBQC for three-vertex linear graph. (a) UBQC Protocol with the explicit values of $\delta(i)$.
    (b) The $\Z$ encryption commutes through the $\CZ$ gates and is canceled out by the later $\Z$ rotation.
    (c) The input $\X$ encryption commutes through the $\CZ$ gates but adds a $\Z$ on qubit $2$, which is canceled out by the $a(1)\pi$ inside the rotation. The $\X$ on qubit $1$ is commuted through the rotation and absorbed by the measurement. The result is Alice's desired MBQC computation up to Pauli corrections and bit-flips.}
  \label{fig:ubqc-cor}
\end{figure*}
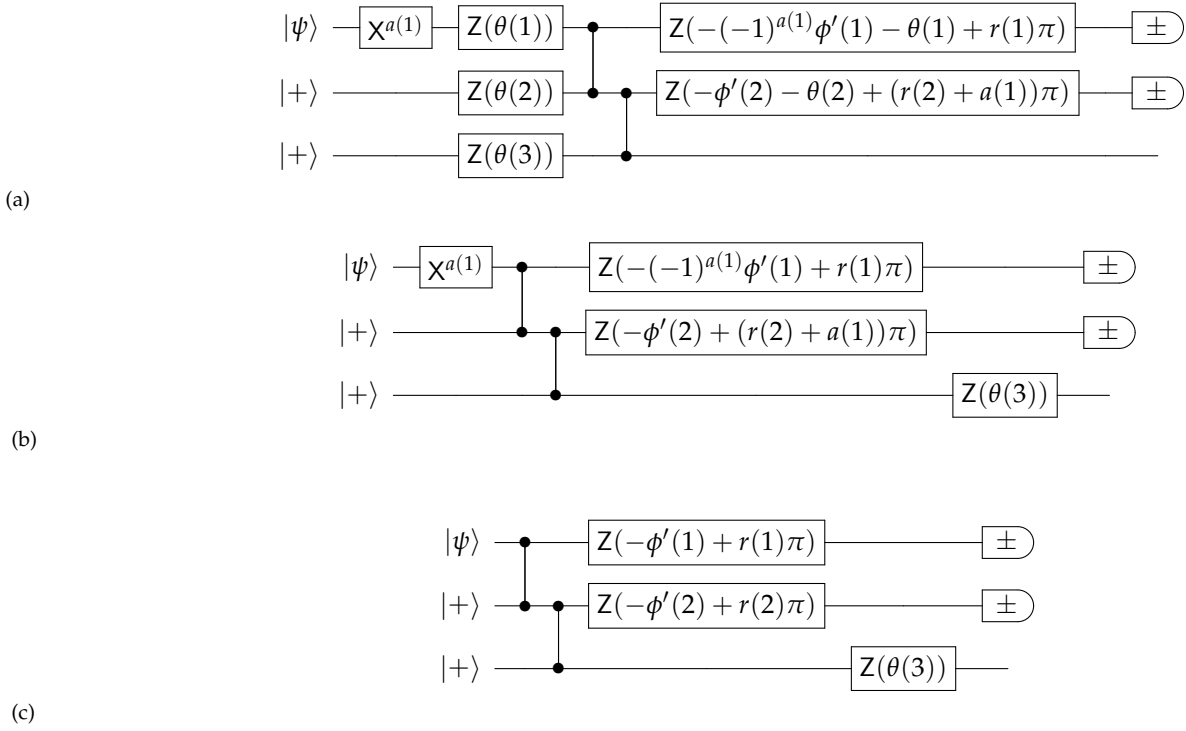
\section{Verified and Blind Delegation}
\label{sec:prelim-vqc}
Integrity complements privacy by allowing the verifier---formerly the client in the blind delegation scenario---to detect deviations in the delegated computation. Verifiability and blindness are the combination of both guarantees: the prover---formerly the sever---learns nothing beyond an explicitly declared leakage and cannot alter the outcome of the delegated computation without being caught. Hence, any malicious behavior either has not effect on the computation or forces the verifier to abort.

We formalize the verifiability and blindness as a resource---the \emph{Secure Delegated Quantum Computation} resource (SDQC)---involving two parties, a verifier and a prover. We then outline the dotted-triple-graph construction of \cite{KW17optimised} that realizes it.  Security intuition and formal theorem conclude the subsection.
\subsection{Verifiability and Blindness}
\label{sec:org2dc1f56}
Although verifiability and blindness are distinct and unrelated properties of delegated quantum computation protocols, they often come together. Indeed, as we will see later, verifiability is obtained with the help of blindness, allowing to interleave the computation of interest with small predictable ones that serve to test the behavior of the prover. The combination of both properties is captured by the following resource, where Alice is playing the role of the verifier and Bob that of the prover:
\begin{resource}[Secure Delegated Quantum Computation (SDQC)]
  \begin{algorithmic}[0]
    \STATE \textbf{Public Information:} Nature of the leakage $l$.

    \STATE \textbf{Permitted leakage $l$:} A set of computations $\mathfrak{C}$, and two subspaces, one for inputs $\Pi_{I,\mathfrak{C}}$ and one for outputs $\Pi_{O, \mathfrak{C}}$.

    \STATE \textbf{Inputs:} 
    \begin{itemize}
    \item Alice inputs the classical description of a computation $\cptp C$ from subspace $\Pi_{I,\mathfrak{C}}$ to subspace $\Pi_{O,\mathfrak{C}}$ and a quantum state $\rho_A$ in $\Pi_{I,\mathfrak{C}}$.
    \item Bob chooses whether or not to deviate. This interface is filtered by two control bits $(e, c)$ (set to $0$ by default for honest behavior).
    \end{itemize}

    \STATE \textbf{Computation by the Resource:}
    \begin{enumerate}
    \item If $e = 1$, the Resource sends the leakage $l$ to Bob's interface; if it receives $c = 1$, the Resource outputs $\dyad{\bot}\otimes\dyad{\abort}$ at Alice's output interface.
    \item Otherwise it outputs $\cptp{C}[\rho_A] \otimes \dyad{\accept}$ at Alice's output interface.
    \end{enumerate}
  \end{algorithmic}
  \label{hdr:res:sdqc}
\end{resource}
The leakage \(l\) is defined as for blindness and consists of the graph that would be used for the computation and the corresponding measurement order. Similarly, the subspaces \(\Pi_{I,\mathfrak{C}}\) and \(\Pi_{O,\mathfrak{C}}\) specify how the input and output information is encoded into the graph state.

As we can readily notice, this resource is secure-by-design:
\begin{itemize}
\item Bob gets at most the leakage \(l\).
\item Whenever it decides to cheat by setting \(e=1\), he can then only force to abort the computation by setting \(c=1\). In case \(e=0\) or \(c=0\), the result of the computation that is transmitted to Alice is correct.
\end{itemize}
\subsection{Protocol for Verified Blind Delegated Quantum Computations}
\label{sec:orgde0f8b8}
We present the dotted-triple-graph VBQC protocol of \cite{KW17optimised}. It is an optimized construction of SDQC (\cref{hdr:res:sdqc}) that shares many features with the original protocol described in \cite{FK17unconditionally}. It has the advantage of yielding streamlined arguments and less complex pictorial representations.

We start by considering that the computation that the verifier wants to perform is protected by an error-correcting code of distance \(d_{\min}\) in a fault-tolerant way, and is described as an MBQC pattern running on the graph state defined by \(G = (V,E)\). There, the subspaces \(\Pi_\mathfrak{I,\mathfrak{C}}\) and \(\Pi_{O,\mathfrak{C}}\) are non-trivial and describe how the verifier's qubits are embedded into the error-correcting code. This being given, we build the \emph{Triple Graph} \(T(G) = (T(V), T(E))\). It is obtained from \(G\) by replacing every vertex \(v \in V\) by a set of three vertices \(P_v = (v_1, v_2, v_3)\), called \emph{primary vertices}, and each edge \((v, w) \in E\) by the nine edges \((v_i, w_j)\) for \(i, j \in \{1, 2, 3\}\). In effect, this step constructs a graph, with three times the initial number of vertices, and that contains \(3^{|E|}\) overlapping copies of \(G\) as subgraphs. Then we build the \emph{Dotted-Triple Graph} \(DT(G)\) by replacing each edge \(e_{i, j} = (v_i, w_j) \in T(E)\) by a vertex \(v_e\) and two edges \((v_i, v_{e_{i, j}})\) and \((v_{e_{i, j}}, w_i)\). These are called the \emph{added vertices}. The previous transformations are summarized in \cref{hdr:fig:dotted-triple-graph}.
\begin{marginfigure}[-12cm]
  \resizebox{6cm}{!}{
\begin{tikzpicture}

\matrix[layer] (m) {
    \node[vertex, scale=0.7] (a1) {}; &
    \node[vertex, scale=0.7] (a2) {}; &
    \node[vertex, scale=0.7] (a3) {}; \\
};

\draw[edge] (a1)--(a2)--(a3);

\end{tikzpicture}

\vspace{1cm}}
  \par\bigskip\bigskip
  \resizebox{6cm}{!}{
\begin{tikzpicture}

\matrix[layer] (m) {
    \node[vertex, scale=0.7] (l1) {}; & \node[vertex, scale=0.7] (c1) {}; & \node[vertex, scale=0.7] (r1) {}; \\
    \node[vertex, scale=0.7] (l2) {}; & \node[vertex, scale=0.7] (c2) {}; & \node[vertex, scale=0.7] (r2) {}; \\
    \node[vertex, scale=0.7] (l3) {}; & \node[vertex, scale=0.7] (c3) {}; & \node[vertex, scale=0.7] (r3) {}; \\
};

\foreach \i in {1,2,3}{
    \foreach \j in {1,2,3}{
        \draw[edge] (l\i) -- (c\j);
    }
}

\foreach \i in {1,2,3}{
    \foreach \j in {1,2,3}{
        \draw[edge] (c\i) -- (r\j);
    }
}

\end{tikzpicture}}
  \par\bigskip\bigskip
  \resizebox{6cm}{!}{
\begin{tikzpicture}

\matrix[layer] (m) {
    \node[vertex, scale=0.7] (l1) {}; & \node[vertex, scale=0.7] (c1) {}; & \node[vertex, scale=0.7] (r1) {}; \\
    \node[vertex, scale=0.7] (l2) {}; & \node[vertex, scale=0.7] (c2) {}; & \node[vertex, scale=0.7] (r2) {}; \\
    \node[vertex, scale=0.7] (l3) {}; & \node[vertex, scale=0.7] (c3) {}; & \node[vertex, scale=0.7] (r3) {}; \\
};

\foreach \i in {1,2,3}{
    \foreach \j in {1,2,3}{
        \node[vertex,scale=0.5] (midL\i\j) at ($(l\i)!0.2!(c\j)$) {};
        \draw[edge] (l\i) -- (midL\i\j) -- (c\j);
    }
}

\foreach \i in {1,2,3}{
    \foreach \j in {1,2,3}{
        \node[vertex,scale=0.5] (midR\i\j) at ($(c\i)!0.2!(r\j)$) {};
        \draw[edge] (c\i) -- (midR\i\j) -- (r\j);
    }
}

\end{tikzpicture}}
  \caption{\raggedright The graph $G$ (up), the corresponding triple graph (middle) and dotted triple graph (down).}
  \label{hdr:fig:dotted-triple-graph}
\end{marginfigure}

To delegate the computation in a verifiable fashion, the verifier will assign three possible labels---computation, dummy and trap---for each vertex of the graph.  More precisely, for each set \(P_v\), there must be one computation, one dummy and one trap vertex assigned at random. The type of the added vertices is chosen so that the two computation vertices are linked with another computation vertex and two dummies are linked by a trap. All other added vertices are dummies. (See \cref{hdr:fig:dtg-labels}).\label{hdr:page:dummy}

\begin{marginfigure}[-4cm]
  \resizebox{6cm}{!}{\begin{tikzpicture}

\definecolor{comp}{RGB}{60,180,75}   
\definecolor{trap}{RGB}{220,50,47}   
\definecolor{dummy}{RGB}{255,255,255}   

\matrix[layer] (m) {
  \node[vertex, scale=0.7, fill=comp ] (l1) {}; &
  \node[vertex, scale=0.7, fill=trap ] (c1) {}; &
  \node[vertex, scale=0.7, fill=dummy] (r1) {}; \\

  \node[vertex, scale=0.7, fill=trap ] (l2) {}; &
  \node[vertex, scale=0.7, fill=comp ] (c2) {}; &
  \node[vertex, scale=0.7, fill=comp ] (r2) {}; \\

  \node[vertex, scale=0.7, fill=dummy] (l3) {}; &
  \node[vertex, scale=0.7, fill=dummy] (c3) {}; &
  \node[vertex, scale=0.7, fill=trap ] (r3) {}; \\
};


        \node[vertex,scale=0.5, fill=dummy] (midL11) at ($(l1)!0.2!(c1)$) {};
        \draw[edge] (l1) -- (midL11) -- (c1);
        \node[vertex,scale=0.5, fill=comp ] (midL12) at ($(l1)!0.2!(c2)$) {};
        \draw[edge] (l1) -- (midL12) -- (c2);
        \node[vertex,scale=0.5, fill=dummy] (midL13) at ($(l1)!0.2!(c3)$) {};
        \draw[edge] (l1) -- (midL13) -- (c3);
        \node[vertex,scale=0.5, fill=dummy] (midL21) at ($(l2)!0.2!(c1)$) {};
        \draw[edge] (l2) -- (midL21) -- (c1);
        \node[vertex,scale=0.5, fill=dummy] (midL22) at ($(l2)!0.2!(c2)$) {};
        \draw[edge] (l2) -- (midL22) -- (c2);
        \node[vertex,scale=0.5, fill=dummy] (midL23) at ($(l2)!0.2!(c3)$) {};
        \draw[edge] (l2) -- (midL23) -- (c3);
        \node[vertex,scale=0.5, fill=dummy] (midL31) at ($(l3)!0.2!(c1)$) {};
        \draw[edge] (l3) -- (midL31) -- (c1);
        \node[vertex,scale=0.5, fill=dummy] (midL32) at ($(l3)!0.2!(c2)$) {};
        \draw[edge] (l3) -- (midL32) -- (c2);
        \node[vertex,scale=0.5, fill=trap ] (midL33) at ($(l3)!0.2!(c3)$) {};
        \draw[edge] (l3) -- (midL33) -- (c3);

        \node[vertex,scale=0.5, fill=dummy] (midR11) at ($(c1)!0.2!(r1)$) {};
        \draw[edge] (c1) -- (midR11) -- (r1);
        \node[vertex,scale=0.5, fill=dummy] (midR12) at ($(c1)!0.2!(r2)$) {};
        \draw[edge] (c1) -- (midR12) -- (r2);
        \node[vertex,scale=0.5, fill=dummy] (midR13) at ($(c1)!0.2!(r3)$) {};
        \draw[edge] (c1) -- (midR13) -- (r3);
        \node[vertex,scale=0.5, fill=dummy] (midR21) at ($(c2)!0.2!(r1)$) {};
        \draw[edge] (c2) -- (midR21) -- (r1);
        \node[vertex,scale=0.5, fill=comp ] (midR22) at ($(c2)!0.2!(r2)$) {};
        \draw[edge] (c2) -- (midR22) -- (r2);
        \node[vertex,scale=0.5, fill=dummy] (midR23) at ($(c2)!0.2!(r3)$) {};
        \draw[edge] (c2) -- (midR23) -- (r3);
        \node[vertex,scale=0.5, fill=trap ] (midR31) at ($(c3)!0.2!(r1)$) {};
        \draw[edge] (c3) -- (midR31) -- (r1);
        \node[vertex,scale=0.5, fill=dummy] (midR32) at ($(c3)!0.2!(r2)$) {};
        \draw[edge] (c3) -- (midR32) -- (r2);
        \node[vertex,scale=0.5, fill=dummy] (midR33) at ($(c3)!0.2!(r3)$) {};
        \draw[edge] (c3) -- (midR33) -- (r3);

\end{tikzpicture}}
  \caption{\raggedright A possible assignment of computation (green), trap (red) and dummy (white) vertices on a subgraph of the dotted triple graph corresponding to two consecutive edges in the original graph.}
  \label{hdr:fig:dtg-labels}
\end{marginfigure}

By performing these steps, we have inserted dummies that can be prepared by the verifier in order to implement break operators leaving trap vertices isolated from the rest of the remaining graph. The added vertices with a computation label can then be used to perform bridge operations and recover \(G\), and finish the computation (See \cref{hdr:fig:broken-dtg}). The trap vertices are then measured in the basis where they have been prepared and their measurement checked. In case there is a discrepancy between the outcome and its expected value the computation is aborted. 

\begin{marginfigure}
  \resizebox{6cm}{!}{\begin{tikzpicture}

\definecolor{comp}{RGB}{60,180,75}   
\definecolor{trap}{RGB}{220,50,47}   
\definecolor{dummy}{RGB}{255,255,255}   

\matrix[layer] (m) {
  \node[vertex, scale=0.7, fill=comp ] (l1) {}; &
  \node[vertex, scale=0.7, fill=trap ] (c1) {}; &
  \node[vertex, scale=0.7, fill=dummy] (r1) {}; \\

  \node[vertex, scale=0.7, fill=trap ] (l2) {}; &
  \node[vertex, scale=0.7, fill=comp ] (c2) {}; &
  \node[vertex, scale=0.7, fill=comp ] (r2) {}; \\

  \node[vertex, scale=0.7, fill=dummy] (l3) {}; &
  \node[vertex, scale=0.7, fill=dummy] (c3) {}; &
  \node[vertex, scale=0.7, fill=trap ] (r3) {}; \\
};


\node[vertex,scale=0.5, fill=dummy] (midL11) at ($(l1)!0.2!(c1)$) {};
\node[vertex,scale=0.5, fill=comp ] (midL12) at ($(l1)!0.2!(c2)$) {};
\draw[edge] (l1) -- (midL12) -- (c2);
\node[vertex,scale=0.5, fill=dummy] (midL13) at ($(l1)!0.2!(c3)$) {};
\node[vertex,scale=0.5, fill=dummy] (midL21) at ($(l2)!0.2!(c1)$) {};
\node[vertex,scale=0.5, fill=dummy] (midL22) at ($(l2)!0.2!(c2)$) {};
\node[vertex,scale=0.5, fill=dummy] (midL23) at ($(l2)!0.2!(c3)$) {};
\node[vertex,scale=0.5, fill=dummy] (midL31) at ($(l3)!0.2!(c1)$) {};
\node[vertex,scale=0.5, fill=dummy] (midL32) at ($(l3)!0.2!(c2)$) {};
\node[vertex,scale=0.5, fill=trap ] (midL33) at ($(l3)!0.2!(c3)$) {};

\node[vertex,scale=0.5, fill=dummy] (midR11) at ($(c1)!0.2!(r1)$) {};
\node[vertex,scale=0.5, fill=dummy] (midR12) at ($(c1)!0.2!(r2)$) {};
\node[vertex,scale=0.5, fill=dummy] (midR13) at ($(c1)!0.2!(r3)$) {};
\node[vertex,scale=0.5, fill=dummy] (midR21) at ($(c2)!0.2!(r1)$) {};
\node[vertex,scale=0.5, fill=comp ] (midR22) at ($(c2)!0.2!(r2)$) {};
\draw[edge] (c2) -- (midR22) -- (r2);
\node[vertex,scale=0.5, fill=dummy] (midR23) at ($(c2)!0.2!(r3)$) {};
\node[vertex,scale=0.5, fill=trap ] (midR31) at ($(c3)!0.2!(r1)$) {};
\node[vertex,scale=0.5, fill=dummy] (midR32) at ($(c3)!0.2!(r2)$) {};
\node[vertex,scale=0.5, fill=dummy] (midR33) at ($(c3)!0.2!(r3)$) {};

\end{tikzpicture}}
  \caption{\raggedright Dotted triple graph after breaking operations at the dummy qubits (up) and after the bridging operations are performed (down). Note that the bridging operations are indeed performed blindly while the computation is carried on so that the prover cannot deduce where the computation and trap qubits are located.}
  \label{hdr:fig:broken-dtg}
\end{marginfigure}

This corresponds to the following protocol:
\begin{protocol}[Verified Blind Quantum Computation (VBQC)]
  \begin{algorithmic}[0]
    \STATE \textbf{Public Information:} A graph $G$ and a fault-tolerant error-correcting scheme with distance $d_{\min}$.

    \STATE \textbf{Alice's Input:} A computation $C$ encoded in a quantum error-correcting code and that can be implemented by an MBQC pattern on a graph $G$.

    \noindent The following steps are performed using UBQC

    \STATE \textbf{Preparation:}
    \begin{itemize}
    \item Input qubits are provided by Alice;
    \item All the non-input qubits corresponding to computation or trap vertices are prepared in $\ket +$;
    \item All qubits corresponding to dummy vertices are prepared in $\ket 0$.
    \end{itemize}

    \STATE \textbf{Pattern execution:} 
    \begin{itemize}
    \item Added qubits are measured: computation ones in $\Y$, so as to perform a bridge operation, traps in $\X$ and dummies with a random angle chosen from $\Theta$,
    \item Primary vertices are measured: computation ones according to the pattern taking into account dependencies from the original pattern and those from the bridge measurements, traps in $\X$, dummies with a random angle chosen from $\Theta$.
    \item For the last layer of primary vertices, Alice performs the measurement of the trap qubits.
    \end{itemize}

    \STATE \textbf{Verification:} If the values of the trap measurements are all $0$ corresponding to obtaining the outcome $\ket +$, Alice accepts and keeps the output qubits. Otherwise she aborts.
  \end{algorithmic}
  \label{hdr:proto:vbqc}
\end{protocol}
Note that, compared to UBQC (\cref{hdr:proto:ubqc}), the preparation now requires not only \(\ket{+_\theta}\) states but also computational basis states \(\ket 0, \ket 1\).
\subsection{Security}
\label{sec:orgfa898cb}
In the protocol above, blindness is inherited from UBQC and preserved by the bridge and break operations: these do not reveal locations of computation, traps or dummy qubits.

Verifiability is then a consequence of the predictability of the trap outcomes when decoded by Alice, as well as ignorance of their location for Bob. As a matter of fact,
\begin{itemize}
\item each qubit is a trap with probability at least \(1/9\) so that the most efficient single-qubit attack has a probability at most \(8/9\) of being undetected.
\item the fault-tolerant encoding with minimum distance \(d_{\min}\) forces Bob to attack at least \(c/2\) positions to change the result of the computation. Hence, the probability of failing to detect a \emph{harmful} deviation---a deviation that would change the result of the unencoded computation---scales as \(\left( \frac{8}{9} \right)^{c/2}\).
\end{itemize}

The AC security of \cref{hdr:proto:vbqc} can be shown using the techniques from \cite{DFPR14composable} and yields:
\begin{theorem}
  \cref{hdr:proto:vbqc} \(\epsilon\)-constructs SDQC (\cref{hdr:res:sdqc}) for arbitrary quantum computations with \(\epsilon\) negligibly small in the distance \(d_{\min}\) of the error-correcting code, while only leaking the graph \(G = (V, E)\) and the order of the measurements used in the MBQC computation.
  \label{thm:full-sdqc}
\end{theorem}
\chapter{Challenges}
\label{hdr:sec:challenges}
\lettrine[refstring,lines=3,lraise=0.15]{T}{he current generation} of verification protocols shows that it is possible for a computationally weak verifier to detect if a  prover is dishonest.  That achievement is intellectually reassuring, at least from a philosophy  of science perspective. Yet the associated techniques remain far too cumbersome for deployment on any foreseeable hardware.  Overheads in qubit count and extreme sensitivity to noise keep these schemes into the realm of \emph{Gedanken} experiments. Practical roll-out demands sharper theoretical tools.

This chapter isolates the obstacles that matter most---not an exhaustive catalog of open questions, but the bottlenecks that must be cleared before verification can migrate from theory pages to laboratory racks and beyond. It also offers perspectives on doing more with the same techniques and resources, making the approach more attractive for hardware vendors to incorporate into their roadmaps. 
\section{Lack of a modular and composable framework}
\label{sec:org7161a33}
Early verification schemes were crafted case-by-case.  Each paper fixed a particular trust model, chose its own security definition, and proofs were monolithic arguments, sometimes redoing what was done already elsewhere, save for some optimizations.  Cross-referencing those proofs---or even lining up their security definitions---quickly becomes an exercise in translation. This ad-hoc landscape has two practical drawbacks:
\begin{itemize}
\item \emph{Protocol design stagnates.} When blindness, trap testing, and error correction are intertwined, tuning any single ingredient forces a return to first principles. A higher-dimensional code, a leaner trap schedule, or a different communication pattern all trigger a full re-proof.  The result is slow iteration and limited cumulative progress.
\item \emph{Composing protocols is generally unsafe.} Quantum accelerators will likely service only a slice of a larger hybrid workflow.  If the quantum slice is verified, then its proofs and indeed all its components need to be composable. This not only ensures the integrity of the quantum computation, but also renders the protocols secure even when used inside hybrid workflows. This ensure that modular engineering can be applied at this larger scale.
\end{itemize}

Without a framework that isolates the various required functionalities, both researchers and practitioners face an explosion of bespoke proofs. \cref{hdr:sec:framework} introduces such a framework. It casts each capability as an abstract cryptography resource with clear interfaces, allowing future optimizations to snap in without reopening the entire security ledger. In addition, it provides a compiler which eliminates the need for the protocol designer to go into security proofs but instead replaces it with checking properties of error detection codes. \cite{DFPR14composable} recognized the need for providing composability but took a slightly  different route than the one we present: its objective was to show how non-composable proofs with \emph{local} criteria can nonetheless yield composable security when additional constraints are added. 
\section{Prover side overhead}
\label{sec:org0c8fbd7}
\subsection{Graph size inflation}
\label{sec:org194cf40}
The crux for turning the UBQC protocol into a verification protocol is to insert traps while maintaining both universality and blindness. This heavily relies on the bridge and break operators that allow to start from a larger graph and reduce it to the desired one. Because both operators are single-qubit transforms that are expressible within the \(XY\) plane, they respect the blindness guarantees of UBQC.

The canonical recipe from \cite{FK17unconditionally} begins with a dotted-complete graph able to host three disjoint copies of the target computation graph \(G\).  One copy executes the algorithm; the other two carry traps, one on the original vertices and one on the dotted vertices. Traps thus probabilistically blanket every physical location (See \cref{hdr:fig:dotted-complete-graph}).

\begin{marginfigure}
  \resizebox{6cm}{!}{\begin{tikzpicture}
\node[vertex, scale=0.5] (v1) at (0.7,1.2) {};   
\node[vertex, scale=0.5] (v2) at (-1,0) {};    
\node[vertex, scale=0.5] (v3) at (1,0) {};     
\node[vertex, scale=0.5] (v4) at (0,-1.2) {};  

\draw[edge] (v1)--(v2);
\draw[edge] (v1)--(v3);
\draw[edge] (v1)--(v4);

\end{tikzpicture}}
  \par\bigskip\bigskip
  \resizebox{6cm}{!}{\begin{tikzpicture}

\node[vertex, scale=0.5] (v1) at (0.7,1.2) {};
\node[vertex, scale=0.5] (v2) at (-1,0) {};
\node[vertex, scale=0.5] (v3) at (1,0) {};
\node[vertex, scale=0.5] (v4) at (0,-1.2) {};

\node[subdiv] (a12) at ($(v1)!0.4!(v2)$) {};
\draw[edge] (v1)--(a12)--(v2);

\node[subdiv] (a13) at ($(v1)!0.4!(v3)$) {};
\draw[edge] (v1)--(a13)--(v3);

\node[subdiv] (a14) at ($(v1)!0.4!(v4)$) {};
\draw[edge] (v1)--(a14)--(v4);

\node[subdiv gray] (b23) at ($(v2)!0.4!(v3)$) {};
\draw[added edge] (v2)--(b23)--(v3);

\node[subdiv gray] (b34) at ($(v3)!0.4!(v4)$) {};
\draw[added edge] (v3)--(b34)--(v4);

\node[subdiv gray] (b24) at ($(v2)!0.4!(v4)$) {};
\draw[added edge] (v2)--(b24)--(v4);

\end{tikzpicture}}
  \caption{Initial graph (up) corresponding dotted-complete graph (down).}
  \label{hdr:fig:dotted-complete-graph}
\end{marginfigure}

The main problem with the dotted-complete construction is that it can introduce up to  \(O(|V|^2)\) additional edges, a quadratic swell that might cancel out any effort to use quantum computers for advantages below quadratic. Several refinements have been proposed to mitigate this overhead :
\begin{itemize}
\item Hybrid protocols that invoke VBQC only on small sub-instances such as those presented in \cite{KDK15optimising}.
\item The dotted-triple graph of \cite{KW17optimised}, which replaces global completeness with structured triples.
\end{itemize}

Unfortunately, they don't fully remove the overhead on the prover's side, yielding a trade-off between the security guarantee and the number of qubits available for computing. The circuit-model variant \cite{B18how} likewise aims to lower the overhead for the prover, but does not provide composability, which hinders its practical usefulness.
\subsection{Security Amplification Demands Heavy Machinery}
\label{sec:org0043671}
To suppress the probability of undetected deviation, VBQC embeds the logical circuit in an error-correcting code and requires three dotted copies of this encoded graph.  The logic is straightforward:
\begin{itemize}
\item A low-weight attack on encoded data is corrected by the code.
\item A high-weight attack must also disturb several traps, whose locations are hidden by blindness and therefore unpredictable.
\item Even if each individual trap detects with constant probability, the probability that harmful deviations are undetected falls  exponentially with the code distance.
\end{itemize}
As a consequence, this overhead grows with the desired security parameter: a larger code distance or a lower targeted security error translates directly into more physical qubits and gates on the prover side. For providers with fixed hardware budgets, the consequence is again a smaller effective computing power.

A smoother path would decouple amplification from graph size. The authors of \cite{KW17optimised} note that classical repetition suffices for amplifying purely classical computations, yet it keeps the triple graph unchanged for the classical case. \cref{hdr:sec:rvbqc} revisits this choice and shows that repetition alone can deliver cryptographic security with no qubit overhead per repetition when compared to the unprotected computation
\section{Verifier side overhead}
\label{sec:orga55bdd4}
From the theorist's perspective, UBQC and VBQC protocols are light\-weight for the verifier. Yet the hardware it presupposes is non-trivial.  The verifier must:
\begin{itemize}
\item Prepare single-qubits---say photons---on demand.
\item Choose one of 10 states \(\{\ket+_{\theta}, \theta \in \Theta\}  \cup \{\ket 0, \ket 1\}\).
\item Deliver those states to the prover with sub-gate timing precision and deterministically.
\end{itemize}
Single-photon sources that approach these specifications exist, but they are bulky, cryogenic, and far from commodity devices.  These requirements therefore block even small-scale proof-of-concepts.

A radical alternative is to remove the verifier’s quantum duties altogether.  Classi\-cal-verifier protocols do exist; they rely on post-quantum cryptography and interactive proofs \cite{M18classical}. But they have their own trade-offs:
\begin{itemize}
\item \emph{Prover overhead.} Preparing and measuring many more qubits per gate before any error correction as a consequence of the encryption of the qubits.
\item \emph{Security model.} Statistical security is replaced by security based on computational hardness assumptions.
\item \emph{Absence of Composability.} Remote State Preparation (RSP)---a crucial ingredient for delegated quantum computation---cannot be securely constructed under computational assumption in a composable framework \cite{BCCK20security}.
\end{itemize}
Hence a gap remains: a protocol with a quantum-minimal verifier yet that is still practical, statistically and composably secure.

\cref{hdr:sec:dummyless,hdr:sec:rotations,hdr:sec:wcp} address this question by either:
\begin{itemize}
\item Reducing the number of states that need to be prepared by the verifier.
\item Removing the need to prepare states, instead relying on states provided by the prover and only applying rotations on it.
\item Replacing single-photon sources with weak-coherent-pulses generators.
\end{itemize}
\section{Over-sensitivity to physical noise}
\label{sec:org1520f12}
Current VBQC schemes deliberately leave traps unencoded: a single-qubit phase-flip on a trap must reveal tampering.  The strategy works against a malicious prover, but it backfires against ordinary device noise.  As gate-level errors occur independently with probability \(p\), a circuit of \(N\) gates therefore accumulates \(O(p N)\) random phase flips---exactly the kind of events that an honest prover cannot avoid.

Because the protocols treat every unexpected trap outcome as dishonesty, the abort probability is overwhelming with circuit size.  Proof-of-concept demonstrations survive, but a computation that stretches into the few-hundred-qubit or few-thousand-gate regime is likely to abort on nearly every run.  The protocol scales in perfect qubit count---increasing the computation size does not demand more repetitions or larger logical encoding for a given targeted security error---, yet not when noise is taken in account. In short, known verification protocols would turn a possibly insecure noisy quantum computing device into a secure noisy quantum \emph{non}-computing device!

In \cref{hdr:sec:sdftqc} we show how to remedy this fragility by delegating a fault-tolerant computation to the prover, thereby protecting traps and data alike while preserving verifiability and blindness.
\section{Extending functionality}
\label{sec:orgff52af2}
All existing blind-and-verified protocols assume a single verifier owns the input register.  Practical workloads, however, often span several data owners who mistrust one another. This could be illustrated by hospitals collaboratively training in a quantum fashion a pattern recognition model over the entire population of a country.  These scenarios call for Secure Multi-Party Quantum Computation (SMPQC), in which:
\begin{itemize}
\item Several clients each supply a disjoint subset of the inputs.
\item A quantum server executes a joint algorithm.
\item Every party learn only its authorized portion of the global output, and does not have access to the proprietary data of the other participants.
\end{itemize}

Yet, the existing MBQC-based SMPQC construction \cite{KP17multiparty} has limited guaran\-tees---blindness but not verifiability---and unusual requirements---absence of client-server collusion---, while such limitations do not exist for circuit-based SMPQC.

\cref{hdr:sec:smpqc} introduces the first composable SMPQC protocol built atop the modular framework of \cref{hdr:sec:framework}.

Other extensions are currently being constructed but fall outside the scope of this document. They deal with benchmarking and noise characterization. Because they rely on the same techniques and broaden the range of applicability of the architecture that security imposes on hardware vendors, they further push them to put the necessary effort to incorporate into their roadmaps.
\part{Modular and Composable Toolbox}
\label{hdr:sec:toolbox}
\cleardoublepage
\lettrine[refstring, lines=3, lraise=0.15]{T}{his part} introduces the toolbox \cite{KKLM22unifying} which played a crucial role in advancing verification techniques. It is a followup to \cite{LMKO21verifying}, building upon the initial modularization capabilities it uncovered. This toolbox generalizes many of the earlier techniques, opening up new possibilities for optimization.

In the following sections, we will begin by deconstructing existing protocols---namely UBQC and VBQC---into their core modules. We will then demonstrate how these modules can be reassembled to form new, more efficient, and robust verification protocols. This will be illustrated through the robust VBQC protocol that laid the foundation of this work, as well as the dummyless protocol, which serves as a cornerstone for all subsequent optimizations presented in \cref{hdr:sec:reducing,hdr:sec:extending}.
\chapter{Framework for Verification Protocols}
\label{hdr:sec:framework}
\lettrine[refstring,lines=3,lraise=0.15]{T}{his chapter,} introduces a \emph{modular} framework\marginnote{The work that motivated the construction of the framework presented here is \cite{LMKO21verifying}. \cite{KKLM22unifying} later distilled the crucial ingredients that were put to work and established the fully modular and composable form we adopt here.} for designing verification protocols. Each module is assigned a single purpose that governs one aspect of the performance or adaptability of the whole protocol. 

Dividing responsibilities in this way brings two immediate pay-offs:
\begin{itemize}
\item \emph{Local reasoning.} Protocol security reduces to module-level properties checked in isolation. These properties are security for the subprotocol in charge of Remote State Preparation; detection, insensitivity, correctness for trappification; overhead for embedding.
\item \emph{Plug-and-play optimization.}  One can refine a single module without reopening the entire security proof.
\end{itemize}

The rest of this chapter proceeds in two steps. First, \cref{hdr:sec:modularization} deconstructs UBQC (\cref{hdr:proto:ubqc}) and VBQC (\cref{hdr:proto:vbqc}), then each extracted module is reassembled into a template protocol (\cref{hdr:proto:tbdqc}), for which only the instantiation of the various modules will be missing. Several instantiations of these modules are described in \cref{hdr:sec:reducing,hdr:sec:extending}, each tackling different challenges.

This instantiation is precisely the role of the next two parts of the document.
\section{Modularization}
\label{hdr:sec:modularization}
The task of the present section is simple to state: we isolate and formalize the three building blocks that recur in every verified and blind delegation scheme, each with a sharply defined role.
\begin{itemize}
\item \emph{Remote State Preparation (RSP).} Enables blindness by allowing the verifier to initialize a single-qubit state inside the prover’s device without disclosing its classical description.
\item \emph{Trappified Scheme.} Encapsulates deviation detection.  The definition comes with quantitative notions of \emph{detection}, \emph{insensitivity}, and \emph{correctness} that convert local error detection behavior into global security guarantees.
\item \emph{Embedding Algorithm.} Compiles an arbitrary target computation into the unspecified region of a graph state while respecting blindness and facilitating deviation detection amplification.
\end{itemize}
\marginnote{\\The underlying idea behind verified-blind computation is relatively simple. It applies the "jealous couple technique" to get confidence on the answer to some question. It hides the true question into many other innocuous questions of the same kind but for which the answer is known. These will be traps. If all traps are passed, then answer to the true question is trusted. If one trap failed\ldots One issue is clearly to ensure the potentially malicious party does not detect the traps by suspecting their nature---meaning trap questions cannot be reused and need to be as difficult as true questions. Another one is to boost the confidence quickly. This is impossible without error correction, which might be why jealous couples don't last long as life quickly turns into constant interrogation. Or maybe is it because of absence of noise robustness in the provided answers.}

By the end of the section the reader will have:
\begin{itemize}
\item Rigorous resource or algorithm definitions for RSP, trappified canvases, the embedding algorithms.
\item A clear view of how the three modules dovetail and which aspect of protocol efficiency or security each module governs.
\end{itemize}

Concrete protocols are deferred to the next section, where the modules are recombined into a template protocol and their local properties are lifted to full composable security.  
\subsection{Extracted Functionalities from VBQC Protocols}
\label{sec:A1-func}
The Verifiable Blind Quantum Computing (VBQC) constructions---from the original one \cite{FK17unconditionally} to the dotted-triple-graph VBQC (\cref{hdr:proto:vbqc})---rest on three logical ingredients, each of which can be isolated and studied on its own.
\begin{itemize}
\item \emph{Blind execution of a computation class.} UBQC (\cref{hdr:proto:ubqc}) allows the verifier to steer the prover through any \cmp{BQP} computation that fits within a given graph state, simply by adjusting single-qubit measurement angles and inserting dummies to break edges.  Crucially, when those angles are restricted to multiples of \(\pi/2\) the unitary describing the computation falls inside the Clifford group and becomes classically simulable. The same interaction pattern between verifier and prover therefore supports both interesting quantum instances and deterministic, classically easy-to-predict ones. Yet, delegated through UBQC, easy instances remain indistinguishable from genuinely quantum ones to the prover.
\item \emph{Local probes of malicious behavior.} A single-qubit disconnected from the graph---or, more generally, a small deterministic computation embedded in the same resource state---can serve as a \emph{trap}.  Any deviation that flips the outcome of the trap signals cheating.  Because not every qubit can be sacrificed for making traps---otherwise no qubit would remain to carry on the computation of the verifier---traps must be sprinkled probabilistically, and their locations kept hidden from the prover.
\item \emph{Amplification of detection probability.} Because traps are chosen probabilistically and because for each trap some deviations are never detected, simply inserting traps is not enough to provide negligible security error. VBQC boosts the detection probability by encoding the verifier's data across many physical qubits---either via a fault-tolerant scheme or by classical repetition---and then distributing traps throughout that enlarged structure.  A deviation, strong enough to alter the logical outcome, must now disturb many sites. These higher-weight \emph{harmful} deviations trigger at least one trap with overwhelming probability.
\end{itemize}

These three functions---blind state preparation and control, trap creation and evaluation, and amplification by fault-tolerance or repetition---will be detailed as the \emph{Remote State Preparation}, \emph{Trappification}, and \emph{Embedding} modules defined formally in the subsections that follow.
\subsection{Blindness: Remote State Preparation}
\label{sec:org8a1512e}
\subsubsection{Motivation}
\label{sec:orgf718da5}
The origin of blindness in UBQC---the functionality that restores the balance of powers between the weak verifier and powerful prover---is simple to state: the prover must hold a qubit whose classical description he does not know.  In the two-prover variant of UBQC \cite{BFK09universal}, this role is played by the second prover, which supplies the qubits on the verifier’s behalf.  Lo’s early work on remote preparation \cite{L00classical} and its cryptographic formalization \cite{DKL11universal} make clear that this essential functionality can be abstracted away from the rest of the protocol. It is commonly captured in a resource called \emph{Remote State Preparation (RSP)}.

The practical requirement is modest---prepare either \(\ket{+_\theta}\) for \(\theta\in\Theta\) or a computational basis state \(\ket{0},\ket{1}\). Yet this single primitive underlies UBQC, VBQC, and the protocols developed later in this document.
\subsubsection{Definition}
\label{sec:org0b22235}
These requirements can be turned into a proper resource within the abstract cryptography framework (\cref{hdr:sec:ac}) where a sender selects a desired state to be prepared at a distant receiver's interface. With Alice being the sender and Bob the receiver, RSP can be defined in the following way:
\begin{resource}[Remote State Preparation]
  \begin{algorithmic}[0]
    \STATE \textbf{Inputs:} Alice inputs a symbol $\lambda \in \Lambda$.

    \STATE \textbf{Computation by the Resource:} The Resource prepares the state $\ket{\lambda}$ at Bob's interface.
    The alphabet \(\Lambda\) is either the equatorial set \(\{+_\theta : \theta\in\Theta\}\) or that set union \(\{0, 1\}\) when computational-basis dummies are required. In the former case, we call the resource \emph{single-plane RSP}.
  \end{algorithmic}
  \label{hdr:res:rsp}
\end{resource}

The definition is deliberately spare. No side-channel is exposed beyond the quantum state itself. The qubit is guaranteed correct \emph{at the moment of delivery}.  What happens afterwards---unitary drift, stochastic noise, measurement, or arbitrary deviation inserted by the prover---must be independent of \(\lambda\). \marginnote{In fact, secure-by-design means that a definition is promoted to a security specification.}This makes it \emph{secure-by-design} as it is simple enough to check the definition fits its textual description.
\subsubsection{Implementation issues}
\label{sec:org0c2fec5}
Implementing such ideal RSP in hardware is subtle.  Imperfect sources or mistimed classical control can imprint a faint \(\lambda\)-dependent signature on auxiliary degrees of freedom---timing, phase noise, or photon number.  Such leakage will break blindness and invalidate composable proofs. Later chapters return to this point:
\begin{itemize}
\item \cref{hdr:sec:rotations} shows that a single trusted \(\Z(\theta)\) rotation, combined with \(\X\) bit flips, suffices for VBQC with minimal overhead for the verifier.
\item \cref{hdr:sec:wcp} presents a composable RSP protocol that tolerates weak coherent pulse sources instead of single photon sources while maintaining the \(\lambda\)-independence demanded by \cref{hdr:res:rsp}.
\end{itemize}
\subsection{Deviation Detection: Trappification}
\label{sec:org72a3840}
\subsubsection{Motivation}
\label{sec:orgf42ba8d}
The original VBQC protocols use the simplest possible probe: a single-qubit, prepared as \(\ket{+}\) and disconnected from the rest of the graph, whose \(\X\)-measurement outcome must be \(0\).  Any non-trivial \(\Z\)-type deviation on that site flips the bit and is caught.  This one-qubit trap is easy to analyze and easy to hide, but it is also extremely narrow: it only tests a tiny corner of the prover’s behavior at a time.

Two further observations drive the present generalization. (i) Not every qubit can be a trap.  Some must carry the actual computation.  Traps therefore have to be sprinkled sparsely and at random, with their locations hidden by blindness; otherwise the prover could simply tip-toe around them. (ii) What truly matters is not the size of the trap but its predictability.  A trap can be any deterministic, classically easy-to-evaluate MBQC pattern run on a subgraph and whose outcome the verifier can compute offline.  Once this is recognized, larger traps are admissible.  A first example of such a multi-qubit trap already appeared in \cite{FKD18reducing}. Here we systematize the idea and use it to shape the detection profile of the overall scheme.

This broader view pays off twice.  It gives far more latitude in how and where traps are inserted, and it lets us tailor their detection profile to the amplification strategy used later on.  The formal objects---partial patterns, trappified canvases, and trappified schemes---introduced in the next subsection make that flexibility precise.
\subsubsection{Definitions}
\label{sec:A4-trap-def}
Our goal is clear: traps must sit beside the computation, yet still detect any deviation that could alter it. If the computation always uses the same physical vertices, an adversary can target those and hide behind legitimate measurement instructions to deviate. We therefore need to randomize where the computation lives and where traps sit. Formally capturing that idea requires three ingredients:
\begin{enumerate}
\item A partially specified MBQC pattern, so we can leave room for either or both trap and computation.
\item A trappified canvas that bundles such a partial pattern together with its input state and an acceptance rule which together allow deviation detection.
\item A trappified scheme: a distribution over many trappified canvases, all indistinguishable to the prover, plus an embedding algorithm that plugs the actual computation into the blank spaces.
\end{enumerate}

We now introduce these notions, keeping the intuition in view.
\begin{definition}[Partial MBQC Pattern]
  Given a graph \(G=(V,E)\), a partial pattern \(P\) on \(G\) is defined by:
  \begin{itemize}
  \item $G_P = (V_P, E_P = E \cap V_P \times V_P)$, a subgraph of $G$;
  \item $I_P$ and $O_P$, the partial input and output vertices, with subspaces $\Pi_{I, P}$ and $\Pi_{O, P}$ defined on vertices $I_P$ and $O_P$ through bases $\mathcal{B}_{I, P}$ and $\mathcal{B}_{O, P}$ respectively;
  \item $\{\phi(i)\}_{i\in V_{P} \setminus O_{P}}$, a set of measurement angles;
  \item $f_p : V_{P} \setminus O_{P} \rightarrow V_{P} \setminus I_{P}$, a flow inducing a partial order $\preceq_P$ on $V_P$.
  \end{itemize}
  \label{hdr:def:partial-pattern}
\end{definition}
Think of a partial pattern as an MBQC pattern where some measurement angles are fixed, others are left blank; some input and output vertices are labeled, others are still free.  This lets us say this contains small sub-computations---a potential trap---without committing the rest of the pattern yet, save for the graph itself and order of measurement.

\begin{definition}[Trappified Canvas]
  A \emph{trappified canvas} \((T, \sigma, \pd{T}, \tau)\) on a graph \(G = (V, E)\) consists of:
  \begin{itemize}
  \item $T$, a partial pattern on a subset of vertices $V_T$ of $G$ with input and output sets $I_T$ and $O_T$;
  \item $\sigma$, a tensor product of single-qubit states on $\Pi_{I, T}$;
  \item $\pd{T}$, an efficiently classically computable probability distribution over binary strings\sidenotemark;
  \item and $\tau$, an efficient classical algorithm that takes as input a sample from $\pd{T}$ and outputs a single bit;
  \end{itemize}
  such that the \(X\)-measurement outcomes of qubits in \(O_T\) are drawn from probability distribution \(\pd{T}\). Let \(t\) be such a sample, the outcome of the trappified canvas is given by \(\tau(t)\). By convention we say that it accepts whenever \(\tau(t)=0\) and aborts for \(\tau(t)=1\).
  \label{hdr:def:trappified-canvas}
\end{definition}
\sidenotetext{\raggedright \footnotesize Here, we really mean computing the probability distribution, not just sampling from it. This is because the decision algorithm $\tau$ might need to access $\mathcal{T}$ while being executed. In practice, $\mathcal{T}$ is so that, given $T$, a single output string has probability 1, the rest being equal to 0, and identification of the only likely outcome is classically easy.}

A trappified canvas is a self-contained trap gadget: it specifies what to prepare, what to measure, how to post-process the outcomes, and how to decide accept or abort, while still having unaffected space available for computing. Compared with the original single-qubit trap, this allows multi-qubit, deterministic sub-computations---as long as the verifier can efficiently predict \(\tau(t)\) in advance.

\begin{definition}[Trappified Pattern]
  Given a computation \(\cptp C \in \mathfrak{C}\) and a trappified canvas \(T\) on graph \(G\) with order \(\preceq_G\), we call the completed pattern \(\TP\) which computes \(\cptp C\) over the unaffected space of the canvas a \emph{trappified pattern}.
  \label{hdr:def:trappified-pattern}
\end{definition}
The canvas leaves white space that is filled with the real computation \(\cptp C\). We will see later, that there are usually many possibilities to embed a computation. The choice of embedding algorithm \(E_{\mathfrak{C}}\) will be discussed later as it governs the detection amplification.

As different runs choose different canvases, the prover must not be able to tell where the logical data ended up. This requires an extra care about blindness. If two canvases differed in their graph, the prover could recognize which canvas was used and, by elimination, where traps sit.  We therefore need:
\begin{definition}[Blind-Compatibility]
  A set of patterns \(\sch P\) is \emph{blind-compatible} if all patterns \(P \in \sch P\) share the same graph \(G\), the same output set \(O\) and there exists a partial ordering \(\preceq_{\sch P}\) of the vertices of \(G\) which is an extension of the partial ordering defined by the flow of any \(P \in \sch P\). This definition can be extended to a set of trappified canvases \(\sch P = \{(T, \sigma, \pd{T}, \tau)\}\). The partial order \(\preceq_{\sch P}\) is required to be an extension of the orderings \(\preceq_T\) of partial patterns \(T\).
  \label{hdr:def:blind-compatibility}
\end{definition}
In short, blind-compatibility ensures that all the trappified patterns share the same graph and the same scheduling for their measurements.  Only the hidden angles and state choices differ, so blindness is preserved.

Finally, we package everything:
\begin{definition}[Trappified Scheme]
  \label{def:trap-scheme}
  A \emph{trappified scheme} \((\sch P, \preceq_G, \pd P, E_{\mathfrak{C}})\) over a graph \(G\) for computation class \(\mathfrak{C}\) consists of:
  \begin{itemize}
  \item $\sch P$, a set of \emph{blind-compatible} trappified canvases over graph $G$ with common partial order $\preceq_{\sch P}$.
  \item $\preceq_G$, a partial ordering of vertices $V$ of $G$ that is compatible with $\preceq_{\sch P}$.
  \item $\pd P$, a probability distribution over the set $\sch P$ which can be sampled efficiently.
  \item $E_{\mathfrak{C}}$, a \emph{proper} embedding algorithm for $\mathfrak{C}$.
  \end{itemize}
\end{definition}

The scheme is the hat from which the verifier draws a canvas at random, then embeds the computation.  The embedding algorithm \(E_{\mathfrak{C}}\) (deferred to \cref{hdr:def:embedding}) ensures that there is a uniform way to map the computation \(\cptp C\) into the unaffected space of the trappified canvas and properness (deferred to \cref{hdr:def:proper-embedding}) prevents information from flowing from the logical computation into the traps, ensuring the preservation of blindness and enabling a composable proof.
\subsubsection{Properties}
\label{sec:A4-trap-prop}
A trappified scheme by itself is only raw material.  To turn it into a verification protocol we must quantify three things (See \cref{hdr:fig:wolf}):
\begin{enumerate}
\item How often do we cry wolf when truly harmful deviation happened (\emph{detection}).
\item How often do we cry wolf when nothing important was touched (\emph{insensitivity}).
\item If we do not cry wolf, how close is the logical output to what it should have been? (\emph{correctness}).
\end{enumerate}

\begin{marginfigure}[-3in]
  \includegraphics[trim={0 0 22cm 0},width=\textwidth]{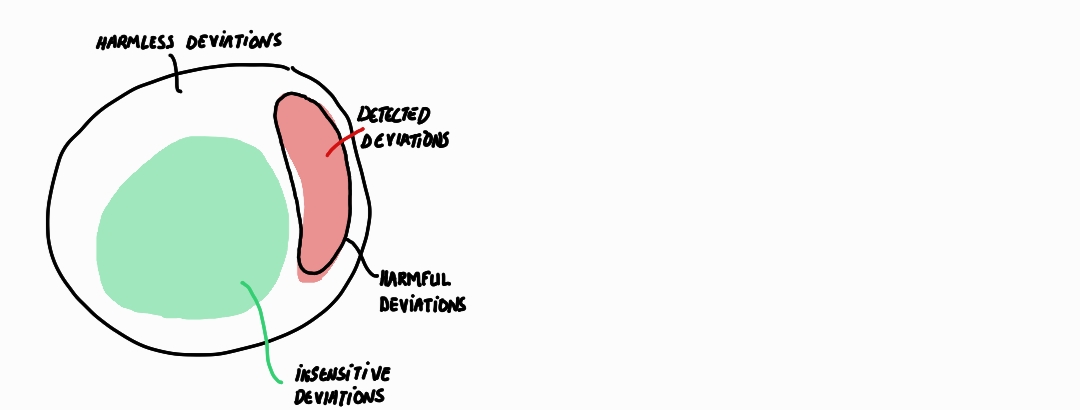}
  \caption{\raggedright Picturing the various types of deviations and their impact on the computation. We want to quantify how often we don't detect deviations that can change the result of the computation \emph{at the logical level}. When this is very small, we will have a verification scheme with good security.}
  \label{hdr:fig:wolf}
\end{marginfigure}

Formally, the random Pauli encryption in UBQC turns an arbitrary channel into a Pauli channel; see, e.g., the standard twirling lemma \cite{DCEL09exact,K16efficient}. As result, any adversarial CPTP map  imposed by the prover reduces to a convex combination of Pauli operators.  From this point on it thus suffices to analyze Pauli deviations. 

The first property we quantify is the Pauli detection capability of the scheme relative to a set of deviations \(\mathcal{E}\). This set is meant to be the set of harmful deviations: \sidenote{\raggedright \footnotesize In \cref{hdr:eq:pauli-detection} the notation $t\sample \pd T_{\cptp E}$ refers to the sampling of $t$ according to the trap measurement when the trap $T$ is affected by deviation $E$ applied after the correct unitary part of the trap applied. The same notation is used in the subsequent equations.}

\begin{definition}[Pauli Detection]
  Let \(T\) be a trappified canvas sampling from distribution \(\pd T\).
  Let \(\mathcal{E}\) be a subset of the Pauli group \(\mathcal P_V\) over the graph vertex qubits.
  For \(\epsilon > 0\), we say that \(T\) \emph{\(\epsilon\)-detects} \(\mathcal{E}\) if: 
  \begin{align}
    \forall \cptp E \in \mathcal{E},\ \Pr_{t\sample \pd T_{\cptp E}}[\tau(t) = 1] \geq 1-\epsilon. \label{hdr:eq:pauli-detection}
  \end{align}

  We say that a trappified scheme \(\sch P\) \emph{\(\epsilon\)-detects} \(\mathcal{E}\) if:
  \begin{align}
    \forall \cptp E \in \mathcal{E}, \ \sum_{T\in \sch P} \Pr_{\substack{T \sample \pd P \\ t\sample \pd T_{\cptp E}}}[\tau(t) = 1, T] \geq 1-\epsilon.
  \end{align}
  \label{hdr:def:detection}
\end{definition}

\marginnote{\\Indeed, we will follow this intuition later in order to obtain global robustness.}
There are also deviations that we do not manage to detect. This is not always detrimental as long as they have little to no impact on the result of the computation. 
\begin{definition}[Pauli Insensitivity]
  Let \(T\) be a trappified canvas sampling from distribution \(\pd T\).
  Let \(\mathcal{E}\) be a subset of \(\mathcal P_V\).
  For \(\delta > 0\), we say that \(T\) is \emph{\(\delta\)-insensitive to} \(\mathcal{E}\) if:   
  \begin{align}
    \forall \cptp E \in \mathcal{E}, \ \Pr_{t\sample \pd T_{\cptp E}}[\tau(t) = 0] \geq 1-\delta.
  \end{align}

  We say that a trappified scheme \(\sch P\) is \emph{\(\delta\)-insensitive to} \(\mathcal{E}\) if:
  \begin{align}
    \forall \cptp E \in \mathcal{E}, \ \sum_{T\in \sch P} \Pr_{\substack{T\sample \pd P \\ t\sample \pd T_{\cptp E}}}[\tau(t) = 0, T] \geq 1-\delta.
  \end{align}
  \label{hdr:def:insensitivity}
\end{definition}

And finally, we want to capture the distance between the ideal computation and what the verifier actually gets when a deviation is sampled from a given set and no trap fires.
\begin{definition}[Pauli Correctness]
  Let \((T, \sigma, \pd T, \tau)\) be a trappified canvas on graph \(G\), \(\preceq_G\) an order on the vertices of \(G\) and \(E_{\mathfrak C}\) an embedding algorithm. Let \(\TP\) be the trappified pattern obtained by embedding a computation \(\cptp C \in \mathfrak{C}\) on \(T\) using \(E_{\mathfrak C}\) and order \(\preceq_G\). Let \(I_C\) be the set of input vertices for the computation \(\cptp C\) in \(\TP\) and let \(\ket{\psi}\) be a state on \(\abs{I_C}+\abs{R}\) qubits, for sufficiently large auxiliary system \(R\), such that \(\Tr_R(\ket{\psi}) \in \Pi_{I,\mathfrak{C}}\), where \(\Pi_{I,\mathfrak{C}}\) is the verifier's input subspace.
  Let \(\mathcal E\) be a subset of  \(\mathcal P_V\). For \(\cptp E \in \mathcal{E}\), we define \(\Cdev{\cptp E} = \Deco\circ\Tr_{O_C^c}\circ\cptp E\circ (\TP)\) to be the CPTP map resulting from applying the trappified pattern \(\TP\) followed by \(\Deco\), the decoding algorithm that maps the output space \(\Pi_{O,\mathfrak{C}}\) to regular qubits.
  For \(\nu \geq 0\), we say that \(T\) is \emph{\(\nu\)-correct} on \(\mathcal E\) if:\sidenotemark
  \begin{align}
    & \forall \cptp E \in \mathcal E, \ \forall \cptp C \in \mathfrak{C}, \nonumber \\ 
    & \qquad \max_\psi\| (\Cdev{\cptp E}  - \cptp C\otimes \Id_T)\otimes\Id_R [\dyad{\psi} \otimes \sigma]\|_{\Tr} \leq \nu. \label{hdr:eq:pauli-cor}
  \end{align}
  This is extended to a trappified scheme \(\sch P\) by requiring the bound to hold on average over \(T \in \sch P\):
  \begin{align}
    & \forall \cptp E \in \mathcal E, \ \forall \cptp C \in \mathfrak{C}, \nonumber \\ 
    & \ \max_\psi \left(\sum_{T\in \sch P} \Pr_{T\sample \pd P}[T] \| (\Cdev{\cptp E}  - \cptp C\otimes \Id_T)\otimes\Id_R [\dyad{\psi} \otimes \sigma]\|_{\Tr}  \right) \leq \nu.
  \end{align}
  \label{hdr:def:correctness}
\end{definition}
\sidenotetext{\raggedright \footnotesize Note that in \cref{hdr:eq:pauli-cor} we are simply expanding the diamond norm between the correct and deviated CPTP maps, but as we need to still fix the input subspace and the input for the trap qubits, this expression avoids defining the deviated CPTP map on the logical result directly.}

Detection and insensitivity are complementary knobs, yet they characterize different sets of deviations. Correctness then ties everything together:
\begin{itemize}
\item If a deviation escapes detection, it must nevertheless leave the logical result almost untouched.
\item Conversely to obtain robustness, traps must be insensitive to deviations that almost don't touch the result.
\end{itemize}
\subsection{Detection Amplification: Embedding}
\label{hdr:sec:embedding-defs}
\subsubsection{Motivation}
\label{sec:org5c0e423}
Traps by themselves only give an inverse–polynomial chance of catching a cheat.  Detection amplification comes from spreading the logical computation across a larger structure---either into a fault–tolerant error-correcting scheme code or into several independent repetitions---and sprinkling traps through that structure.  The embedding algorithm is the classical piece of machinery that performs this spread:
\begin{itemize}
\item It maps the target computation into the unused area of the chosen canvas.
\item It fixes how logical outputs are decoded (majority vote, code decoder,\ldots).
\item It ensures that small, local Pauli errors are \emph{harmless} (corrected or outvoted), while large, \emph{harmful} ones are likely to hit a trap.
\end{itemize}

Abstract Cryptography forces one more constraint: the algorithm must be public and must not let information flow from the logical computation back into the traps.  Otherwise blindness, and hence the simulator in the security proof, would fail as the choice of computation could influence the trap detection probability for a given deviation.  This is the content of \emph{properness} below. This is the analog of the independent local-verifiability criterion used in \cite{DFPR14composable}. Here, the property appears very explicitly as a necessary ingredient used to construct the simulator in the security proof.
\subsubsection{Definitions}
\label{sec:orgaf542ca}
\begin{definition}[Embedding Algorithm]
  Let \(\mathfrak{C}\) be a class of quantum computations. An \emph{embedding algorithm} \(E_{\mathfrak{C}}\) for \(\mathfrak{C}\) is an efficient classical probabilistic algorithm that takes as input:

  \begin{itemize}
  \item $\cptp{C} \in \mathfrak{C}$, the computation to be embedded.
  \item $G = (V, E)$, a graph, and an output set $O$.
  \item $T$, a trappified canvas on graph $G$.
  \item $\preceq_G$, a partial order on $V$ which is compatible with the partial order defined by $T$.
  \end{itemize}
  and outputs:
  \begin{itemize}
  \item A partial pattern $C$ on $V \setminus V_T$, with
    \begin{itemize}
    \item Input and output vertices $I_C \subset V \setminus V_T$ and $O_C = O \setminus O_T$.
    \item Two subspaces (resp.) $\Pi_{I,\mathfrak{C}}$ and $\Pi_{O,\mathfrak{C}}$ of (resp.) $I_\mathfrak{C}$ and $O_\mathfrak{C}$ with bases (resp.) $\mathcal{B}_{I,\mathfrak{C}}$ and $\mathcal{B}_{O,\mathfrak{C}}$.
    \end{itemize}
  \item A decoding algorithm \(\Deco\).
  \end{itemize}
  such that the flow \(f_C\) of partial pattern \(C\) induces a partial order which is compatible with \(\preceq_G\). If \(E_{\mathfrak C}\) is incapable of performing the embedding, it outputs \(\bot\).\sidenotemark
  \label{hdr:def:embedding}
\end{definition}
\sidenotetext{\raggedright \footnotesize This is required to ensure that the behavior of the embedding algorithm is well defined for all possible input computation $\mathsf{C}$. This is important to cover the case where the distinguisher could try to input very large computations that cannot fit into the available blank space in the trappified canvas.}
\noindent In this definition, \(E_{\mathfrak C}\) fills the blanks of the canvas with the computation pattern and tells the verifier how to read the answer back.  Different choices of
\(E_{\mathfrak C}\) yield different amplification behaviors.

If the embedding is such that it stays in the logical space of a quantum error-correcting code, deviations with a weight smaller than half the minimum distance of the code will be harmless. The set of deviations that would need to be detected are those that are not low weight under this criterion, and that would typically be done with a very high probability, even on average over the choice of trappified canvas. Similarly, whenever the computation is classically repeated and the result obtained by majority vote, we are computing in a classically protected subspace. This will give the correct result as long as there are less than half the rounds that are affected by a deviation. There again we will be able to isolate the low-weight from high-weight deviations, the latter being detected with very high probability.

With the embedding comes the subtle notion of \emph{proper embedding} which captures the absence of influence between the computation and the traps. 
\begin{definition}[Proper Embedding]
  We say that an embedding algorithm \(E_{\mathfrak{C}}\) is \emph{proper} if, for any computation \(\cptp{C} \in \mathfrak{C}\) and trappified canvas \(T\) that do not result in a \(\bot\) output, we have that:
  \begin{itemize}
  \item The flow $f$ does not induce dependencies on vertices $V_T$ of partial pattern $T$, in the sense of \cref{hdr:eq:angle-update}.
  \item The input and output subspaces $\Pi_{I,\mathfrak{C}}$ and $\Pi_{O,\mathfrak{C}}$ do not depend on the trappified canvas $T$ save for a relabeling of the qubits.
  \end{itemize}
  \label{hdr:def:proper-embedding}
\end{definition}
In essence, a proper embedding guarantees that the simulator in the security proof can sample a trappified canvas and any computation from the class \(\mathfrak{C}\) and play the role of the verifier with the distinguisher playing that of the prover. The trap would be computed and yield the same transcript even if the actual computations would differ. Without the above property, this could not be guaranteed and the distinguisher could use it to its advantage.
\section{Reconstruction}
\label{hdr:sec:reconstruction}
The previous section carved verification into three modules---RSP for blindness, Trappification for deviation detection, Embedding for amplification. Here we put the pieces back together. The goal is twofold: (i) give a concrete, VBQC-style protocol template that calls each module; (ii) show how local guarantees---detection, insensitivity, correctness---combined with the chosen amplification strategy lift to a global, composable security bound. This reconstruction also makes plain how each module’s parameters feed directly into the overall performance.
\subsection{Template protocol}
\label{sec:orgc025086}
An informal template for VBQC-style protocols using the notions defined earlier can be described succinctly, with Alice being the verifier and Bob the prover:
\begin{protocol}[Trappified Blind Delegated Quantum Computation, Informal]
  \begin{enumerate}
  \item Alice samples a trappified canvas from the trappified scheme and embeds its computation in the canvas, yielding a trappified pattern.
  \item Alice blindly delegates this trappified pattern to Bob via UBQC, and collects the output qubits and the measurement outcomes.
  \item Alice evaluates the decision function, and accepts or aborts.
  \item If she didn't abort, Alice performs some simple classical or quantum single-qubit post-processing on the output.
  \end{enumerate}
  \label{hdr:proto:tbdqc}
\end{protocol}

If we restrict to \cmp{BQP} computations, both inputs and outputs are classical and the amplification can be performed through repetitions and majority voting. This way, we arrive at a more concrete template for verified \cmp{BQP} computations, where only traps and RSP implementations are left to specify:
\begin{protocol}[Trappified Blind Delegated Quantum Computation]
  \begin{algorithmic}[0]
    \STATE \textbf{Public Information:} 
    \begin{itemize}
    \item $G = (V, E, I, O)$, a graph with input and output vertices $I$ and $O$ respectively.
    \item $\sch P$, a trappified scheme on graph $G$.
    \item $\preceq_G$, a partial order on the set $V$ of vertices.
    \item $N, d, w$, parameters representing the number of runs, the number of computation runs, and the number of tolerated failed tests.
    \end{itemize}
    
    \STATE \textbf{Alices's Inputs:} A set of angles $\{\phi_i\}_{i \in V}$ and a flow $f$ which induces an ordering compatible with $\preceq_G$.
    
    \STATE \textbf{Protocol:}
    \begin{enumerate}
    \item Alice samples uniformly at random a subset $C \subset [N]$ of size $d$ representing the runs which will be its desired computation, henceforth called computation runs.
    \item For $k \in [N]$, Alice and Bob perform the following:
      \begin{enumerate}
      \item If $k \in C$, Alice sets the computation for the run to its desired computation $(\{\phi_i\}_{i \in V}, f)$. Otherwise, Alice samples a test $(T, \sigma, \tau)$ from the trappified scheme $\sch P$.
      \item Alice and Bob blindly execute the run using the UBQC.
      \item If it is a test, Alice uses $\tau$ on the measurement results to decide whether the test passed or not.
      \end{enumerate}
    \item At the end of all runs, let $x$ be the number of failed tests. If $x \geq w$, Alice aborts and outputs $(\bot, \abort)$.
    \item Otherwise, Alice accepts the computation. She then performs a majority vote on the output results of the computation runs and sets the result as its output.
    \end{enumerate}  
  \end{algorithmic}
  \label{hdr:proto:tbdqc-classical-io}
\end{protocol}
\noindent 
\subsection{Key Theorems}
\label{sec:A4-main-thms}
The following theorems can be applied to protocols following the informal template and are given here only with a brief sketch of proof. Their proofs nonetheless follow an intuitive route:
\begin{itemize}
\item UBQC turns any deviation into a mixture of Paulis.
\item Split that set into detected, ignored but harmless, and the rest.
\item Bound each contribution to the distinguishing advantage.
\end{itemize}

\begin{theorem}[Detection Implies Verifiability, Informal]
  Let \(\mathcal E_\epsilon, \mathcal E_\nu \subset \mathcal{P}_V\), where \(\mathcal{P}_V\) is the set of Pauli operators on the qubits indexed by the graph vertices (deviations), such that:
  \begin{itemize}
  \item $\mathcal{P}_V \setminus \mathcal E_\epsilon \subseteq \mathcal{E}_\nu$.
  \item $\Id \in \mathcal E_\nu$.
  \end{itemize}
  If Trappified Blind Delegated Quantum Computation uses a trappified scheme which:
  \begin{itemize}
  \item \(\epsilon\)-detects $\mathcal E_\epsilon$.
  \item is \(\delta\)-insensitive to at least $\{\Id\}$.
  \item is \(\nu\)-correct on $\mathcal E_\nu$.
  \end{itemize}
  then the protocol is \(\max(\epsilon, \delta+\nu)\)-secure against an arbitrarily malicious unbounded prover.
\end{theorem}

\begin{proof}[Sketch of Proof]
  In short: (i) with probability at most \(\epsilon\) a harmful Pauli slips past detection; (ii) undetected Paulis lie in \(\mathcal{E}_\nu\), so their logical effect is bounded by \(\nu\); (iii) insensitivity ensures we do not falsely abort more than \(\delta\). The triangle inequality then gives the stated bound.
\end{proof}

The following theorem addresses the sensitivity to noise of the initial VBQC protocols \cite{FK17unconditionally,KW17optimised}. It links insensitivity to robustness:
\begin{theorem}[Robust Detection Implies Robust Verifiability ()Informal)]
  We assume now that the prover in  Trappified Blind Delegated Quantum Computation is honest-but-noisy: the error applied is in \(\mathcal E_\delta\), the set of \(\delta\)-insensitive deviations, with probability \((1-p_\delta)\) and \(\mathcal{P}_V \setminus \mathcal E_\delta\) with probability \(p_\delta\).  Then, the verifier accepts with probability at least \((1-p_\delta)(1-\delta)\). If furthermore \(\mathcal{E}_\delta \subseteq \mathcal{E}_\nu\), the set of \(\nu\)-correct deviations, then the total correctness error is \(p_\delta + \delta + \nu\).
\end{theorem}

\begin{proof}[Sketch of Proof]
  The theorem is obtained after quantifying: (i) how often honest-noise falls in the insensitive set, passes the tests and gives an incorrect result; and (ii) how often it falls outside of it and is then likely to provide an erroneous result.
\end{proof}

In essence, it acknowledges that by tuning the acceptance function, we can change the size of the set of deviations that the trappified scheme is insensitive to. By choosing the embedding accordingly, this allows to minimize the size of the deviations that yield a correct result but are detected and therefore make the protocol abort. By making these adjustments, it is possible to ensure that the insensitive set is as large as possible and comprises the deviations that are due to noise with high probability.

Combining both, we get the concrete, \cmp{BQP}-specific statement:
\begin{theorem}[Security and Noise-Robustness of \cref{hdr:proto:tbdqc}, Theorem 13 from \cite{KKLM22unifying}]
  Let \(\cmp{BQP}_G\) be the set of \cmp{BQP} computations that can be performed on graph \(G\), with bounded error \(c\). Let \(\sch P\) be a trappified scheme on graph \(G\) and let \(\mathcal E_\epsilon, \mathcal E_\delta, \mathcal{E}_\nu\) be subsets of Pauli deviations such that:
  \begin{itemize}
  \item $\mathcal{E}_\nu \subseteq \{\cptp E \in \mathcal{P}_V \mid \forall \cptp C \in \cmp{BQP}_G, \forall T \in \sch P, \Cdev{\cptp E} = \Cdev{\cptp \Id}\}$.
  \item $\mathcal{P}_V \setminus \mathcal E_\epsilon \subset \mathcal{E}_\nu$.
  \item $\sch P$ \(\epsilon\)-detects $\mathcal E_\epsilon$, is \(\delta\)-insensitive to $\mathcal E_\delta$ and perfectly insensitive to $\Id$.
  \item The honest prover's noise is modeled by sampling for each computation or test round an error $\cptp E \in \mathcal E_\delta$ with probability $p_\delta$ and $\cptp E = \Id$ with probability $1 - p_\delta$. 
  \end{itemize}
  Let \(n, d, w\) be defined as in \cref{hdr:proto:tbdqc-classical-io}. If the following conditions are satisfied:
  \begin{itemize}
  \item $\frac{w}{(n-d)(1 - \epsilon)}$ is away by a constant and upper-bounded by $\frac{1 - 2c}{2 - 2c}$.
  \item $p_\delta$ is away by a constant and upper-bounded by $\frac{w}{(n-d)\delta}$.
  \end{itemize}
  Then \cref{hdr:proto:tbdqc}, using \(\sch P\), \(\eta(n)\)-constructs the Secure Delegated Quantum Computation (\cref{hdr:res:sdqc}) for computations in set \(\cmp{BQP}_G\) in the Abstract Cryptography framework.
  \label{hdr:thm:tbdqc-classical-io-security}
\end{theorem}
Note that the value of \(\eta(n)\) heavily depends on the value of \(\delta\) and \(\epsilon\), in particular via the coefficient in the exponential. This is a motivation for carefully designing the trappification scheme and embedding algorithm.
\subsection{Traps from Stabilizer Tests}
\label{hdr:sec:stabilizer-tests}
To instantiate a rich family of trappified canvases on any graph state, subset-stabilizer tests are natural. Let \(\grp S\) be the stabilizer group for \(\ketbra G\), the graph state associated to \(G = (V, E)\), and consider the set of canonical generators of \(\grp S\) equal to \(\{ \cptp S_v = \X_v \bigotimes_{(v,w)\in E} \Z_w, v \in V \}\).  We can write uniquely any \(\cptp R \in \grp S\) as \(\prod_{v \in V} S_v^{r_v}\) with \(r_v \in \bin\) and define \(\one_{\cptp R} := \{v \in V \mid r_v = 1\}\) to be the support of \(\cptp R\) in the generators \(\cptp S_v\).

A test round for \(\cptp R \in \grp S\) can be obtained by:
\begin{enumerate}
\item Input: prepare an eigenstate of \(\cptp R\).
\item Compute: Instruct an MBQC pattern consisting of measuring \(\cptp R\).
\item Decide: Accept iff the measured outcome corresponds to the prepared eigenstate.
\end{enumerate}

The test for a given stabilizer \(\cptp R \in \grp S\) is defined by having the verifier prepare each qubit \(v\) in the state:
\begin{itemize}
\item \(\ket{0}\) if \(\cptp R(v) = \pm \Z\).
\item \(\ket{+}\) otherwise.
\item It flips the qubit---while staying in the same basis---for \(v_0\) if there was a minus sign in \(\cptp R\).
\end{itemize}

The verifier then simply instructs the prover to measure each qubit in the \(\Y\)-basis if \(\cptp R(v) = \pm \Y\) and in the \(\X\) basis otherwise, getting outcome \(t(v)\) for vertex \(v\). It then computes \(\tau(t) := \bigoplus_{v\in \one_{\cptp R}} t(v)\) to determine the outcome of the measurement of \(\cptp R\) on the inputs provided by the verifier. For an honest prover, this value should return \(0\).

The freedom in choosing which stabilizers \(\cptp R\) to include in the trappified scheme \(\sch P\) will be at the core of constructing dummyless verification protocol.
\section{Takeaways}
\label{sec:orge63ca02}
We have split verifiable blind delegation into three minimal modules (RSP, Trappification, Embedding) and then reassembled them into a full protocol with composable guarantees. The broader lessons are:

\begin{itemize}
\item \emph{Modularization is leverage, not ornament.} Once blindness, deviation detection, and amplification are isolated, security reduces to checking a short list of local properties (detection, insensitivity, correctness) plus an explicit amplification rule. Proof effort scales with the module you touch, not with the whole protocol.
\item \emph{Blindness is a primitive.} Treating Remote State Preparation as its own resource clarifies what must remain hidden (the classical label of the prepared state) and what may leak (nothing beyond the state itself). It also makes room for alternative implementations (two-prover, weak coherent pulses, limited rotations) without disturbing the rest of the stack.
\item \emph{Traps are computations, not just qubits.} Starting from single-qubit traps, we generalized the concept to trappified canvases and schemes. Any deterministic, classically predictable partial pattern qualifies, provided its outputs can be evaluated from MBQC outcomes.
\item \emph{Proper embedding is the hinge between detection and amplification.} It must (i) keep information from the computation from flowing into the traps, to preserve blindness and enable simulation, and (ii) place the computation inside a redundancy structure---code space or repetition set---so that harmful deviations are necessarily high weight. Low-weight deviations are then corrected and should fall into the insensitivity set.
\item \emph{Local parameters drive global bounds.} The trio (detection, insensitivity, correctness) and the amplification policy (code distance, number of repetitions, acceptance threshold) map directly to the security error in Abstract Cryptography and to the robustness of the protocol. Tuning any of them is a transparent trade-off: stricter detection raises protection against cheating but can hurt robustness to honest noise, and vice versa.
\end{itemize}

A design workflow also emerged:
\begin{enumerate}
\item Pick or engineer an RSP implementation that matches the verifier's hardware.
\item Choose a trappified scheme that detects the deviations you care about and is insensitive to the noise you expect.
\item Select an embedding/amplification strategy that makes undetected harmful attacks improbable.
\item Read off the global bounds from the template theorems.
\end{enumerate}

In the next chapters we exploit this plug-and-play structure: we swap the naive RSP implementation for a lean one to lighten the verifier, design dummyless trappification to cut prover overhead, extend the orchestration to multi-party settings, and finally address verifier-side noise robustness. Each improvement touches one module (or its interface) while leaving the others intact.
\chapter{Robust VBQC}
\label{hdr:sec:rvbqc}
\lettrine[refstring,lines=3,lraise=0.15]{W}{e now} take the modular machinery for a spin. This chapter instantiates the three modules of the previous one---Remote State Preparation for blindness, Trappification for deviation detection, and Embedding for amplification---in the simplest non–fault-tolerant setting. The target class is \cmp{BQP}, so amplification comes from plain repetition and a majority vote, not from a bulky quantum code. Tests and computations are pulled apart into separate rounds: traps live in their own, and a single-round looks no heavier than the unprotected UBQC pattern.

The objectives are:
\begin{itemize}
\item Show, concretely, how a full VBQC-style protocol is instantiated once the three modules are fixed.
\item Drive the prover’s space overhead to zero: every qubit in a computation round is used for the computation, while traps are exiled to distinct test rounds.
\item Tune the acceptance rule so that honest physical noise---at moderate rate---is treated as harmless, giving a first notion of robustness without fault-tolerance.
\item Read off global, composable security bounds directly from the local properties---detection, insensitivity, correctness---established for trappification and using the chosen amplification.
\end{itemize}

In short, this chapter is an example, not a detour: it shows that the framework is practical and that meaningful robustness and low overhead can already be achieved with repetition alone, provided one is content with classical outputs.
\section{Protocol}
\label{sec:org819e6b8}
The construction below is the straight-line instantiation of our three modules for the special case of \cmp{BQP} computations. Two design choices drive everything:
\begin{itemize}
\item \emph{Amplification by repetition only.} No fault-tolerant code is inserted; we just repeat the whole delegated computation \(d\) times and take a majority vote.
\item \emph{Zero space overhead on the prover.} Computation and tests live in different rounds. In a computation round every qubit is used for the target MBQC pattern; in a test round every qubit serves the trap. Thus a single honest computation round costs exactly what an unprotected UBQC run would cost.
\end{itemize}

Security is therefore tied to how many rounds we run, not to inflating a single-round. This will let us tolerate a moderate amount of honest noise by loosening the trap acceptance threshold.

\begin{marginfigure}
  \resizebox{6cm}{!}{\input{figs/5-1}}
  \caption{\raggedright Possible sequence of 6 rounds of computation and trap rounds for a 5-qubit graph MBQC. Dotted qubits are used for computation, Hatched ones are for traps, while empty ones are for dummies isolating single-qubit traps.}
  \label{hdr:fig:rvbqc-rounds}
\end{marginfigure}

For a given graph \(G\) with flow \(f\), that implicitly defines a computation class \(\mathfrak C\), we denote by \(\{V_i\}_{i\in K}\)  a \(|K|\)-coloring of \(G\). With that in place, the following protocol verifies computations in \(\mathfrak C\) (See \cref{hdr:fig:rvbqc-rounds} for a visual definition of the various types of rounds used by the protocol).

\begin{protocol}[Robust VBQC]
  \begin{algorithmic}[0]
    \STATE \textbf{Alice's Inputs:} Angles $\qty{\phi_v}_{v \in V}$ and flow $f$ on graph $G$, classical input to the computation $x \in \bin^{|I|}$ .

    \STATE \textbf{Protocol:}
    \begin{enumerate}
    \item[1.] Alice chooses uniformly at random a partition $(C, T)$ of $[n]$ ($C \cap T = \emptyset$) with $|C| = d$, the sets of indices of the computation and test rounds respectively.

    \item[2.] For $j \in [n]$, Alice and Bob perform the following sub-protocol (Alice may send message $\Redo_j$ to Bob before step 2.c while Bob may send it to Alice at any time, both parties then restart round $j$ with fresh randomness):

      \begin{enumerate}
      \item[(a)] If $j \in T$ (test), Alice chooses uniformly at random a color $\mathsf{V}_j \sample \qty{V_k}_{k \in [K]}$ (this is the set of traps for this test round).
      \item[(b)] Alice sends $|V|$ qubits to Bob. If $j \in T$ and the destination qubit $v \notin \mathsf{V}_j$ is a non-trap qubit (therefore a dummy), then Alice chooses uniformly at random $d_v \sample \bin$ and sends the state $\ket{d_v}$. Otherwise, Alice chooses at random $\theta_v \sample \Theta$ and sends the state $\ket{+_{\theta_v}}$.
      \item[(c)] Bob performs a $\CZ$ gate between all its qubits corresponding to an edge in the set $E$.
      \item[(d)] For $v \in V$, Alice sends a measurement angle $\delta_v$, Bob measures the appropriate corresponding qubit in the \(\delta_v\)-basis, returning outcome $b_v$ to Alice. The angle $\delta_v$ is defined as follows:
        \begin{itemize}
        \item If $j \in C$ (computation), it is the same as in UBQC, computed using the flow and the computation angles $\qty{\phi_v}_{v \in V}$. For $v \in I$ (input qubit) Alice uses $\tilde{\theta}_v = \theta_v + x_v\pi$ in the computation of $\delta_v$.
        \item If $j \in T$ (test): if $v \notin \mathsf{V}_j$ (dummy qubit), Alice chooses it uniformly at random from $\Theta$; if $v \in \mathsf{V}_j$ (trap qubit), it chooses uniformly at random $r_v \sample \bin$ and sets $\delta_v = \theta_v + r_v\pi$.
        \end{itemize}
      \end{enumerate}

    \item[3.] For all $j \in T$ (test round) and $v \in \mathsf{V}_j$ (traps), Alice verifies that $b_v = r_v \oplus d_v$, where $d_v = \bigoplus_{i \in N_{G}(v)} d_i$ is the sum over the values of neighboring dummies of qubit $v$. Let $c_{\mathit{fail}}$ be the number of failed test rounds (where at least one trap qubit does not satisfy the relation above), if $c_{\mathit{fail}} \geq w$ then Alice aborts by sending message $\abort$ to Bob.

    \item[4.] Otherwise, let $y_j$ for $j \in C$ be the classical output of computation round $j$ (after corrections from measurement results). Alice checks whether there exists some output value $y$ such that $|\left\{ y_j \, | \, j \in C,\, y_j = y \right\}| > \frac{d}{2}$. If such a value $y$ exists (this is then the majority output), it sets it as its output and sends message $\ok$ to Bob. Otherwise it sends message $\abort$ to Bob.
    \end{enumerate}
  \end{algorithmic}
  \label{hdr:proto:rvbqc}
\end{protocol}
\section{Module definitions}
\label{sec:org0c96d26}
\subsection{Remote State Preparation}
\label{sec:org496e461}
The RSP implementation is the naive one. We simply have the verifier prepare quantum states that are sent to the prover via a quantum channel. This is enough for now as any post-delivery noise is by construction \(\theta\)-\emph{independent}. With Alice playing the verifier and Bob the prover we have:
\begin{protocol}[RSP (Naive Instantiation)]
  \begin{algorithmic}[0]

    \STATE \textbf{Public Information:} $\Lambda = \{0, 1, +_{\theta} \mbox{ for } \theta \in \Theta\}$ the classical description of the states that can be prepared.

    \STATE \textbf{Alice's Input:} $\lambda \in \Lambda$

    \STATE{\textbf{Preparation of the state:}}
    \begin{itemize}
    \item Alice prepares $\ket \lambda$ in a quantum register.
    \end{itemize} 

    \STATE{\textbf{Transmission of the state:}}
    \begin{itemize}
    \item Alice sends the quantum register prepared in $\ket \lambda$ to Bob.
    \end{itemize}

  \end{algorithmic}
  \label{hdr:proto:rvbqc-rsp}
\end{protocol}
\subsection{Trappification}
\label{sec:orgd9cd98a}
Because amplification is done using time in sequential rounds, not in space, the trappified canvases are simple: the trap is made of extra copies of the same graph state  \(\ket G\) prepared with inputs that force deterministic outcomes (See \cref{hdr:fig:rvbqc-rounds}).
\begin{itemize}
\item For a computation round: inputs and angles follow UBQC for the computation defined by \(G\), flow \(f\) and measurement angles \(\phi_v, \ v \in V\).
\item For a test round: we prepare another copy of \(G\), but (i) fill non-trap vertices with dummies \(\ket{0},\ket{1}\) to isolate single trap-qubits when needed, and (ii) choose angles so that selected vertices appear random.
\end{itemize}

Randomly interleaving which rounds are tests and which vertices host traps gives blindness over trap placement. The acceptance is decided when we know how many test rounds are correct, giving an idea of how many computation rounds are unaffected.

Choosing this trappification scheme, there is no space overhead per round. Computation and tests live in different rounds. Every qubit of a computation round serves the target MBQC; every qubit of a test round is either a single trap-qubit or a dummy. Hence a single-round costs exactly what the unprotected UBQC would.
\subsection{Embedding}
\label{sec:orgb22ef7e}
Here the embedding module is nothing but repetition:
\begin{itemize}
\item Choose \(d\) indices for computation rounds \(C\subset[n]\), run the target MBQC there.
\item Run \(n-d\) test rounds drawn from the trappified scheme.
\item Decode by majority vote on the \(d\) outcomes.
\end{itemize}

The amplification relies on the embedding of the computation result into a classical repetition code. For a \cmp{BQP} instance with error \(c\)---otherwise said with a gap bigger than \(c\)---the harmless deviations are those that flip less than a fraction \(\frac{2c-1}{2c-2}\) of the computation rounds. For in this case, it is impossible to flip the result of the majority vote. Therefore, these should be included into the insensitive set. To this end, we adjust the acceptance rule to accept as many as possible of these harmless deviations. Given that a computation with a single deviated qubit might yield the wrong result, the best we can do is to have at most \(\frac{2c-1}{2c-2}\times d\) computation rounds affected by a deviation. This in turn can be probed by the test rounds thanks to the randomization of their locations and the overall blindness. One only needs to keep in mind that test rounds do not detect deviations with probability 1. For universal graphs like the Brickwork state, this is equal to \(0.5\) as the chance of a given qubit to be a trap in a test round is \(\frac{1}{2}\).\sidenote{\raggedright \footnotesize This allows to have only two types of test rounds while ensuring a uniform coverage of all qubits. In practice, this means that a single qubit deviation will be detected with a probability independent its location. From a security point of view, this seems optimal, yet with respect to noise robustness, the question is largely unexplored.}   Putting everything together, if we set the acceptance function to output \(0\) as long as the fraction of failed test rounds is below and bounded away from \(\frac{1}{2}\times \frac{2c-1}{2c-2}\) by a constant then the accepted deviations are guaranteed to be harmless.

As a consequence, the protocol remains correct and accepts with high probability even when the prover is noisy, provided the global fraction of failed rounds stays below and bounded away by a constant from the threshold of \(\frac{1}{2}\times\frac{2c-1}{2c-2}\).
\section{Zero Space-Overhead Robust Verification}
\label{sec:org8b8830a}
Using the framework introduced earlier, it is rather straightforward to prove the security of \cref{hdr:proto:rvbqc}.
\marginnote{\\ The fact that we have such sharply distinct behavior between aborting with probability one and accepting with probability one in case of low noise signals the interest of these verification techniques for probing noise strength. This is an avenue that is currently pursued for setting benchmarking frameworks of quantum computers and networks that was indeed pioneered in the Quantum Internet Alliance project \cite{AFCD23requirements}. This possibility offers an unexpected leverage to encourage hardware makers to integrate the necessary features for implementing verification into their devices as it not only serves a utility-regime protocol---verification---but also fulfills a more immediate need of providing certifiable benchmarks for their early machines.}
\begin{theorem}[Security of \cref{hdr:proto:rvbqc}]
  Let \(G\) be a \(k\)-colorable graph, and \(\mathsf C\) a computation defined as an MBQC pattern on \(G\). For \(n = d+t\) such that \(d/n\) and \(t/n\) are fixed in \((0,1)\) and \(w\) such that \(w/t\) is fixed in \((0, \frac{1}{k}\cdot \frac{2c-1}{2c-2})\), where \(c\) is the inherent error probability of the \cmp{BQP} computation, \cref{hdr:proto:rvbqc} with \(d\) computation rounds, \(t\) test rounds, and a maximum number of tolerated failed test rounds of \(w\) is \(\epsilon\)-composably-secure with \(\epsilon\) exponentially small in \(n\).
  \label{hdr:thm:rvbqc-security}
\end{theorem}
This clearly demonstrates that the absence of space-overhead is not detrimental to efficiency as the security error remains negligible in the total number of round \(n\).

The robustness can be obtained likewise. One can even obtain that setting the threshold too low would yield to overwhelmingly aborting the computation.
\begin{theorem}
  As before, \(c\) denotes the inherent error probability for the \cmp{BQP} computation instances running as MBQC patterns on a \(k\)-colorable graph \(G\). Assume a Markovian round-dependent model for the noise on Verifier and Prover devices and let \(p_{\mathrm{min}} \leq p_{\mathrm{max}} < \frac{1}{k}\cdot \frac{2p-1}{2p-2}\) be respectively a lower and an upper-bound on the probability that at least one of the trap-qubit measurement outcomes in a single test round is incorrect. If \(w/t > p_{\mathrm{max}}\), \cref{hdr:proto:rvbqc} is \(\epsilon_{\mathrm{cor}}\)-correct with exponentially low \(\epsilon_{\mathrm{cor}}\). On the other hand, if \(w/t < p_{\mathrm{min}}\), then the probability that \cref{hdr:proto:rvbqc} terminates without aborting is exponentially low.
  \label{hdr:thm:rvbqc-correctness}
\end{theorem}
\section{Takeaways}
\label{sec:org50a2617}
We have shown that using UBQC wrapped in a schedule of computation and test rounds together with a majority vote enable a verification protocol that is both prover-light and noise robust.

In summary: 
\begin{itemize}
\item Every qubit in a computation round serves the target task; traps live in separate test rounds.
\item As outputs are  classical, we boost detection via majority vote, not heavy fault-tolerant codes.
\item The threshold on failed tests trades amplification strength against robustness to honest noise.
\item RSP is the plain UBQC sender-prepares primitive; trappification are stand-alone trap rounds; embedding is repetition and majority decoding.
\item Local guarantees---detection, insensitivity, correctness, are lifted to global AC security, while honest-but-noisy provers still pass.
\end{itemize}

This sets the stage for the next chapters: slimming the verifier further, adding multi-party features, and eventually handling fully fault-tolerant workloads.
\chapter{Dummyless VBQC}
\label{hdr:sec:dummyless}
\lettrine[refstring, lines=3, lraise=0.15]{T}{he goal} of this chapter is to show, by direct instantiation of the modular framework, that verifiable blind delegation does not actually need dummies\textasciitilde{}(see \cref{hdr:page:dummy}) when the target computation has classical input and classical output. In VBQC, dummies serve to disconnect parts of the graph and insert traps. Here we replace them entirely by (i) an intrinsic symmetry of classical-Input/Output MBQC, which makes one specific non-trivial Pauli error harmless, and (ii) stabilizer traps built only from \(\I, \X, \Y\) operators. The outcome is a protocol that:
\begin{itemize}
\item Keeps blindness.
\item Detects every harmful deviation with negligible failure probability.
\item Imposes no extra space overhead on the prover: every qubit in a computation round can be used for the computation as traps live in separate rounds.
\item Uses the same RSP instantiation as robust VBQC (\cref{hdr:proto:rvbqc}) but with states restricted to \(XY\) plane.
\end{itemize}

Beyond immediate practicality (simpler sources, simpler calibration), the exercise is conceptually useful: it stresses how the framework lets us push invariances into the insensitive set and trim the trappification module accordingly, while preserving full composable security. More importantly, the obtained protocol will be used for all subsequent improvements.
\section{Protocol}
\label{sec:org3f63f7f}
The dummyless protocol is a direct instantiation of VBQC (\cref{hdr:proto:vbqc}) with specific choices of RSP, Trappification and Embedding. Indeed, the next section will first instantiate the Trappification as it will unlock the ability to have an RSP with only states in the equatorial plane. The embedding will be the repetition as we will restrict ourselves to \cmp{BQP} computations.
\section{Module Definitions}
\label{sec:org530237a}
\subsection{Dummyless Trappification}
\label{hdr:sec:dummyless_verification}
\subsubsection{A Natural Invariance of MBQC with Classical Input and Output}
\label{hdr:sec:natural-invariance}
In MBQC, computation qubits, i.e. \(v\in O^c\), are measured in the \(\ket{\pm_{\phi'(v)}}\) basis, with \(\phi'(v) \in \Theta=\qty{\frac{k\pi}{4}}_{k \in \qty{0, \ldots, 7}}\) fixed by the pattern. Because a projective measurement only depends on its projectors, any unitary that leaves those projectors invariant will not change the measurement statistics. When measuring along the \(\phi'(v)\)-axis in the \(XY\) plane, any rotation around that axis let the outcome statistics unchanged. This invariance is routinely used in UBQC security proofs to twirl the prover's attack.

The same observation extends beyond unitaries. If we allow any local invertible map that reflects the Bloch sphere with respect to the \(XY\) plane, the projectors onto \(\ket{\pm_{\phi'(v)}}\) are still left unchanged. Hence, measurement probabilities are unaffected. When \emph{all} qubits are measured in the \(XY\) plane---corresponding to classical-I/O---this invariance propagates to the final classical result.

To this end, we introduce a linear map \(F_A\) that flips the sign of every \(\Z\) term on a chosen subset \(A\subset V\) in the Pauli expansion of a state \(\rho\) on \(V\). More formally, for
\begin{align}
  \rho &
         = \sum_{\cptp P\in \{\I,\X,\Y,\Z\}^{\otimes n}} \alpha_{\cptp P} \cptp P, \mbox{ then } \\
  F_A(\rho) &
              = \sum_{\cptp P\in \{\I,\X,\Y,\Z\}^{\otimes n}} (-1)^{\zwt_A(\cptp P)} \alpha_{\cptp P} \cptp P,
\end{align}
where \(\zwt_A(\cptp P) = |\{ v\in A | \cptp P_v = \Z \}|\) counts the number of vertices in \(A\) on which \(\cptp P\) equals the Pauli \(\Z\).
Then we have:
\begin{lemma}
  Let \(G = (V, E)\)  and an MBQC pattern on \(G\) measuring all qubits in \(V\). Then for any \(A \subset G\), applying \(F_A\) right before the \(XY\) plane measurements leaves the output distribution unchanged.
  \label{hdr:lem:fa_harmless}
\end{lemma}

The map \(F_A\) is not generally a unitary, but it is still possible to find a unitary transformation that has the same effect as \(F_A\) on \(\ketbra{G}\). As a result, \(F_A\left[\ketbra{G}\right]\) is a physical state:
\begin{lemma}
  For any graph \(G = (V,E)\) it holds that \(F_A\left[\ketbra{G}\right] = \cptp U \ketbra{G} \cptp U^\dagger\), where
  \begin{align}
    \cptp U = \prod_{\substack{v \in V, \\ \deg{v} \equiv 1 \!\!\!\!\! \pmod{2}}} \Z_v
  \end{align}
  describes the application of \(\Z\) operators to all odd-degree vertices of \(G\).
  \label{hdr:lem:fv_as_unitary}
\end{lemma}

Putting both lemmas together yields the specific harmless deviation that will matter later:
\begin{theorem}
  Let \(G = (V,E)\) be a graph and \(\cptp E^*\) be the unitary operation given by
  \begin{align}
    \cptp E^* = \prod_{\substack{v \in V, \\ \deg{v} \equiv 1 \!\!\!\!\! \pmod{2}}} \Z_v,
  \end{align}
  describing the application of \(\Z\) operators to all odd-degree vertices of \(G\). For MBQC on \(G\) with classical input and output, the application of \(\cptp E^*\) before the measurements has no effect on the results of the computation.
  \label{hdr:thm:harmless}
\end{theorem}

In short, for any classical-I/O MBQC there exists a non-trivial, generally non-stabilizer Pauli deviation---the product of \(Z\)'s on odd-degree vertices---that is undetectable by our forthcoming dummyless traps, yet harmless for correctness. We should exclude it explicitly from the set of errors we try to catch as we pay no security price in doing so.
\subsubsection{Trappification Scheme}
\label{sec:org648c2f4}
We now exploit this newly found invariance. Because we restrict ourselves to preparations in the \(XY\) plane and MBQC authorizes only \(XY\) measurements, every trap we insert must use only those resources. Concretely, we want stabilizers of the graph state that never ask us to prepare or measure a \(\Z\) eigenstate. That pushes us to look for generators of the stabilizer group of \(\ket G\) made only of \(\I,\X,\Y\). This is because preparing an eigenstate of a generator containing a \(Z\) at a given location  is done by preparing a \(\ket 0\) state for this location.

Intuitively, we need to (i) build as many independent, deterministic \(\X,\Y\) stabilizer tests as the graph allows; (ii) randomize over them; (iii) ignore the one residual deviation that we proved harmless. We will show that this is indeed enough to detect all harmful deviations.

The following lemma formalizes the first step.
\begin{lemma}
  Let \(G = (V,E)\), the graph state \(\ket G\) and its stabilizer group \(\grp S\). There exists \(|V|-1\) generators of \(\grp S\) that are tensor products of \(\I\), \(\X\) and \(\Y\) only. Denote by \(\grp R\) the subgroup generated by those \(|V| -1\) elements.
  \label{hdr:lem:stab_tests}
\end{lemma}

We can now define our trappified scheme \(\sch P\) by sampling uniformly at random one of these \(|V|-1\) generators and running the corresponding stabilizer test traps (\cref{hdr:sec:reconstruction}).
\begin{definition}
  A dummyless trappified scheme \(\sch P\) with parameters \(N,d\) for a \cmp{BQP} computation on graph \(G\) is obtained by considering \(N\) rounds that will be delegated using UBQC on graph \(G\). \(d\) rounds will be sampled uniformly at random to host the delegated computation, and for each of the remaining \(N-d\) ones, a trap will be sampled at random by choosing one of the \(|V| - 1\) stabilizer of \(G\) that contain only \(\Id, \X, \Y\) operators.
  \label{hdr:def:dummyless-scheme}
\end{definition}

Because of the previous lemma, one stabilizer generator of \(\ket G\) is outside \(\grp R\). Call it \(\cptp S_0\). Using the analysis of the previous section, it is in fact easy to discover. It consists of a tensor product of \(\Z\)'s at the odd-degree vertices. As a consequence, all the Pauli operators that commute with \(\grp R\) and are not in \(\grp R\)---these are non-trivial undetectable errors---are necessarily of the form \(\cptp S_0 \grp R\). Because they commute with \(\grp R\) they cannot be detected (see \cref{hdr:fig:dummyless}). Fortunately, \cref{hdr:thm:harmless} shows that \(\cptp S_0\) is indeed harmless for the computation and so are all the elements in \(\cptp S_0 \grp R\).

\begin{marginfigure}[-10cm]
  \resizebox{6cm}{!}{\begin{tikzpicture}[
  dot/.style={circle, fill=black, inner sep=1.5pt},
  empty/.style={circle, draw=black, thick, inner sep=1.5pt},
  node font=\tiny\sffamily\bfseries
  ]

  \definecolor{myred}{RGB}{200,30,30}
  \definecolor{mygreen}{RGB}{30,160,30}
  \definecolor{mywhite}{RGB}{255,255,255}

  \newcommand{\drawlattice}{
    \foreach \x in {0,1,2} {
      \draw (\x, 1) -- (\x, 2);
    }
    \draw (0,2) -- (0.5,2.5) -- (1,2) -- (1.5,2.5) -- (2,2) -- (2.5,2.5);
    \draw (0,2) -- (-0.5,2.5);
    \foreach \x in {0,1,2} {
      \draw (\x, 1) -- (\x-0.5, 0.5) -- (\x, 0);
      \draw (\x, 1) -- (\x+0.5, 0.5) -- (\x, 0);
    }
  }

  \begin{scope}[shift={(0,0)}]
    \drawlattice
    \foreach \x in {0,1,2} {
      \node[dot] at (\x,0) {}; \node[dot] at (\x,1) {}; \node[dot] at (\x,2) {};
      \node[dot] at (\x-0.5, 0.5) {}; \node[dot] at (\x+0.5, 0.5) {};
      \node[dot] at (\x-0.5, 2.5) {}; \node[dot] at (\x+0.5, 2.5) {};
    }
    
    \node[mywhite] at ( 2.7,0.5) {X};
    
  \end{scope}

\end{tikzpicture}}
  \par\bigskip\bigskip
  \resizebox{6cm}{!}{\begin{tikzpicture}[
    dot/.style={circle, fill=black, inner sep=1.5pt},
    empty/.style={circle, draw=black, thick, inner sep=1.5pt},
    node font=\tiny\sffamily\bfseries
]

\definecolor{myred}{RGB}{200,30,30}
\definecolor{mygreen}{RGB}{30,160,30}

\newcommand{\drawlattice}{
    \foreach \x in {0,1,2} {
        \draw (\x, 1) -- (\x, 2);
    }
    \draw (0,2) -- (0.5,2.5) -- (1,2) -- (1.5,2.5) -- (2,2) -- (2.5,2.5);
    \draw (0,2) -- (-0.5,2.5);
    \foreach \x in {0,1,2} {
        \draw (\x, 1) -- (\x-0.5, 0.5) -- (\x, 0);
        \draw (\x, 1) -- (\x+0.5, 0.5) -- (\x, 0);
    }
}

\begin{scope}[shift={(0,0)}]
  \drawlattice
  \foreach \x in {0,1,2} {
    \node[dot] at (\x,0) {}; \node[dot] at (\x,1) {}; \node[dot] at (\x,2) {};
    \node[dot] at (\x-0.5, 0.5) {}; \node[dot] at (\x+0.5, 0.5) {};
    \node[dot] at (\x-0.5, 2.5) {}; \node[dot] at (\x+0.5, 2.5) {};
  }

  \foreach \x  in {0,1,2} {
    \node[myred] at (\x.2,0) {X}; \node[myred] at (\x.2,1) {Y}; \node[myred] at (\x.2,2) {Y};
  }

  \node[myred] at (-0.3,0.5) {X}; \node[myred] at (-0.3,2.5) {Y};
  \node[myred] at ( 0.7,0.5) {X}; \node[myred] at ( 0.7,2.5) {X};
  \node[myred] at ( 1.7,0.5) {X}; \node[myred] at ( 1.7,2.5) {X};
  \node[myred] at ( 2.7,0.5) {X}; \node[myred] at ( 2.7,2.5) {Y};

\end{scope}

\end{tikzpicture}}
  \par\bigskip\bigskip
  \resizebox{6cm}{!}{\begin{tikzpicture}[
    dot/.style={circle, fill=black, inner sep=1.5pt},
    empty/.style={circle, draw=black, thick, inner sep=1.5pt},
    node font=\tiny\sffamily\bfseries
]

\definecolor{myred}{RGB}{200,30,30}
\definecolor{mygreen}{RGB}{30,160,30}

\newcommand{\drawlattice}{
    \foreach \x in {0,1,2} {
        \draw (\x, 1) -- (\x, 2);
    }
    \draw (0,2) -- (0.5,2.5) -- (1,2) -- (1.5,2.5) -- (2,2) -- (2.5,2.5);
    \draw (0,2) -- (-0.5,2.5);
    \foreach \x in {0,1,2} {
        \draw (\x, 1) -- (\x-0.5, 0.5) -- (\x, 0);
        \draw (\x, 1) -- (\x+0.5, 0.5) -- (\x, 0);
    }
}

\begin{scope}[shift={(0,0)}]
  \drawlattice
  \foreach \x in {0,1,2} {
    \node[dot] at (\x,0) {}; \node[dot] at (\x,1) {}; \node[dot] at (\x,2) {};
    \node[dot] at (\x-0.5, 0.5) {}; \node[dot] at (\x+0.5, 0.5) {};
    \node[dot] at (\x-0.5, 2.5) {}; \node[dot] at (\x+0.5, 2.5) {};
  }
  \node[mygreen] at (0.2,0) {Y}; \node[mygreen] at (0.2,1) {X}; \node[myred] at (0.2,2) {Y};
  \node[mygreen] at (1.2,0) {Y}; \node[mygreen] at (1.2,1) {X}; \node[myred] at (1.2,2) {Y};
  \node[myred] at (2.2,0) {X}; \node[myred] at (2.2,1) {Y}; \node[myred] at (2.2,2) {Y};
  
  \node[myred] at (-0.3,0.5) {X}; \node[myred] at (-0.3,2.5) {Y};
  \node[mygreen] at ( 0.7,0.5) {I}; \node[myred] at ( 0.7,2.5) {X};
  \node[myred] at ( 1.7,0.5) {X}; \node[myred] at ( 1.7,2.5) {X};
  \node[myred] at ( 2.7,0.5) {X}; \node[myred] at ( 2.7,2.5) {Y};

\end{scope}

\end{tikzpicture}}
  \par\bigskip\bigskip
  \resizebox{6cm}{!}{\begin{tikzpicture}[
    dot/.style={circle, fill=black, inner sep=1.5pt},
    empty/.style={circle, draw=black, thick, inner sep=1.5pt},
    node font=\tiny\sffamily\bfseries
]

\definecolor{myred}{RGB}{200,30,30}
\definecolor{mygreen}{RGB}{30,160,30}

\newcommand{\drawlattice}{
    \foreach \x in {0,1,2} {
        \draw (\x, 1) -- (\x, 2);
    }
    \draw (0,2) -- (0.5,2.5) -- (1,2) -- (1.5,2.5) -- (2,2) -- (2.5,2.5);
    \draw (0,2) -- (-0.5,2.5);
    \foreach \x in {0,1,2} {
        \draw (\x, 1) -- (\x-0.5, 0.5) -- (\x, 0);
        \draw (\x, 1) -- (\x+0.5, 0.5) -- (\x, 0);
    }
}

\begin{scope}[shift={(0,0)}]
  \drawlattice
  \foreach \x in {0,1,2} {
    \node[dot] at (\x,0) {}; \node[dot] at (\x,1) {}; \node[dot] at (\x,2) {};
    \node[dot] at (\x-0.5, 0.5) {}; \node[dot] at (\x+0.5, 0.5) {};
    \node[dot] at (\x-0.5, 2.5) {}; \node[dot] at (\x+0.5, 2.5) {};
  }
  \node[myred] at (0.2,0) {X}; \node[myred] at (0.2,1) {Y}; \node[myred] at (0.2,2) {Y};
  \node[myred] at (1.2,0) {X}; \node[mygreen] at (1.2,1) {I}; \node[mygreen] at (1.2,2) {I};
  \node[myred] at (2.2,0) {X}; \node[myred] at (2.2,1) {Y}; \node[myred] at (2.2,2) {Y};
  
  \node[myred] at (-0.3,0.5) {X}; \node[myred] at (-0.3,2.5) {Y};
  \node[mygreen] at ( 0.7,0.5) {Y}; \node[mygreen] at ( 0.7,2.5) {Y};
  \node[mygreen] at ( 1.7,0.5) {Y}; \node[mygreen] at ( 1.7,2.5) {Y};
  \node[myred] at ( 2.7,0.5) {X}; \node[myred] at ( 2.7,2.5) {Y};

\end{scope}

\end{tikzpicture}}
  \par\bigskip\bigskip
  \caption{\raggedright Traps based on "poking holes" on even degree vertices (up), traps based on chains linking two odd degree vertices (middle), the uncaught deviation that has no effect on \cmp{BQP} computations (down).}
  \label{hdr:fig:dummyless}
\end{marginfigure}
\subsection{Remote State Preparation}
\label{hdr:sec:dummyless-rsp}
Nothing fancy is needed here: we reuse the RSP construction of \cref{hdr:proto:rvbqc-rsp}. The only difference for the dummyless setting is that now the verifier never has to send \(\ket 0, \ket 1\) states. Indeed, \(\ket{+_{\theta}}, \ \theta \in \Theta\) suffice.

\begin{protocol}[Single-Plane RSP]
  \begin{algorithmic}[0]

    \STATE \textbf{Public Information:} $\Theta$

    \STATE \textbf{Alice's Input:} $\theta \in \Theta$

    \STATE{\textbf{Preparation of the state:}}
    \begin{itemize}
    \item Alice prepares $\ket{+_{\theta}}$ in a quantum register.
    \end{itemize} 

    \STATE{\textbf{Transmission of the state:}}
    \begin{itemize}
    \item Alice sends the quantum register prepared in $\ket{+_{\theta}}$ to Bob.
    \end{itemize}

  \end{algorithmic}
  \label{hdr:proto:dummyless-rsp}
\end{protocol}
As in the robust case, the resource is secure-by-design: Bob learns only what the quantum state itself reveals, and any later deviation cannot depend on \(\theta\).
\subsection{Embedding}
\label{hdr:sec:dummyless-embedding}
Just as in robust VBQC (\cref{hdr:proto:rvbqc}), amplification is obtained by repetition and majority vote: classical-I/O lets us avoid using fault-tolerant quantum computation procedures.  The embedding algorithm here is also the same: run the same computation round \(d\) times; interleave with \(t\) trap rounds sampled from \(\sch P\); post-process by majority vote over the \(d\) computation rounds.

The properness of the embedding algorithm---no information flow from computation to traps---is ensured by keeping rounds disjoint.
\section{Verification without dummies}
\label{hdr:sec:dummyless-security}
Using the previous definitions we can now apply \cref{hdr:thm:tbdqc-classical-io-security} using the  trappified scheme \(\sch P\) defined earlier (\cref{hdr:def:dummyless-scheme}) that:
\begin{enumerate}
\item Detects the set \(\mathcal{E}\) of harmful \(\cptp Z\)-Pauli errors.
\item Correctly evaluates the target computation in the presence of any other \(\cptp Z\)-Pauli error in at least \(\mathcal{P}_V^{\cptp Z} \setminus \mathcal{E}\). \cref{hdr:thm:harmless}, then shows that there is a specific error \(\cptp E^\ast\) which never affects the output distribution of the target computation and thus does not need to be detected.
\end{enumerate}

This results in:
\begin{theorem}
  Let \(G=(V,E)\) be a graph, and \(\sch P\) the trappified scheme on \(G\) defined by sampling at random from a generating set of \(\grp R\) containing only stabilizers with no \(\Z\)'s. Then, \(\sch P\) constructs the SDQC \cref{hdr:res:sdqc} for \(\BQP\) computations that can be embedded on the graph \(G\) with negligible correctness error and security error.
\end{theorem}
\section{Takeaways}
\label{sec:dummyless_takeaways}
We have shown that dummies are not intrinsically required for verifying quantum computations. This rests on: 
\begin{itemize}
\item Building traps entirely in the \(XY\) plane using stabilizer generators without \(\Z\)’s.
\item Recognizing and factoring out a single unavoidable, but harmless, deviation.
\item Reuse the repetition-based amplification and acceptance tuning of robust VBQC.
\end{itemize}

Practically, the verifier no longer needs to prepare computational basis states, simplifying hardware; the prover still sees the same per-round resource state as in VBQC; and composability is preserved by working inside the framework (\cref{hdr:sec:framework}).
\part{Reducing Requirements for the Verifier}
\label{hdr:sec:reducing}
\cleardoublepage
\lettrine[refstring, lines=3, lraise=0.15]{T}{his Part} explores two ways in which the requirements for verification can be reduced on the Verifier's side. In both cases the goal is to equip verifiers with good quality, robust and cheap devices. The first one deals with verifiers without the ability to prepare quantum states and relying on the prover to do that for them. The second is about verifiers using weak coherent pulses---obtained by attenuating a laser---instead of single-photon sources. The difficulty with those changes is that the devices now provide less advanced or even imperfect functionalities.  As a consequence, this needs to be compensated by more complex protocols, whose security is harder to obtain because of the increased amount of information exchanged and the number of interactions with the prover.
\chapter{Trusted Rotations}
\label{hdr:sec:rotations}
\lettrine[refstring, lines=3, lraise=0.15]{T}{he objective} of this chapter is to lighten the verifier: no single-photon source, no state preparation. Instead, we assume only that the verifier can apply trusted single-qubit \(\Z\)-rotations and random \(\X\) flips to whatever qubit the prover hands over. Because Remote State Preparation (RSP) is a composable module, it suffices to replace it locally by an equivalent primitive. We show that a Remote Rotation with Dephasing resource achieves this for \cmp{BQP} computations, provided the prover is trusted to send qubits.
\section{Motivation}
\label{sec:B1-motive}
VBQC-like protocols assume the verifier can prepare \(\ket{+_\theta}\) states, and sometimes also \(\ket 0, \ket 1\) states. This implies some relatively heavy set-up for the verifier: good near-deterministic photon sources require cryogenics and careful calibration and maintenance.  If the prover could prepare the raw qubits and the verifier only perform additional trusted rotations, life would be easier.

This question was first raised as open in \cite{DKL11universal}. In \cite{MKAC22qenclave} it was shown to provide blindness. We exhibit below a solution that is statistically secure but has one caveat. The source is dimension-bounded: the verifier is guaranteed to receive \emph{bona fide} qubits.

Other approaches provide verification with light verifiers. \cite{FKD18reducing} and \cite{FKD19accrediting}  remove the need for trusted preparation by considering some physically motivated noise---thus do not provide verification in the fully malicious scenario. This is in addition to fully classical verifier verification that holds under computational assumptions \cite{M18classical}, but which incurs a large overhead for the prover.
\section{Security Failures Without Trusted Preparations}
\label{sec:org651f861}
Here we want to dispel a confusion that we have seen repeatedly in published papers: Blindness alone does not stop an adversary from mounting selective attacks on dummy qubits in VBQC type protocols. Consider the following idealized resource where Alice instructs a unitary on a state held by Bob:
\begin{resource}[Remote Unitary (RU)]
  \begin{algorithmic}[0]
    \State \textbf{Alice's Input:} A unitary $\cptp U$ chosen from $\{\Ha, \X\Ha, \Z(\theta) \mbox{ for } \theta \in \Theta\}$.
    \State \textbf{Bob's Input:} A state $\rho$
    \State \textbf{Computation by the Resource:} The Resource outputs $\cptp U\left[\rho\right]$ at Bob's interface.
  \end{algorithmic}
  \label{hdr:res:remote-unitary}
\end{resource}

A trivial implementation consists for Bob to send a state \(\rho\) to Alice; Alice applies  \(\cptp U\) and sends it back. The problem with this resource is that it does not allow performing verified quantum computation---irrespective of its implementation. Indeed, it lets the prover selectively attack the dummy qubits, even though it does not know which of the qubits are dummies.

For instance, instantiate the RSP with RU where the prover's input is set to \(\ket +\). Let the verifier decide to prepare a dummy. The verifier simply needs to perform \(\Ha\) on the provided state and send it back. If it wants to prepare a \(\ket{+_{\theta}}\) state, it simply performs a \(\Z(\theta)\) rotation. The problem is that the prover could decide to cheat by sending \(\ket -\) instead of \(\ket +\), and apply \(\Z\) whenever it receives the qubit back from the verifier. A quick calculation shows that \(\ket {+_{\theta}}\) states are unaffected, but dummies get an \(\X\) deviation (See \cref{hdr:fig:attack-on-dummies}). This provides a way to craft a concrete selective attack on the dotted-triple-graph variant of VBQC \cite{KW17optimised} (\cref{hdr:proto:vbqc}).

\begin{marginfigure}
  \resizebox{6cm}{!}{\begin{tikzpicture}[
    thick,
    gate/.style={draw, fill=white, minimum size=1em, font=\tiny},
    ket/.style={anchor=west, font=\tiny}
]

\newcommand{\quantumcircuit}[5]{
    \begin{scope}[yshift=#1]
        \draw[dashed] (2.2, 0.6) -- (2.2, -1.1);
        
        \draw (1, 0.2) -- (0.5, 0.2) -- (0.5, -0.7) -- (1, -0.7);
        
        \draw (1, 0.2) -- (3.2, 0.2);
        \draw (1, -0.7) -- (3.2, -0.7);
        
        \node[gate] at (1.5, -0.7) {#4};
        
        \draw (3.2, 0.2) -- (2.7, 0.2); 
        \node[ket] at (3.2, 0.2) {#2};
        
        \draw[->] (2.7, -0.7) -- (3.2, -0.7); 
        \node[ket] at (3.2, -0.7) {#3};
        
        \ifnum#5=1
            \node[gate] at (2.6, 0.2) {\(\Z\)};
            \node[gate] at (2.6, -0.7) {\(\Z\)};
        \fi
    \end{scope}
}


\quantumcircuit{6cm}{$\ket{+}$}{$\ket{0}$}{\(\Ha\)}{0}

\quantumcircuit{4cm}{$\ket{+}$}{$\ket{+_{\theta}}$}{$\Z(\theta)$}{0}

\quantumcircuit{2cm}{$\ket{+}$}{$\ket{1}$}{\(\Ha\)}{1}

\quantumcircuit{0cm}{$\ket{+}$}{$\ket{+_{\theta}}$}{$\Z(\theta)$}{1}

\node[font=\tiny] at (1.2, -1.5) {Sender};
\node[font=\tiny] at (3, -1.5) {Receiver};

\end{tikzpicture}}
  \caption{Concrete attack on dummy qubits obtained by having the receiver provide $\ket -$ states instead of a $\ket +$ states.} 
  \label{hdr:fig:attack-on-dummies}
\end{marginfigure}

This is a manifestation of a deeper phenomenon: deviations or noise can pick-up a dependency on the secret parameters used to parameterize the preparation circuit and break the security proofs of the protocols. This will become a crucial observation and a difficulty in designing protocols for fault-tolerant quantum computations, as honest noise occurring during the preparation of states by the verifier can pick-up dependency upon secret parameters which voids the security of naive protocols.
\section{Recovered Security with Remote State Rotation}
\label{sec:org6ee8555}
To neutralize the above attack and remove any potential ones due to secret-dependency picked by an incorrect preparation, we switch to the dummyless setting. The sender in the RSP only needs to ensure \(\ket{+_{\theta}}\) states are available at the receiver's interface once its actions are done. The difficulty is that we want to allow whatever single-qubit state the receiver\sidenote{\raggedright \footnotesize We have kept the terminology of sender, receiver from the RSP resource, yet in our setting the state that is available to the sender is prepared by the receiver. This could be because the receiver initially sends it over a quantum channel, or more likely this corresponds to a situation where the source that is used to create the initial states upon which the sender will apply $\cptp U$ is potentially malicious.} is providing to the sender, both parties still end up with a secure, composable RSP.

The selective-attack from the previous subsection---hidden \(\Z\) before and after---no longer works: \(\Z\) commutes with \(\Z(\theta)\). But a subtler problem remains. If the Prover sends a state that is not in the \(XY\) plane, say \(\Y(\alpha) \ket +\) for \(\alpha \in [0,\pi/2]\), then generally \(Z(\theta)\Y(\alpha)\ket + \neq \Y(\alpha)Z(\theta)\ket +\), so the deviation cannot be pushed "after" the resource. In composable terms, the prover’s pre-rotation now depends on the verifier’s secret \(\theta\) once the operations are reordered. This then breaks the security of the implementation of the RSP resource.

The cure is to force the incoming state into the \(XY\) plane before applying \(\Z(\theta)\). A uniformly random \(\X\) flip does exactly that: it decoheres in the \(\X\) basis, wiping any \(\Z\)-component that could pick up \(\theta\). This leads to the following resource where Alice plays the sender and Bob the receiver:
\begin{resource}[Remote Rotation with Dephasing ($\text{RR}_\text{D}$)]
  \begin{algorithmic}[0]
    \STATE \textbf{Inputs:}
    \begin{itemize}
    \item Alice inputs an angle $\theta \in \Theta$.
    \item Bob inputs a single-qubit in state $\rho$.
    \end{itemize}
    \STATE \textbf{Computation by the Resource:}
    \begin{enumerate}
    \item The Resource samples a bit $b \sample \bin$ uniformly at random.
    \item The Resource outputs a qubit in state $Z(\theta) \X^b\left[\rho\right]$ to Bob.
    \end{enumerate}
  \end{algorithmic}
  \label{hdr:res:rrd}
\end{resource}

Constructing the 8-state \(XY\) plane RSP from \(\mathrm{RR}_{\mathrm{D}}\) is done easily (see \cref{hdr:fig:single-plane-rsp}
\begin{protocol}[Single Plane RSP from $\text{RR}_\text{D}$]
  \begin{algorithmic} [0]
    \STATE \textbf{Input:} Alice inputs an angle $\theta \in \Theta$.
    \STATE \textbf{Protocol:}
    \begin{enumerate}
    \item Alice and Bob call the $\text{RR}_\text{D}$ \cref{hdr:res:rrd}:
      \begin{itemize}
      \item Alice inputs the angle $\theta$.
      \item Bob inputs a qubit in the state $\ket{+}$.
      \item The Resource returns a single-qubit to Bob who sets it as its output.
      \end{itemize}
    \end{enumerate}
  \end{algorithmic}
  \label{hdr:proto:rsp-from-rrd}
\end{protocol}

\begin{marginfigure}
  \resizebox{6cm}{!}{\begin{tikzpicture}[
    thick,
    gate/.style={draw, fill=white, minimum size=1em, font=\tiny},
    ket/.style={anchor=west, font=\tiny}
]


\draw[dashed] (2.2, 0.6) -- (2.2, -1.1);

\draw (1, 0.2) -- (0.5, 0.2) -- (0.5, -0.7) -- (1, -0.7);

\draw (1, 0.2) -- (3.2, 0.2);
\draw (1, -0.7) -- (3.2, -0.7);

\node[gate] at (1.5, -0.7) {\(\Z(\theta)\)};

\draw (3.2, 0.2) -- (2.7, 0.2); 

\draw[->] (2.7, -0.7) -- (3.2, -0.7); 

\node[gate] at (1.5, 0.2) {\(\X^{a}\)};

\node[font=\tiny] at (1.2, -1.5) {Sender};
\node[font=\tiny] at (3, -1.5) {Receiver};

\end{tikzpicture}}
  \caption{Single Plane RSP from $\text{RR}_\text{D}$ protocol.}
  \label{hdr:fig:single-plane-rsp}
\end{marginfigure}

The following theorem ensures that the construction is secure:
\begin{theorem}[Security of Protocol \cref{hdr:proto:rsp-from-rrd}]
  \cref{hdr:proto:rsp-from-rrd} perfectly constructs \cref{hdr:res:rsp} where the states alphabet \(\Sigma = \Theta\) from \cref{hdr:res:rrd}.
  \label{hdr:thm:rsp-from-rrd}
\end{theorem}
\section{Verification with Remote State Rotation}
\label{sec:org8948797}
Thanks to composability, nothing dramatic happens at the protocol level: in dummyless VBQC---which follows \cref{hdr:proto:tbdqc-classical-io}---replace each call to the single plane RSP by \cref{hdr:proto:rsp-from-rrd}, i.e. use \(\mathrm{RR}_{\mathrm{D}}\) under the hood. This constructs an SDQC resource with the same security parameters. The construction error for RSP is zero when using \(\mathrm{RR}_{\mathrm{D}}\), so the global bound is untouched.

There is, however, one caveat worth flagging explicitly. In the informal discussion above we implicitly assumed the prover still hands over a single-qubit. If instead higher-dimensional systems can slip through---e.g., two photons when one is expected---the random-\(\X\) dephasing no longer guarantees that the prover’s deviation is \(\theta\)-independent.  A concrete threat in photonic implementations is the photon-splitting attack: send two photons, let both accumulate the secret \(\Z(\theta)\), then split and measure to glean \(\theta\). Curbing this requires an extra assumption---trusted dimension---or an explicit countermeasure---which could be provided by the WCP techniques discussed later. The moral is rotations-only RSP solves the preparation problem, not the dimension problem.
\section{Takeaways}
\label{sec:B7-take}
We have shown that stripping the verifier down to trusted \(\Z(\theta)\)-rotations and random \(\X\)-flips is enough for verifiable delegation---provided the prover really sends a qubit.

Two broad lessons emerge. Firstly, a single phase modulator and a fast bit flip suffice to verify computations. Blind-only schemes such as \cite{MKAC22qenclave} can be upgraded to full verifiability by a drop-in replacement of their RSP and appropriate trappification changes. Secondly, negligible security error survives. Because \(\mathrm{RR}_{\mathrm{D}}\) perfectly constructs RSP, any VBQC variant using RSP inherits its proof unchanged.
Yet, this also emphasizes that blindness is not a force field. Even if the prover cannot locate dummies, it can still craft attacks that selectively hurt specific positions---or their complements. Several papers overlook this subtlety, undermining their claims.
\chapter{Weak Coherent Pulses}
\label{hdr:sec:wcp}
\lettrine[refstring, lines=3, lraise=0.15]{O}{ur goal} in this chapter is to replace the single-photon Remote State Preparation module by one implementable with weak coherent pulses (WCP). Concretely:
\begin{itemize}
\item We keep the dummyless VBQC; only the RSP module is replaced.
\item We model multi-photon leakage and take into account the loss acknowledgment by the prover as the attack lever, then build a resource that still delivers at least one secret qubit per batch.
\item We recover verifiability by combining that batch resource with the phase-accumulation gadget of \cite{DKL11universal}.
\item We improve over \cite{DKL11universal} by providing verifiability and practicality: the number of pulses scales like \(1/\eta^2\)---where \(\eta\) is the channel transmittance---instead of \(1/\eta^4\) in the original scheme.
\end{itemize}

In short, this chapter shows how to (i) verify while tolerating realistic laser sources, while (ii) maintaining acceptable performance in spite of the much larger attack surface of the RSP module.
\section{Motivation}
\label{sec:orgbec9ad4}
The previous chapter argued that single-photon sources are a costly bottleneck on the verifier’s side. Weak coherent pulses (WCPs) offer a different incentive: they propagate farther through fiber---thanks to higher usable launch power, mature telecom hardware---and they are already ubiquitous in QKD deployments. If we could base an RSP construction on WCPs, long-distance verified delegation could possibly ride on existing QKD infrastructure with minimal extra optics.

Two issues must be faced head-on:
\begin{itemize}
\item \emph{Multi-photon leakage and acknowledgments.} A WCP occasionally contains 2+ photons. A malicious prover can keep one, measure the rest, and learn the secret angle \(\theta\). Because the channel is lossy, the prover must be allowed to confirm reception; he can craft a successful attack by simply declaring all 0- and 1-photon pulses lost and accept only exploitable ones from his malicious perspective.
\item \emph{Incomplete or non-composable security in earlier work.} The original WCP proposal of \cite{DKL11universal} is phrased in a composable framework, but it only guarantees blindness as it couldn't produce dummies and hence couldn't verify. Later decoy-state based refinements \cite{JWHX19remote} improve efficiency but abandon a full AC treatment. Additionally, we show a concrete flaw that undermines their claim to security for VBQC---and even affects the standard QKD decoy analysis.
\end{itemize}

Our approach is to stay strictly within Abstract Cryptography. We model a realistic WCP device as a resource, prove how to construct from it a batch containing at least one qubit prepared in a genuine \(\ket{+_{\theta}}\) state, and plug that \emph{batchRSP} resource into the dummyless VBQC (\cref{hdr:proto:tbdqc-classical-io} with dummyless tests). The reward is composable security against fully malicious provers; and by using multiple intensities, a scaling that improves from \(1/\eta^4\) to \(1/\eta^2\) for a fiber with transmittance \(\eta\).
\section{Security Failures Due to Multi-Photon Pulses}
\label{sec:org94dc380}
Remote State Preparation's security asks that the receiver\sidenote{\raggedright \footnotesize Here again we switch to the RSP terminology of sender / receiver as it could be used outside the verification setting of verifier / prover.} learns nothing beyond the quantum state it receives, and that any imperfection can be pushed into a deviation after delivery.  Weak coherent pulses (WCP) depart from this specification in a decisive way.

An attenuated laser does not emit single photons on command, it emits pulses whose photon number is Poisson distributed with mean \(\mu\).  Vacuum shots are harmless, single-photon shots are fine, but as soon as two---or more---photons slip through, a malicious receiver can split the pulse.  One (or more) photon is measured to glean information about the secret angle \(\theta\); at least one is kept pristine to satisfy later checks.  The deviation has now acquired a dependence on the secret.  Composable security forbids exactly this: whatever imperfection over the ideal resource must map to a deviation independent of \(\theta\) and applied after the perfect state is delivered.

Indeed, the protocol itself hands the adversary the steering wheel.  Because most pulses are empty and the channel is lossy, the sender must let the receiver acknowledge which pulses arrived.  That acknowledgment is a perfect filter: discard every 0- or 1-photon event, keep only the multi-photon ones you can exploit.  Blindness does not save us here; the filter acts on photon number, not on the intended preparation.

The verdict is immediate.  An RSP constructed naively from WCP source is not secure. The leakage of \(\theta\) on multi-photon rounds cannot be swept under the rug and treated as a post-resource deviation.  Unless we alter the construction by constraining the acceptance rule so that any profitable photon-number filtering becomes statistically visible, verification fails before trappification or embedding even enter the stage.
\section{Recovered Security and Performance}
\label{sec:orgc36236b}
The cure comes in two doses. First, we pin down exactly what a weak coherent pulse source gives an honest or a cheating receiver. Second, we add classical statistics and a tiny quantum gadget to make any profitable multi-photon filtering show up in those statistics. We then provide two statistical tests. The one of \cite{DKL11universal} which gives an asymptotic \(1/\eta^4\) scaling and one inspired by the decoy state method that scales as \(1/\eta^2\).
\subsection{Modeling the source: WCPGenerator}
\label{sec:org38570f5}
We wrap the physical laser and lossy fibre in a resource that is deliberately pessimistic for the dishonest case. If a proof survives this model, it survives the lab. We call this resource WCPGenerator.
 
When the receiver is honest, it receives \(n\) photons drawn at random from \(\Pois(\mu \eta)\)---\(\mu\) is the power of the laser and \(\eta\) is the channel trans\-mittance---prepared as \(\left(\cptp U\left[\rho\right]\right)^{\otimes n}\).  When it cheats, it gets the lossless distribution \(\Pois(\mu)\), and if \(n>1\) it is even handed the classical description of \(\cptp U\). For Alice playing the sender and Bob the receiver this gives the following resource:\sidenote{\raggedright \footnotesize Here we emphasize again that \cref{hdr:res:wcp} is not a physically realistic resource, but rather a model that is pessimistic enough so that reasonable physical resources are more powerful than this one. As a consequence, if we prove security with this resource, we will have security with a more powerful resource.}

\begin{resource}[WCPGenerator]
  \begin{algorithmic}[0]
    \State \textbf{Public information:} set of unitaries $\mathcal{U}$; classical description of a quantum state $\rho$, power of the laser $\mu$, transmittance $\eta \in [0,1]$.
    \State \textbf{Alice's Input:} $\cptp U \in \mathcal{U}$ and $\mu \in \mathbb{R}^+_0$.
    \State \textbf{Bob's Input:} $c\in\bin$, set to $0$ if honest.
    \State{\textbf{Computation by the Resource}}{}
    \begin{itemize}
    \item $c=0$, it samples $n \sample \Pois(\mu\eta)$ and sends the state $\cptp U\left[\rho\right]^{\otimes n}$ to Bob.
    \item $c=1$, it samples $n \sample \Pois(\mu)$. If $n\leq 1$, it sends the state $\cptp U\left[\rho\right]^{\otimes n}$ to Bob. If $n>1$, it sends $n$ and the classical description of $\cptp U$ to Bob.
    \end{itemize}
  \end{algorithmic}
  \label{hdr:res:wcp}
\end{resource}
\subsection{Phase–accumulation gadget (from \cite{DKL11universal})}
\label{sec:org7e0b801}
A gadget does the heavy lifting for giving the sender some handle over the receiver.  Label qubits \(0,\ldots,N\) and prepare them in the \(XY\) plane.  Apply \(\CNOT\)'s with control on qubit 0 and target on qubit \(i\) for all \(i\geq 1\). Then measure qubits 1 to \(N\) in the computational basis. Up to the known outcomes, the phase angles add on qubit 0.  One unknown angle therefore serves as a one-time pad for the output phase.

From a security perspective, if among \(N\) pulses at least one is truly single-photon, the gadget can channel that secrecy into a single clean qubit (See \cref{hdr:fig:gadget}).

\begin{marginfigure}
  \resizebox{6cm}{!}{\begin{tikzpicture}[
    thick,
    target/.style={circle, draw, fill=white, inner sep=0pt, minimum size=0.35cm},
    dot/.style={circle, fill=black, inner sep=0pt, minimum size=0.15cm},
    font=\small\sffamily
]

\newcommand{\drawmeter}[2]{
    \begin{scope}[shift={(#1,#2)}]
        \draw[fill=white] (-0.2, -0.3) -- (0.05, -0.3) 
              arc [start angle=-90, end angle=90, radius=0.3] 
              -- (-0.2, 0.3) -- cycle;
        \draw (-0.15, -0.1) arc [start angle=150, end angle=30, radius=0.2];
        \draw[->, >=stealth] (-0.0, -0.15) -- (0.25, 0.2);
    \end{scope}
}

\newcommand{\drawtarget}[2]{
    \node[target] (T) at (#1,#2) {};
    \draw (T.north) -- (T.south);
    \draw (T.east) -- (T.west);
}


\foreach \i in {0,1,2,3,4} {
  \draw (0, \i) -- (5.5, \i);
}

\foreach \i in {1,2,3,4} {
  \draw[cwire] (5.5, \i) -- (5.8, \i);
}

\draw[wire] (5.5, 0) -- (5.8, 0);

\drawtarget{1}{4}
\draw (1, 4) -- (1, 0); 
\node[dot] at (1, 0) {};
\drawmeter{5.3}{4}
\node[anchor=west] at (5.7, 4) {$b_0$};

\drawtarget{2}{3}
\draw (2, 3) -- (2, 0); 
\node[dot] at (2, 0) {};
\drawmeter{5.3}{3}
\node[anchor=west] at (5.7, 3) {$b_1$};

\drawtarget{3}{2}
\draw (3, 2) -- (3, 0); 
\node[dot] at (3, 0) {};
\drawmeter{5.3}{2}
\node[anchor=west] at (5.7, 2) {$b_2$};

\drawtarget{4}{1}
\draw (4, 1) -- (4, 0); 
\node[dot] at (4, 0) {};
\drawmeter{5.3}{1}
\node[anchor=west] at (5.7, 1) {$b_3$};

\end{tikzpicture}}
  \caption{This gadget produces the quantum state $\ket{+_\theta}$ where $\theta = \theta_0 + \sum_{i=1}^{n} (-1)^{b_i}\theta_i$.}
  \label{hdr:fig:gadget}
\end{marginfigure}
\subsection{First recovery: single-intensity defense (from \cite{DKL11universal})}
\label{sec:org37a066c}
The insight of \cite{DKL11universal} is twofold.  First, given the phase–accumulation gadget, one single genuine single-photon is enough to hide the overall phase. This is because of the additive nature of the overall phase that is one-time-padded by the unknown phase of the single-photon pulse. Hence the optimal cheating move for the receiver is obvious: make sure the sender never gets a lone-photon pulse.  Declare every 0- and 1-photon pulse “lost”, keep only 2-or-more-photon pulses.  The gadget then has no unknown angle left to channel. 

Second, the gadget itself gives the sender a handle to fight back---by simple counting.  In the honest case, the number of pulses the receiver cannot acknowledge is governed by the WCP intensity and the loss: essentially the fraction of 0-photon events after the channel.  In the dishonest strategy above, every 0- and 1-photon emissions must be discarded to prevent secrecy from sneaking through.  By tuning the WCP intensity appropriately, the first one can be made smaller than the second by a constant margin allowing efficient statistical testing.

This translates into the following actions of the sender:
\begin{enumerate}
\item He chooses a mean photon number \(\mu\) that fulfills the above condition.
\item He checks the receiver’s "not-received" reports against the honest prediction.   Too many discards would yield an  abort; the multi-photon filtering cannot be hidden.
\end{enumerate}

This restores the secure RSP construction with a single intensity and a single threshold test.  But there are two shortcomings in the original proposal--besides the absence of verifiability of the whole:
\begin{itemize}
\item The resource cost explodes as \(1/\eta^4\) when the fiber transmittance \(\eta\) becomes small making it inefficient for long distance.
\item Photon-number–resolving detection were assumed on the receiver even when he is honest; strictly speaking, that is unnecessary once we look closer at the proof.
  These are the issues the next subsection fixes by moving to two intensities and constructing \emph{batchRSP}.
\end{itemize}
\subsection{Recovered Security and Performance}
\label{sec:orga4ab287}
The solution proposed below is described in \cite{GLMO24composably} and later refined into \cite{GLMO25multi}.  It is based on the previous approaches of \cite{DKL11universal} for the gadget and general approach toward the proof and on \cite{JWHX19remote} for the motivation behind using decoy states.
Instead of going for the full RSP, we focus instead on the creation of a batch of size \(K\) where at least one pulse has a single photon. Indeed, this is enough to construct RSP using the previously defined gadget. 
\begin{resource}[BatchRSP]
  \begin{algorithmic}[0]
    \STATE \textbf{Public Information:} $K \in \mathbb{N}^+$; set of unitaries $\mathcal{U}$; classical description of a quantum state $\rho$.
    \STATE \textbf{Alice's Input:} $K$ unitaries $\cptp U_1,\ldots, \cptp U_K \in \mathcal{U}$.
    \STATE \textbf{Bob's Input:} $c \in\bin$, set to $0$ if honest.
    \STATE{\textbf{Computation by the Resource}}{}
    \begin{itemize}
    \item {$c=0$}, it sends the states $\cptp U_1 \left[\rho\right], \ldots, \cptp U_K \left[\rho\right]$ to Bob.
    \item $c=1$
      \begin{itemize}
      \item It samples $k \sample \{ 1,\ldots,K \}$ uniformly at random.
      \item  It sends $k$, the classical descriptions of $\{ \cptp U_i \}_{i\neq k}$, and the quantum state $\cptp U_k\left[\rho\right]$ to Bob.
      \end{itemize}
    \end{itemize}
  \end{algorithmic}
  \label{hdr:res:batchrsp}
\end{resource}

To ensure that BatchRSP has at least one pulse with a single photon, we construct a statistical test similar in spirit to the one above. However, this time we will use 2 intensities.

Informally, the sender will generate pulses deciding their intensity randomly. He will keep track of the intensity of the acknowledged pulses. Then it will compute a quantity using the number of pulses acknowledged for each intensity.

More concretely, let \(\nu\) and \(\nu'\) be the two intensities and \(P\) and \(P'\) the corresponding number of acknowledged pulses. Define \(a = e^{-\nu}, b = \nu e^{-\nu}, c = \nu^2 e^{-\nu}/2\) and similarly for \(a', b', c'\) where \(\nu\) is replaced by \(\nu'\). These are the expected fractions of 0, 1 and 2 photons pulses in a WCP of intensity \(\nu\) and \(\nu'\). Next we define two numbers that the sender can compute on its own given the intensities of the laser it uses, the expected value of the transmittance \(\eta\) and the reported acknowledged pulses:
\begin{align}
  T & = \frac{c'P - cP'}{bc' - b'c}.\\
  t & = \frac{c'(1-e^{-\eta\nu}) - c(1-e^{-\eta\nu'})}{bc'-b'c}\,\frac{N}{2}.
\end{align}

For a total number of \(N\) pulses---half for each intensity---and with \(I\) the set of acknowledged pulse, the sender accepts or aborts using the following algorithm:
\begin{algorithm}[Algorithm $\mathcal B$]
  \begin{algorithmic}[0]
    \State \textbf{Input:} $I$ 
    \State{\textbf{Estimation:}}
    \begin{itemize}
    \item Fix $\Delta_0 > 0$.
    \item Compute $t$ and $T$.
    \end{itemize}
    \State \textbf{If} {$T  \geq t - \Delta_0 \frac{N}{2}$}, return $\accept$.
    \State \textbf{Otherwise}, return $\abort$.
  \end{algorithmic}
  \label{hdr:algo:b}
\end{algorithm}
The value of \(\Delta_{0}\) is a security margin for the statistical test that essentially gives a higher confidence in the decision when \(\Delta_0\) is bigger for a fixed \(N\).

This can be summarized in a full protocol:
\begin{protocol}[Multi-Intensity Weak Coherent Pulse Method]
  \begin{algorithmic}[0]
    \State \textbf{Public Information:} $N, K\in \mathbb{N}^+$, where $K \leq N$; group of unitaries $\mathcal{U}$; classical description of a quantum state $\rho$ ; transmittance $\eta \in [0,1]$; pulse intensities $\nu, \nu' \in \mathbb{R}^+_0$; efficient, classical estimation algorithm $\mathcal{B}$.
    
    \State{\textbf{Alice's Input:}} $K$ unitaries $\cptp U_1,\ldots, \cptp U_K \in \mathcal{U}$.

    \State{\textbf{Alice - Sending WCP}}{}
    \begin{itemize}
    \item Sample a random permutation $\pi \sample \operatorname{S}_N$.
    \item Let $(\nu_1, \ldots, \nu_N) \gets (\nu_{\pi(1)}, \ldots, \nu_{\pi(N/2)}, \nu'_{\pi(N/2+1)}, \ldots, \nu'_{\pi(N)})$.
    \item {$i \in [1,N]$}
      \begin{itemize}
      \item  Sample $\cptp U_i' \sample \mathcal{U}$ from the Haar measure on $\mathcal{U}$.
      \item Call WCPGenerator with $\cptp U = \cptp U_i'$, $\mu = \nu_i$, the classical description of the quantum state $\rho$, and transmittance $\eta$.
      \end{itemize}
    \end{itemize}

    \State{\textbf{Bob - Acknowledging WCP Reception}}{}
    \begin{itemize}
    \item  Let $\tilde{I}$ be the set of rounds whose pulses contain at least one photon.
    \item if {$|\tilde{I}| < K$}  Send $\abort$ to the Sender and stop.
    \item Let $I' \gets$ random subset of $\tilde{I}$ of size $K$.
    \item For $i \in I'$, store one round-$i$ photon in quantum memory.
    \item Send $I'$ to the Sender.
    \end{itemize}

    \State{\textbf{Alice - Estimation}}{}
    \begin{itemize}
    \item Undo the permutation, i.e. let $I \gets \{ \pi^{-1}(i) | i \in I' \}$.
    \item Run estimation algorithm $\mathcal{B}$ on input $I$.
      \begin{itemize}
      \item If {$\mathcal{B}$ returns $\abort$},  send $\abort$ to Bob and stop.
      \end{itemize}
    \item Relabel the set of kept states in $I'$ from $1$ to $K$. Sample a random permutation $\sigma \sample \operatorname{S}_K$.
    \item {For $j \in [1,K]$}, let $\tilde{\cptp{U}}_j = \cptp{U}_{\sigma(j)} \cptp{U'}_j^{\dagger}$.
    \item Send the corrections $\sigma$, $(\tilde{\cptp{U}}_j)_{j\in \tilde{I}}$ to Bob.
    \end{itemize}

    \State{\textbf{Bob - Corrections}}{}
    \begin{itemize}
    \item     For {$j \in [1,K]$}, apply the unitary $\tilde{\cptp{U}}_j$ to the \(j\)-th kept state.
    \item Permute all kept $K$ quantum states using $\sigma$, and set the result as the output.
    \end{itemize}
  \end{algorithmic}
  \label{hdr:proto:decoy-state-method}
\end{protocol}

The following theorem then gives the security error of the construction when compared to the ideal BatchRSP (\cref{hdr:res:batchrsp}):
\begin{theorem}[Correctness and Security Errors]
  Given \(\Delta_0 > 0\) in Algorithm \(\mathcal B\), \(\delta > 0\) such that \(K = ((2-e^{-\eta\nu}-e^{-\eta\nu'})/2 - \delta)N\) and \(C = \max(c, c')\), the correctness error \(\varepsilon_{\text{corr}}\) satisfies
  \begin{align}
    \varepsilon_{\text{corr}} \leq \exp(-\delta^2 N) + \exp(-\frac{\Delta_0^2 (bc' - b'c)^2}{4 C^2}N)
  \end{align}
  while the security error is negligible in \(N\) whenever there are additional constants \(\Delta_0', \Delta_0'' > 0\) such that 
  \begin{align}
    \label{eq:uuubound}
    &\Delta_0 + \Delta_0' + \frac{c'}{bc'-b'c}\Delta_0'' \nonumber \\
    &\qquad = \frac{c'(1-e^{-\eta\nu}) - c(1-e^{-\eta\nu'}) - c'(1-a-b-c)}{bc'-b'c}. 
  \end{align}
  \label{hdr:thm:instantiation}
\end{theorem}

In addition, in the high-loss limit \(\eta\to0\), the number of pulses needed is given by the following theorem:
\begin{theorem}[Scaling for $\eta \to 0$]
  \label{thm:scaling}
  For \(\nu=\alpha\nu'\) with \(0<\alpha<1\), the number of pulses needed to obtained a given error scales as \(1/\eta^{2}\) when \(\eta\to0\).
\end{theorem}

The intuition is straightforward. Because every optical and classical choice is randomized, a malicious receiver’s strategy boils down to bookkeeping: count how many 0-, 1-, 2-, \ldots-photon pulses arrived, then decide how many in each bucket to acknowledge.  With that reduction in hand, one can easily check that the linear combination \(c'P - cP'\) is independent of whatever the receiver does with the 2-photon bucket---remember \(P\) and \(P'\) are the acknowledged counts for the two intensities, while \(c\) and \(c'\) are the corresponding Poisson weights of 2-photon events.  Removing the dominant loophole from the single-intensity setting of \cite{DKL11universal}, namely the free play on two-photon pulses, forces the adversary to rely on much rarer higher-photon events.  That is exactly what yields the tighter bounds and the improved \(1/\eta^2\) scaling in the low-transmittance regime.
\section{Verification with Weak Coherent Pulses}
\label{sec:org39af42c}
The full SDQC protocol needs no new security proof once the WCP-based preparation is in place.  We simply plug the composable building block in:
\begin{enumerate}
\item Run the dummyless VBQC protocol.
\item Every time the verifier would call the RSP resource as the sender, substitute the multi-intensity WCP procedure (\cref{hdr:proto:decoy-state-method}) followed by the gadget to distill a single hidden qubit.
\end{enumerate}

Abstract Cryptography then does the rest: the construction error of the RSP block adds to the global correctness and security bounds. No further simulator needs to be written, and no trap analysis has to be redone.

One warning remains: we have closed the side-channel opened by multi-photon event, yet we have assumed no other side-channels are present in the WCPGenerator. In the event a device exposes other side channels they must be addressed separately.
\section{Takeaways}
\label{sec:org9c54a6f}
In this chapter, we have gained practicality and insights.  We transformed weak off-the-shelf optics into a clean resource.  Insisting on composability, we obtain a verification protocol that:
\begin{itemize}
\item avoids single-photon sources and still verifies \cmp{BQP} computations.
\item scales as \(1/\eta^2\)---for two intensities---instead of \(1/\eta^4\), a practical gain over long fibres.
\item keeps the heavy lifting on the protocol side (randomization, estimation, privacy amplification), not on the hardware.
\end{itemize}

There are trade-offs though.  The gadget adds complexity, and pushing the scaling further---to the \(1/\eta\) of folklore decoy proofs---would likely require more intensities and an even hairier analysis.

In addition, the same machinery immediately strengthens QKD proofs: it fixes a gap in decoy-state security arguments and delivers full composability.

In short: careful modularization plus statistics let us recycle today’s QKD infrastructure for tomorrow’s verified quantum cloud.
\part{Extending Capabilities}
\label{hdr:sec:extending}
\cleardoublepage
\lettrine[refstring, lines=3, lraise=0.15]{T}{his chapter} explores two extensions to VBQC protocol. The first one is turning single-verifier protocols into secure multi-party delegated quantum computations. The development of Secure Multiparty Quantum Computation was chronologically the first project that was started. It lasted for a long time with ups and downs as we even broke an earlier version of the protocol, realizing at that occasion that blindness does not prevent attacks on specific types of qubits. We also better understood that attacker that is allowed to intervene in between trusted step rather than at the end was, even for perfectly blind protocols, is a much more adverse situation from the point of view of the honest party. While this project was started early, it is only after introducing our modular verification framework \cite{KKLM22unifying} that we managed to advance in the right direction. Indeed, it is for the purpose of this project that we introduced the dummyless protocol, and realized that the gadget from \cite{DKL11universal} could serve many purposes. To avoid repetition, we will present the lift from SMPC to SMPQC by building on the already presented dummyless verification protocol.

The second extension is about  securely delegating fault-tolerant quantum computations. The problem lies in the failure of the naive solution consisting of running entire protocols encoded into an error-correcting code. Indeed, even if one is willing to grant the verifier to perform single logical-qubit preparations, security remains the challenge.  The reason for such difficulty is that instead of using a 2-dimensional system, it is replaced by a high-dimensional system whose extra dimension increase the attack surface for a malicious party. This was acknowledged in detail in \cite{ABEM17interactive}. The solution we propose starts by properly defining imperfections---called compromise events---which are assumed to be stochastic. From there, it adds requirements to the fault-tolerance design rules in order to ensure not only correctness---standard fault-tolerance---but also security, whenever the compromise probability is below a constant threshold. This is then used to concretely construct the logically encoded RSP, thereby clearing the path to practical implementations of statistically secure, large-scale verified quantum computation.
\chapter{Secure Multiparty Quantum Computation}
\label{hdr:sec:smpqc}
\lettrine[refstring, lines=3, lraise=0.15]{T}{his chapter} shows how to lift a classical secure multiparty computation to the quantum world within our modular VBQC framework. We keep the hardware asymmetric---—one quantum server, many lightweight clients---and the security composable. The core steps are:
\begin{itemize}
\item Replace the single client VBQC private traps and pads by collective ones, so as to accommodate multiple parties without revealing where traps and data live.
\item Introduce a purely classical orchestrator to aggregate randomness, choose traps, and drive the UBQC dialogue without touching any qubits.
\item Realize the orchestrator’s only quantum need---remote state prepa\-ration---via a Collective RSP primitive that a single honest client is enough to secure.
\item Finally, discharge the orchestrator by running a standard classical SMPC among the clients, preserving the global security bound.
\end{itemize}

The result is a statistically secure lift of a classical SMPC into  SMPQC for \cmp{BQP} computations, achieved with minimal quantum trust and network complexity.
\section{Motivation}
\label{sec:orgb7d0ae3}
Classical secure multiparty computation (SMPC) has been around for decades and is now a mature toolbox: a set of parties can jointly evaluate a function of their private inputs while learning nothing beyond their own output.  Secure multiparty quantum computation (SMPQC) predates client–server verification: already in the early 2000s, quantum protocols were proposed to let several distrustful parties securely evaluate a joint circuit without revealing their inputs (see e.g. \cite{CGS02secure,BCGH06secure,DNS12actively,DGJM20secure}).

Our aim here is different in flavor.  We want to keep the VBQC asym\-metry---one quantum workhorse, many lightweight clients---while moving from a single client to many.  Doing this in MBQC raises two immediate questions:
\begin{enumerate}
\item How do we hide all sensitive structure---pads, traps, inputs---when no single client can be trusted with it?
\item How do we remain secure if an arbitrary subset of clients colludes with the server?
\end{enumerate}

Answering these will occupy the rest of the chapter: we first identify what truly needs to be shared, then show how to share it without breaking blindness or verifiability, and finally fold everything back into a composable security statement.
\section{Challenges for Security in Multi-Party Scenario}
\label{sec:orgf4316c4}
We adopt the strongest natural threat model: any party marked dishonest (server or client) may behave in an arbitrarily malicious, stateful, and adaptive way.  Dishonest parties may collude without restriction—exchanging all their classical transcripts and any quantum side information they hold.  Under such a model it is always sound to merge the colluding set into a single, all-powerful adversary.  By contrast, honest clients cannot be safely merged.  They do not know who else is honest and must not assume shared secrets.

With that in mind, a useful maneuver is to study extreme cases by aggregation. One such case is where all clients are honest, but the server dishonest. If we fuse the honest clients into one super-client, we should recover a textbook single-client VBQC: random \(\Z(\theta)\) pads hide inputs; traps catch deviations.

Split that super-client back into individual ones, however, and two problems surface.
\begin{itemize}
\item \emph{Ownership of traps is fatal.} Traps must cover every site where a harmful deviation could occur.  If a trap---or its one-time pad---is owned by client \(j\), then downstream dependencies will reveal where the client's data or trap lay---violating blindness.
\item \emph{Selective cheating becomes possible.} The colluding block (server + any dishonest clients) knows which qubits it helped encrypt.  It also knows which ones it did not.  Deviations can then target the latter, eroding the detection/insensitivity guarantees that underpin single-client VBQC security proofs.
\end{itemize}

Hence pads and traps cannot be private possessions.  They must be collective objects---prepared and encrypted so that no strict subset of honest players---and certainly not the adversary---can tell which qubits are traps or carry data. As a consequence of the traps covering all the possible locations, data will also need to be collectively encoded to keep the indistinguishability between the two.
\section{Recovered Security through Collective RSP}
\label{sec:org288531e}
The previous section told us what we need: traps and pads that belong to no single client. Here we show how to manufacture them.

The first trick is to introduce a new party, the orchestrator, who will own the prepared quantum state, but without having to have quantum capabilities on his own. The second trick is already an old one: it relies on the phase-adding gadget of \cite{DKL11universal}. If several parties each supply a random equatorial qubit, then---with the help of a round of \(\CNOT\)'s and measurements in the computation basis---their phases simply add up.  One honest client is enough to hide the sum.  That is exactly the collective one-time pad we require.

We formalize this as a resource and a simple protocol that constructs it, where several Alices play the clients, Bob the server, and Oscar the orchestrator.
\begin{resource}[Collective Remote State Preparation for $n$ Senders]
  \begin{algorithmic}[0]
    \STATE \textbf{Inputs:}
    \begin{itemize}
    \item Oscar has as input an angle $\theta \in \Theta = \qty{\frac{k\pi}{4}}_{k \in \qty{0, \ldots, 7}}$.
    \item The Alices have no input.
    \item Bob has no input.
    \end{itemize}
    \STATE \textbf{Computation by the Resource:} The Resource prepares and sends the state $\ket{+_\theta}$ to Bob.
  \end{algorithmic}
  \label{hdr:res:crsp}
\end{resource}

The protocol for constructing this resource is then:
\marginnote{\\ Providing a CRSP with the possibility to hide the preparation of dummy qubits is difficult. Indeed, early papers exploring RSP \cite{L00classical,LS03oblivious} point toward impossibility. More precisely, if we prepare dummies and states in the $XY$ plane, we prepare a /generic/ ensemble of states in the language defined in \cite{LS03oblivious}. By requiring security, we have that any such protocol can be transformed into an equivalent one in which the sender also performs operations that are fixed in advance. But if we provide the sender a fixed state, say $\ket +$ this cannot be the case. So we don't seem to be able to perform the operation---the link is that we want each of them to apply an encryption on top of an already encrypted state... but that seems to be possible if we restrict to OTP encoding by performing teleportation...}
\begin{protocol}[Collective Remote State Preparation]
  \begin{algorithmic} [0]
    \STATE \textbf{Input:} Oscar has as input an angle $\theta \in \Theta$. Alices and Bob have no input.
    \STATE \textbf{Protocol:}
    \begin{itemize}
    \item Alice number $j$ samples $\theta_j \sample \Theta$ and sends $\ket{+_{\theta_j}}$ to Bob.
    \item Alice $j$ sends $\theta_j$ to Oscar using a Secure Classical Channel.
    \item For each $j \neq n$, Bob applies $\CNOT_{n,j}$ between the qubits $n$ and $j$, with the first being the control and the second the target. It measures the target qubit $j$ in the computational basis with measurement outcome $t_j$. It sends the vector $\bm{t}$ containing all the measurement outcomes to Oscar.
    \item Oscar computes $\theta' = \theta_n + \sum_{j \in [n-1]} (-1)^{t_j} \theta_j$ and sends a correction $(b, (-1)^b\theta - \theta')$ to Bob  with $b \sample \bin$. Bob then applies $\X^b\Z((-1)^b\theta - \theta')$ to the unmeasured qubit, keeping it as output.
    \end{itemize}
  \end{algorithmic}
  \label{hdr:proto:crsp}
\end{protocol}

{
  In effect, the protocol allows Oscar to implement an RSP without having any quantum capabilities.
  \begin{theorem}[Security of Collective Remote State Preparation]
    \cref{hdr:proto:crsp} perfectly constructs the Remote State Preparation \cref{hdr:res:rsp} from Secure Classical Channel Resources between each client and the orchestrator, for malicious coalitions that include the Server and at most \(n-1\) Clients.
    \label{hdr:thm:crsp-security}
  \end{theorem}
  \section{Verification with Collective RSP}
  \label{sec:org97f457e}
  The CRSP primitive lets a classical orchestrator drive VBQC while the clients supply only random equatorial qubits and their chosen classical inputs.  We first spell out that orchestrated version.  Then we remove the orchestrator: a classical SMPC among the clients emulates its role without altering the global security bound.
  \subsection{Orchestrated SMPQC}
  \label{sec:orgca4550f}
  \begin{protocol}[Secure Multi-Party Delegated Quantum Computation with Orchestrator]
    \begin{algorithmic} [0]
      \STATE \textbf{Public Information:} 
      \begin{itemize}
      \item $G = (V, E, I, O)$, a graph with input and output vertices $I$ and $O$ respectively;
      \item $\{I_j\}_{j \in [n]}$, a partition of the input vertices, with each $I_j$ being associated to Alice number $j$.
      \item $\sch P$, a trappified scheme on graph $G$;
      \item $\preceq_G$, a partial order on the set $V$ of vertices;
      \item $N, d, w$, parameters representing the number of runs, the number of computation runs, and the number of tolerated failed tests.
      \end{itemize}  

      \STATE \textbf{Alices' Inputs:}
      \begin{itemize}
      \item Alice number $j$ has as input a classical bit-string $x_j \in \bin^{\abs{I_j}}$.
      \item The $n$ Alices collectively have as input a set of angles $\{\phi_i\}_{i \in V}$ and a flow $f$ which induces an ordering compatible with $\preceq_G$.
      \end{itemize}

      \STATE \textbf{Protocol:}
      \begin{enumerate}
      \item All Alices send their input $x_j$ to Oscar, together with the computation angles $\{\phi_i\}_{i \in V}$ and flow $f$. Let $x$ be the concatenation of all $x_j$.
      \item Oscar and Bob perform an execution of the Trappified Blind Delegated Quantum Computation (\cref{hdr:proto:tbdqc-classical-io}). Instead of having Oscar send rotated states during the UBQC execution, it performs for each state an instance of \cref{hdr:proto:crsp} implementing CRSP together with the $n$ Alices.
        \begin{enumerate}
        \item Oscar samples uniformly at random a subset $C \subset [N]$ of size $d$ representing the computation runs.
        \item For $k \in [N]$:
          \begin{enumerate}
          \item If $k \in C$, Oscar sets the computation for the run to $(\{\phi_i\}_{i \in V}, f)$ with input $x$. Otherwise, Oscar samples a test $(T, \sigma, \tau)$ from the trappified scheme $\sch P$.
          \item Oscar and Bob execute the chosen run with the UBQC \cref{hdr:proto:ubqc}. For each qubit sent during the execution of the protocol, they instead execute the Collective RSP \cref{hdr:proto:crsp} together with the $n$ Alices. Oscar samples the Collective RSP with input and angle $\theta$ each time sampled at random.
          \item If the run is a test, Oscar checks whether it passed.
          \end{enumerate}
        \item If the number of failed tests is greater than $w$, Oscar sets the output to $(\bot,\abort)$.
        \item Otherwise, let $O$ be the majority vote on the output results of the computation runs. Oscar sets the output to $(O, \accept)$.
        \end{enumerate}
      \item Oscar sends the output to all Clients.
      \end{enumerate}
    \end{algorithmic}
    \label{hdr:proto:smpqc}
  \end{protocol}
  \subsection{Removing the Orchestrator}
  \label{sec:org3cb68a9}
  In the previous protocol, the orchestrator's actions are purely classical. They can be replaced by a classical SMPC among the clients. It needs only to provide composable security to perform  coin-tossing, basic string operations (array lookup) and computations in \(\mathbb  Z_8\) and \(\mathbb Z_2\).

  The security of this fully distributed SMPQC is given by the following theorem:

  \begin{theorem}
    Suppose that the Trappified Blind Delegated Quantum Computation (\cref{hdr:proto:tbdqc-classical-io}) \(\epsilon_V\)-constructs the Secure Delegated Quantum Computation with classical-I/O (\cref{hdr:res:sdqc}) for leak \(l_\rho\). Then \cref{hdr:proto:smpqc} \(\epsilon_V\)-constructs the Quantum Secure Multi-Party Computation with classical-I/O from an interactive Classical Secure Multi-Party Computation for the same leak \(l_\rho\), against malicious coalitions that include at most the Server and \(n-1\) Clients.
  \end{theorem}

  The idea of the proof is simple. It is enough to retrace the order of compositions and use abstract cryptography to keep track of each construction error. All being at most negligible, and each resource being call at most a polynomial number of times, the security error remains negligible. 
  \section{Takeaways}
  \label{sec:orge449b50}
  We have shown that there is no fundamental difference in capabilities between the symmetric setup and the asymmetric one presented here. Both can be used to lift a composable classical SMPC to the quantum realm. To this end, we have used that one honest client already suffices to hide every phase and every trap. The heavy classical bookkeeping---who randomized what, where the traps sit, how to drive UBQC---can be concentrated in a purely classical orchestrator, and once written that way it can just as well be emulated by a standard composable SMPC.  The modular framework earns it here: we swap RSP for CRSP, the single-client for an orchestrator and then for SMPC, and the global security bound never has to be re-derived.

  The asymmetry of resources is preserved: the server does the quantum work, clients remain light, and the network need not be complete. There are, however, natural limits.  Our construction is tailored to \cmp{BQP} computations; pushing quantum data through the same pipeline would require traps that live outside a single Bloch-plane or a different masking mechanism altogether.  Architectures such as QLine \cite{PLLC23multi} suggest that even the CRSP step can be realized in hardware, by letting phases accumulate along a single fiber.

  In short, MBQC matches circuit-based approaches for \cmp{BQP} computation, and gains a practical path to implementation.  The remaining open question---collective RSP beyond one plane---are where the next advances could be.
  \chapter[Secure Delegated FTQC]{Secure Delegated Fault-Tolerant \newline Quantum Computation}
  \label{hdr:sec:sdftqc}
  \lettrine[refstring, lines=3, lraise=0.15]{T}{his chapter} tackles the question all earlier protocols politely dod\-ged. Can a verification scheme survive the avalanche of physical noise that accumulates in a large, fault-tolerant quantum computation?  Our goal is to show that it can.  We extend the modular VBQC framework to the fault-tolerant regime, and we do so without sacrificing composability nor efficiency---security error is still negligible in physical resources.

  In short, the chapter demonstrates that verifiable delegation scales: one can delegate a full fault-tolerant \cmp{BQP} computation, keep the verifier’s secrets, and retain a composable security bound negligible in the size of the logical registers. What we do not claim is that the verifier can stay limited to physical single-qubit operations---our construction grants the verifier single logical-qubit operations. Whether this can be lifted remains open.
  \section{Motivation}
  \label{sec:orgdf201f6}
  Shrinking the verifier's quantum duties was worth the effort---proof-of-concept experiments, certification, and cost all benefit from it. Yet one question still looms: do these protocols scale?

  As they stand, they do not. The best we achieved so far is robustness to global noise: if only a fixed fraction of whole rounds misfire, repetition and majority vote rescue the answer. Real devices misbehave differently. Errors strike locally, gate by gate. In a circuit with thousands of locations, even a per-gate error rate of \(10^{-3}\) guarantees that almost every round suffers a non-trivial fault. Round-level repetition is then powerless. What we need is robustness at the gate level---in other words, a way to delegate fault-tolerant quantum computations securely.

  Our route will therefore be: (i) make the imperfections explicit by modeling preparations, gates and measurements as stochastically compro\-mised---their control parameters may trickle out, and individual subsystems may be handed to the adversary; (ii) rebuild Remote State Preparation at the logical level of a concatenated code with transversal \(\Z(\theta)\) gates; (iii) split-compile each secret rotation into two random pieces and follow it by two error-correction blocks, so the effective leak probability drops quadratically at every concatenation step; and (iv) plug this fault-tolerant RSP back into the dummyless VBQC template and prove a threshold-style theorem \cite{AGP06quantum}: below a constant noise rate, the adversary's distinguishing advantage stays negligible even as the depth grows.
  \section{Side Channels Due to Error Correction}
  \label{sec:org0037182}
  In \cref{hdr:sec:rotations} we learned a lesson: to keep blindness intact, and thus verifiability, any deviation before a secret gate must be pushable to after that gate without picking up dependence on the secret.  Allowing the prover to supply the state and letting the verifier apply a hidden \(\Z(\theta)\) was fatal precisely because commuting an arbitrary error through a secret rotation imprints the secret angle \(\theta\) on the error.

  Exactly the same mechanism lurks for fault-tolerance. In the naive logical RSP construction, replace the single-qubit by a logical block: we prepare a logical \(\ket +\), then apply a logical \(\Z(\theta)\). Suppose the code is CSS and admits transversal \(\Z(\theta)\). A single physical \(\X\) error on one qubit of the block, followed by a perfect logical \(\Z(\theta)\), behaves differently depending on \(\theta\): if \(\theta = 0\), the error stays \(\X\); if \(\theta= \pi/2\), it turns into \(\Y\). The syndrome then changes character: from purely \(\Z\)-type it becomes mixed \(\X/\Z\)-type. Should the subsequent error correction miss that fault, the prover upon receiving the prepared faulty logical \(\ket{+_{\theta}}\) can, by reading out the syndrome, glean information about \(\theta\).

  One "if" too many?  Security arguments are exactly about preventing long chains of "ifs" from aligning.

  This example calls for a precise model of the imperfections that might arise on the verifier's side. This model must be compatible with standard fault-tolerance yet reflect delegation. To this end, we switch terminology for describing the model, where we have a user trying to perform quantum operations and an eavesdropper that plays the role of the adversary interfering with the computation.  In textbook fault-tolerance \cite{AGP06quantum}, the eavesdropper is typically granted the full circuit description even before fault locations are fixed.  In a delegated setting, this would instantly kill secrecy. To avoid that, we distinguish between public parameter domains and the chosen value.  Preparations, gates, and measurements are parameterized: the set \(\Lambda\) is public, but the actual \(\lambda \in \Lambda\) leaks only stochastically, while the eavesdropper can then choose an error to apply on the corresponding physical qubit. This is captured by:
  \begin{resource}[Stochastically Compromised Preparations, Gates and Measurements]
    A \emph{Stochastically Compromised Preparation, Gate or Measurement} with compromise probability \(p_{\mathrm{c}}\) is an ideal resource that behaves in the following way, where Alice plays the user and Bob the eavesdropper:
    \begin{algorithmic}[0]

      \STATE \textbf{Public Information:} 
      \begin{itemize}
      \item The compromise probability $p_c$;
      \item The size of the input and output registers $n$;
      \item The set of parameters that can be used to classically control the operation $\Lambda$, reduced to the name of the preparation, gate or measurement if not parameterized;
      \item The set of CPTP maps $\{\cptp U(\lambda)\}_{\lambda \in \Lambda}$.
      \end{itemize}

      \STATE \textbf{Alice's Input:}
      \begin{itemize}
      \item A value $\lambda \in \Lambda$ that parameterizes the preparation, gate or measurement;
      \item The quantum registers that are acted upon by the Resource.
      \end{itemize}	    

      \STATE \textbf{Bob's Input:} Bits $c_i \in \{0, 1\}$ indicating whether it wants to cheat, set to $0$ in the honest case and to $1$ when it wants to compromise the $i$-{th} location of the preparation, gate or measurement.

      \STATE \textbf{Computation by the Resource:}
      \begin{itemize}
      \item If the input register is non-empty, it inserts the state $\rho$ of the provided register into its internal register. Else it initializes its input register in the $\ket{0}^{\otimes n}$ state.
      \item It applies $\cptp U(\lambda)$ to $\rho$.
      \item For each $i$ s.t. $c_{i} =1$, with probability $p_c$:
        \begin{itemize}
        \item It sends $\lambda$ and the register at location $i$ to Bob;
        \item It receives from Bob a quantum system and inserts into its internal register at position $i$.
        \end{itemize}
      \item It outputs the state in its internal register at Alice's output interface.
      \end{itemize}
    \end{algorithmic}
    \label{hdr:res:compromised-operation}
  \end{resource}

  This resource faithfully reproduces the informal problem above. Implement the naive logical RSP with compromised \(\Z(\theta)\) gates, and every physical site becomes a potential leak of \(\theta\). As the block size grows, the probability that some site spills the secret tends to one. We therefore need an RSP construction that explicitly withstands stochastic compromise events. In the next section we show how concatenation and angle splitting drive the effective compromise rate down to negligible as the size of logical registers increases.
  \section{Secure Fault-Tolerant RSP via Split-Compilation}
  \label{sec:org9c1da07}
  \subsection{Conditions for Privacy and Gadget Construction}
  \label{sec:org0fa2dbb}
  The first constraint for our design is that the preparation has to be fault-tolerant. This does not necessarily imposes that everything on the verifier's side is fault-tolerant---the prover could possibly assist with multi-qubit gates and error correction---but we were not able to maintain security in this scenario. We thus resorted to grant the verifier the ability to prepare and operate on a single logical-qubit by itself. It would then follow the naive implementation proposed above: prepare a logical \(\ket +\) and apply a logical \(\Z(\theta)\). To keep the circuit uniform over all \(\theta\), we pick a code with a transversal \(\Z(\theta)\): the \([[15,1,3]]\) quantum Reed–Muller code. This avoids magic-state injection---whose circuitry would depend on \(\theta\) and could thus introduce side-channels---and keeps the proof tractable.

  The second constraint is that compromise probability must not scale with the block size. As we have seen, if each physical \(\Z(\theta)\) is compromised with probability \(p_{\mathrm{c}}\), a transversal gate leaks almost surely once the block is large enough. We fix this by ensuring that the effective probability of compromise events is squared for each added concatenation level. The trick is to apply \emph{split-compilation}: to add one level of concatenation, each physical qubit is transformed into a code-block of 15 qubits; the \(Z(\theta)\) is applied on each wire \(i\) by sampling \(\alpha_{i} \sample \Theta\), setting \(\beta_i = \theta - \alpha_i\), and applying \(\Z(\alpha_i)\) followed by \(\Z(\beta_i)\). Intuitively, both must be compromised for the adversary to reconstruct \(\theta\).

  Yet, this intuition is only partly correct. A single incoming \(\X\) error can still drag \(\theta\) into the syndrome if the error-correction block is itself compromised. The fix is simple: after each logical gate we run two EC blocks. If one is compromised, the other resets the syndrome to the all-zero string and forgets about \(\theta\).

  This yields the level-\(1\) split-compilation for \(\Z(\theta)\)---below \(0\)-Ga is synonym for physical gate, \(1\)-Ga are gates applied to logical qubits encoded using one level of concatenation:
  \begin{definition}[\(1\)-Safe$\Z(\theta)$]
    Using the \([[15,1,3]]\) quantum Reed Muller code, the \(1-\mathsf{Safe}\Z(\theta)\) gate that implements the corresponding \(1\)-Ga for the single-qubit \(\Z(\theta)\) gate is constructed as follows:
    \begin{algorithmic}[0]
      \STATE \textbf{For $i \in [1,15]$}:
      \begin{itemize}
      \item Sample $\alpha_i \sample \Theta$.
      \item Set $\beta_i = \theta - \alpha_i$.
      \item Apply in succession the two physical gates corresponding to $\Z(\beta_i) \circ \Z(\alpha_i)$ to the $i$\textsuperscript{th} qubit, via calls to the Stochastically Compromised Gates Resource.
      \end{itemize}
    \end{algorithmic}
  \end{definition}

  This gives the following procedure for the recursive simulation:
  \begin{definition}
    A \emph{\(1\)-SafeRec} for simulating a \(0\)-Ga consists of the \(1\)-Ga followed by two consecutive \(1\)-EC for each output block of the \(1\)-Ga.
  \end{definition}
  Here, \(1\)-EC is the level-1 error-correction routine that corrects a single logical-qubit encoded using one level of concatenation.
  \marginnote{\\ For avoiding confusion, we use the terminology of accurate and private so that it is distinct from the correct and secure ones that we reserve for the constructed RSP. Indeed, here, we will improve the accuracy and privacy by concatenation, up to a level where it will be good enough to provide correctness and security of the corresponding RSP. These notions of accuracy and privacy need not compose as they are only used in the context of the verifier-side RSP, while correctness and security do compose at the RSP level.}
  
  \begin{definition}[\(k\)-SafeRec]
    A \(k\)-SafeRec is obtained from a \((k-1)\)-SafeRec by replacing each \(0\)-Ga by its corresponding \(1\)-SafeRec.
  \end{definition}
  \subsection{Security of Fault-Tolerant RSP}
  \label{sec:org049cd7c}
  While the robustness to noise of the fault-tolerant RSP follows the usual threshold machinery, privacy does not come for free. We must rework the argument so that, after simulation, (i) the logical action is \emph{accurate} and (ii) \emph{private} in the sense that the entire syndrome history is independent of the secret angles \(\theta\). The main result is:
  \begin{theorem}[Threshold theorem for accurate and private computations]
    The diamond distance \(\delta\) between a computation with \(L\) locations simulated using \(k\)  levels of concatenation and an ideal computation performing the correct computation on the logical space and producing a syndrome that is independent of the parameters \(\theta\) of the gates satisfies 
    \begin{equation}
      \delta \leq 2L \times p_{\mathrm{0}}(p_{\mathrm{c}}/p_{\mathrm{0}})^{2^k} \label{hdr:eq:threshold},
    \end{equation}
    where \(p_{\mathrm{c}}\) is the compromise event probability of \cref{hdr:res:compromised-operation}.
  \end{theorem}

  The proof mirrors the classical threshold proof \cite{AGP06quantum}, but each step now has a security twin:
  \begin{itemize}
  \item \emph{Level-0 to level-1 improvement.} Show that accuracy and privacy improve when moving from level-\(0\) to level-\(1\), for sufficiently low compromise probability.
  \item \emph{Extended Safe Rectangles.} As in \cite{AGP06quantum}, we enlarge rectangles---a \(k\)-Ga followed by the two \$\(k\)-ECs---to include one preceding \(k\)-EC. Here this is needed not only for accuracy---to catch incoming errors---but also to take into account the possibility of inducing \(\theta\)-dependence by commutation.
  \item \emph{Good vs. bad recursion.} We define conditions where simulations would be either inaccurate or non-private. This notion of goodness and badness is essential to perform the recursion and extend the quadratic decrease of effective compromise probability from level-\(0\) to level-\(1\).
  \item \emph{Bad rectangles as compromised gates.} The technical part is again ensuring that a bad level-\(k\) rectangle can be represented as a simulated compromised gate.
  \end{itemize}

  A remark in \cite{AGP06quantum} is instructive: "the syndrome becomes an effective environment that [\ldots] interacts with the data during 0-faults. In their setting this is harmless because errors can be arbitrarily correlated; in ours it becomes a side channel affecting the secrecy of the computation. The second EC block precisely severs this channel: unless both blocks are compromised, one of them resets the syndrome to all zeros and erases any trace of \(\theta\) in the syndrome.

  Armed with this theorem, we construct a level-\(k\) simulation of RSP and provide its security guarantee with Alice playing the role of the user and Bob that of the eavesdropper:
  \begin{protocol}[Level-$k$ Simulation of Remote State Preparation]
    \begin{algorithmic}[0]
      \STATE \textbf{Alice:} Choose $\theta \in \Theta$.

      \STATE \textbf{Alice:} Implements the level-$k$ preparation of $\ket +$.

      \STATE \textbf{Alice:} Implements the level-$k$ $Z(\theta)$ using the Safe$Z(\theta)$ construction.

      \STATE \textbf{Bob:} For each location in the circuit decide to cheat, i.e. stochastically compromise it, and if succeeding gets the leaks and provides the replacement for the leaked registers.

      \STATE \textbf{Alice:} Sends the produced quantum state.
    \end{algorithmic}
    \label{hdr:proto:ft-rsp}
  \end{protocol}

  \begin{theorem}
    \cref{hdr:proto:ft-rsp} constructs the Level-\(k\) Single Plane Remote State Preparation---i.e. the Single-Plane RSP (\cref{hdr:res:rsp}) applied on level-\(k\) logical qubits with error bounded above by \(4 p_{\mathrm{0}}(p_{\mathrm{c}}/p_{\mathrm{0}})^{2^k}\).
  \end{theorem}
  \section{Verification of Fault-Tolerant Quantum Computations}
  \label{sec:org5dc7041}
  The full protocol is obtained by taking the dummyless verification template and replacing every call to RSP by its level-\(k\) fault-tolerant version. Gates and measurements on the prover’s side are likewise lifted to their level-\(k\) simulations.

  \begin{protocol}[Secure Delegation of Fault-Tolerance Quantum Computation]
    \begin{algorithmic}[0]
      \STATE \textbf{Public Information:} 
      \begin{itemize}
      \item $G = (V, E, I, O)$, a graph with input and output vertices $I$ and $O$ respectively.
      \item $\mathfrak{W}$, the set of dummyless tests on graph $G$.
      \item $\preceq_G$, a partial order on the set $V$ of vertices.
      \item $N, d, w$, parameters representing the number of runs, the number of computation runs, and the number of tolerated failed tests.
      \item $k$, the concatenation level used to encode logical information.
      \end{itemize}
      
      \STATE \textbf{Alice's Inputs:} A set of angles $\{\phi(i)\}_{i \in V}$ and a flow $f$ which induces an ordering compatible with $\preceq_G$.

      \STATE \textbf{Protocol:}
      Alice and Bob execute \cref{hdr:proto:tbdqc-classical-io} with dummyless traps for the intended pattern on graph $G=(V,E)$ in the following way:
      \begin{itemize}
      \item For each call to the RSP, they run \cref{hdr:proto:ft-rsp}.
      \item For each gate or measurement instructed by Alice, Bob implements the corresponding level-$k$ simulation, considering for measurements that the outcome is that for the measurement of the corresponding logical qubit.
      \end{itemize}
    \end{algorithmic}
    \label{hdr:proto:ft-sdqc}
  \end{protocol}

  The overhead and the security bound follow from simple counting plus composability.
  \begin{theorem}[Threshold Theorem for Secure Delegation of Fault-Tolerant Quantum Computations]
    Let \(d\) be proportional to \(N\), and let \(c\) the fixed bounded error of the \(\BQP\) class of computations.  Let \(\epsilon\) be the (constant) lower bound on the probability that a test sampled uniformly from set \(\mathfrak{W}\) fails in the presence of a \(\Z\) error.  Let \(w\) be the maximum number of test rounds allowed to fail, chosen such that \(w < \frac{2c-1}{2c-2}(N-d)(1 - \epsilon)\).  Let the leak be defined as \(l_{\cptp U} = (G, \mathfrak{W}, \preceq_G)\), and let the Verifier use a register of size \(15^{k}\) corresponding to \(k\) level of concatenation of the quantum [15,1,3] Reed-Mueller code with \(15^k \in O(Nlog(|V|))\).   Let \(p_c\) be the probability that a location is compromised, and \(p_0\) be the threshold probability from \cref{hdr:eq:threshold}.

    Then, \cref{hdr:proto:ft-sdqc} \(\eta(N)\) constructs the Secure Delegated Quantum Computation (\cref{hdr:res:sdqc}) in the Abstract Cryptography framework for \(\eta(N)\) negligible in \(N\).
  \end{theorem}
  \section{Takeaways}
  \label{sec:orgd15c569}
  We have shown how to remove the main scalability roadblock, but one sharp question survives: can the verifier drop back to single physical-qubit operations and still stay secure? Yet, this looks largely a theoretical question. Without its own error correction, a verifier must lean on the prover for every correction step, forcing interactive, back-and-forth rounds between logical gates. That chatter would explode the latency and sink performance.

  In short, we now know how to scale securely—but doing so with a truly qubit-only verifier remains an open, and probably impractical, challenge.
  \part{Conclusions \& Perspectives}
  \label{hdr:sec:perspectives}
  \cleardoublepage
  \lettrine[refstring, lines=3, lraise=0.15]{T}{his work}, begun in 2019 and developed throughout this manu\-script, turns verification from a single monolithic proof into a stack of interoperable modules. By cleanly separating blindness, trappification, embedding, and fault-tolerance, we could patch specific shortcomings of earlier protocols and push the field in two complementary directions:
  \begin{itemize}
  \item Lighter clients \& lower overhead for servers.
  \item Richer functionality, from multi-party inputs to gate-level noise robustness.
  \end{itemize}

  The broader aim is practical: bring verification and blindness into the design brief of hardware vendors, cloud providers, and end-users---before devices reach the too big to secure-by-design regime. Showing that the same template handles secure multi-party quantum computation and tolerates physical noise at the gate level is a first step toward that integration.

  But the road is long, and several milestones lie ahead.
  \section*{Perspectives}
  \label{sec:org9f6b9d4}
  Verification is still viewed by many hardware vendors as an after-thought, something to bolt on once the hardware works. Experience with classical security---and with quantum error correction---suggests the opposite: early design choices lock in what becomes possible later. In particular, adaptivity and the central role of a secure RSP must be recognized up-front. Quantum analogues of trusted-execution environments could fill the gap where high-bandwidth quantum links are impractical, but such isolation shapes the rest of the QPU stack.

  Our view is that vendors will embrace verification once it offers clear, immediate value to them. Two levers can help.
  \begin{itemize}
  \item \emph{Benchmark-driven competition.}  Trap-based tests give a worst-case to average-case reduction, linear-time scaling, and negligible security error per extra resource invested. Turning these tests into public benchmarks should spark vendor competition and push verification features into the control stack.
  \item \emph{In-field characterisation.} Trap outcomes double as a fine-grained log of faults.  A server that knows it is honest can mine those logs to calibrate and stabilize its device without downtime.  Early inclusion of verification thus pays off even before quantum advantage is reached.
  \end{itemize}
  \section*{Action plan}
  \label{sec:org4d96662}
  The first priority is to translate every module of the framework into the circuit model, so that superconducting platforms can adopt the techniques without resorting to MBQC workarounds. In parallel, we should refine the mathematical analysis until we can obtain hardware efficient proofs. Finally, the initial trap-based benchmarks must be carried through to public pilot reports, turning them into yardsticks that vendors cannot ignore.

  Next, we will explore how far quantum communication can be reduced by exposing a tightly quantified amount of information about the delegated algorithm, or by leaning on standard computational assumptions. At the same time, we need to lift robustness from majority-vote classical outputs to genuine quantum outputs encoded in error-correcting codes. Alongside those technical goals, we will package the protocol suite into an open-source library that hides cryptographic plumbing from circuit designers yet remains auditable by specialists.
  \section*{Long-term questions}
  \label{sec:org4a63c07}
  Over the longer horizon, we should seek a resource-theoretic account of verification that reveals which primitives are truly fundamental. We must also confront the long-standing challenge of achieving statistical soundness with a fully classical client—no qubits, no measurement device—while keeping complexity reasonable. Finally, we ought to clarify what it really means to "verify a quantum output" once the output itself is an entangled state destined for further quantum processing, and to devise tests that respect that richer notion of correctness.
  \part{Appendix}
  \label{sec:orgd49c6e7}
  \appendix  
  \chapter{Abstract Cryptography}
  \label{hdr:sec:ac}
  The Abstract Cryptography (AC) security framework \cite{MR11abstract-cryptography,M12constructive-cryptography} used in this work follows the \emph{ideal/real simulation paradigm}. A protocol is considered secure if it is a good approximation of an ideal version called a \emph{resource}. Its main advantage is that any system that follows the structure defined by the framework is inherently composable, in the sense that if two protocols are secure separately, the framework guarantees at an abstract level that their sequential or parallel execution is also secure.  We refer the reader to \cite{DFPR14composable} for a more in-depth presentation, in particular regarding the framework's composability in the context of SDQC.

  In this framework, the purpose of a secure protocol \(\pi\) is, given a number of available resources \(\mathcal{R}\), to construct a new resource -- written as \(\pi \mathcal{R}\).  This new resource can be itself reused in a future protocol.  A resource \(\mathcal{R}\) is described as a sequence of CPTP maps with an internal state.  It has \emph{input and output interfaces} describing which party may exchange states with it.  It works by having each party send it a state (quantum or classical) at one of its input interfaces, applying the specified CPTP map after all input interfaces have been initialized and then outputting the resulting state at its output interfaces in a specified order. An interface is said to be \emph{filtered} if it is only accessible by a dishonest player. The actions of an honest player \(i\) in a given protocol is also represented as a sequence of efficient CPTP maps \(\pi_i\) -- called the \emph{converter} of party \(i\) -- acting on their internal and communication registers. We focus here on the two-party setting, in which case \(\pi = (\pi_1, \pi_2)\).

  In order to define the security of a protocol, we need to give a pseudo-metric on the space of resources.  We consider for that purpose a special type of converter called a \emph{distinguisher}, whose aim is to discriminate between two resources \(\mathcal{R}_1\) and \(\mathcal{R}_2\), each having the same number of input and output interfaces.  It prepares the input, interacts with one of the resources according to its own (possibly adaptive) strategy, and guesses which resource it interacted with by outputting a single bit.  Two resources are said to be indistinguishable if no distinguisher can guess correctly with good probability.

  \begin{definition}[Statistical Indistinguishability of Resources]
    Let \(\epsilon > 0\), and let \(\mathcal{R}_1\) and \(\mathcal{R}_2\) be two resources with same input and output interfaces.  The resources are \emph{\(\epsilon\)-statistically-indistinguishable} if, for all unbounded distinguishers \(\mathcal{D}\), we have:
    \begin{equation}
      \label{eq:dist}
      \Bigl\lvert\Pr[b = 1 \mid b \leftarrow \mathcal{D}\mathcal{R}_1] - \Pr[b = 1 \mid b \leftarrow \mathcal{D}\mathcal{R}_2]\Bigr\rvert \leq \epsilon.
    \end{equation}
    We then write \(\mathcal{R}_1 \underset{\epsilon}{\approx} \mathcal{R}_2\).
    \label{def:ind-res}
  \end{definition}

  The construction of a given resource \(\mathcal{S}\) by the application of protocol \(\pi\) to resource \(\mathcal{R}\) can then be expressed as the indistinguishability between resources \(\mathcal{S}\) and \(\pi \mathcal{R}\).  More specifically, this captures the correctness of the protocol.  The security is captured by the fact that the resources remain indistinguishable if we allow some parties to deviate in the sense that they are no longer forced to use the converters defined in the protocol but can use any other CPTP maps instead.  This is done by removing the converters for those parties in Equation \ref{eq:dist} while keeping only \(\pi_H = \prod_{i \in H} \pi_i\) where \(H\) is the set of honest parties.  The security is formalized as follows in Definition \ref{def:ac-sec} in the case of two parties.

  \begin{definition}[Construction of Resources]
    Let \(\epsilon > 0\). We say that a two-party protocol \(\pi\) \(\epsilon\)-statistically-constructs resource \(\mathcal{S}\) from resource \(\mathcal{R}\) if:
    \begin{enumerate}
    \item It is correct: $\pi \mathcal{R} \underset{\epsilon}{\approx} \mathcal{S}$.
    \item It is secure against malicious party $P_i$ for $i \in \{1, 2\}$: there exists a \emph{simulator} (converter) $\sigma_i$ such that $\pi_j\mathcal{R} \underset{\epsilon}{\approx} \mathcal{S} \sigma_i$, where $j \neq i$.
    \end{enumerate}
    \label{def:ac-sec}
  \end{definition}

  The General Composition Theorem (Theorem 1 from \cite{MR11abstract-cryptography}) captures the security of protocols which use other secure protocols as subroutines, in sequence or in parallel.
  \begin{theorem}[General Composability of Resources]
    Let \(\mathcal{R}\), \(\mathcal{S}\) and \(\mathcal{T}\) be resources, \(\alpha\), \(\beta\) and \(\mathsf{id}\) protocols (where protocol \(\mathsf{id}\) does not modify the resource it is applied to). Let \(\circ\) and \(\mid\) denote respectively the sequential and parallel composition of protocols and resources. Then:
    \begin{itemize}
    \item The protocols are \emph{sequentially composable}:
      \begin{equation}
        \alpha \mathcal{R} \!\!\!\underset{\mathrm{stat},\ \epsilon_{\alpha}}{\approx}\!\!\! \mathcal{S}
        \ \wedge \
        \beta \mathcal{S} \!\!\!\underset{\mathrm{stat},\ \epsilon_{\beta}}{\approx}\!\!\! \mathcal{T}
        \ \Rightarrow \
        (\beta \circ \alpha) \mathcal{R} \!\!\!\underset{\mathrm{stat},\ \epsilon_{\alpha} + \epsilon_{\beta}}{\approx}\!\!\! \mathcal{T}.
      \end{equation}
    \item The protocols are \emph{context-insensitive}:
      \begin{equation}
        \alpha \mathcal{R} \!\!\!\underset{\mathrm{stat},\ \epsilon_{\alpha}}{\approx}\!\!\! \mathcal{S}
        \ \Rightarrow \
        (\alpha \mid \mathsf{id}) (\mathcal{R} \mid \mathcal{T}) \!\!\!\underset{\mathrm{stat},\ \epsilon_{\alpha}}{\approx}\!\!\! (\mathcal{S} \mid \mathcal{T}).
      \end{equation}
    \end{itemize}
    \label{thm:comp-res}
  \end{theorem}

  Combining the two properties presented above yields concurrent composability (the distinguishing advantage is additive as well).

  \bibliographystyle{alpha}
  \bibliography{../qubib/qubib}
\end{document}